\documentclass{article}
\usepackage{authblk}
\usepackage{amsmath,amssymb}

\usepackage{amsthm} 

\usepackage{hyperref}
\usepackage{amsfonts}

\usepackage{verbatim}
\usepackage{enumitem}
\usepackage{stmaryrd}
\usepackage{fancyhdr}
\usepackage{graphicx}
\usepackage{xcolor}
\usepackage[left=1.2in, right=1.2in, bottom=1.2in,top=1.2in]{geometry}

\newtheorem{theorem}{Theorem}[section]   
\newtheorem{lemma}[theorem]{Lemma}
\newtheorem{proposition}[theorem]{Proposition}

\theoremstyle{definition}

\numberwithin{equation}{section}
\numberwithin{theorem}{section}
\newtheorem{note}[theorem]{Note}

\newcommand{\rme}{\mathrm{e}}
\newcommand{\rmi}{\mathrm{i}}
\newcommand{\rmd}{\mathrm{d}}

\newcommand{\eps}{\varepsilon}

\newcommand{\lav}{M_f^{loc}}
\newcommand{\dif}{\mathrm{d}}

\newcommand{\cL}{\mathcal{L}_{gl}}
\newcommand{\ccL}{\mathcal{L}_{loc}}
\newcommand{\R}{\mathbb{R}}
\newcommand{\gk}{\mathcal{G}_k}
\newcommand{\Tt}{\mathbb{T}}

\newcommand{\be}{\begin{equation}}
\newcommand{\ee}{\end{equation}}

\newcommand{\lv}{\left\vert}
\newcommand{\rv}{\right\vert}
\newcommand{\pa}{\partial}

\newcommand{\app}{\zeta}
\newcommand{\Fx}{F_{\xi}}

\newcommand{\rel}{\Re(\lambda)}
\newcommand{\iml}{\mathrm{Im}(\lambda)}
\newcommand{\dk}{d_k}
\newcommand{\ck}{c_k}
\newcommand{\rk}{r_k}

\graphicspath{ {./img/} }
\usepackage{subcaption}
\usepackage{bbm}
\DeclareMathOperator{\atan}{atan}
\usepackage[citestyle=numeric-comp, bibstyle=numeric, abbreviate=true, giveninits=true, maxnames=10, doi=false, isbn=false, url=false, eprint=true, date=year]{biblatex}

\title{Non-mean-field Vicsek-type models for collective behaviour}
\author[1]{ P. Butt\`a}
\author[2]{ B. Goddard}
\author[3]{ T. M. Hodgson} 
\author[4]{ M. Ottobre}
\author[5]{ K. J. Painter}
\affil[1]{Dipartimento di Matematica, SAPIENZA Universit\`a di Roma, P.le Aldo Moro 5, 00185 Rome, Italy,  butta@mat.uniroma1.it}
\affil[2]{School of Mathematics and Maxwell Institute for Mathematical Sciences, University of
Edinburgh, Edinburgh EH9 3FD, UK. bgoddard@ed.ac.uk}
\affil[3]{Maxwell Institute for Mathematical Sciences and Mathematics Department, Heriot-Watt University, Edinburgh EH14 4AS, UK. tmh2@hw.ac.uk}
\affil[4]{Maxwell Institute for Mathematical Sciences and Mathematics Department, Heriot-Watt University, Edinburgh EH14 4AS, UK. m.ottobre@hw.ac.uk}
\affil[5]{DIST, Politecnico di Torino, Torino, Italy, Viale Pier Andrea Mattioli, 39, 10125.  kevin.painter@polito.it}

\date{}
\begin{document}

\maketitle

\begin{abstract}
	We consider interacting particle dynamics with Vicsek type interactions, and their macroscopic PDE limit,  in the non-mean-field regime; that is, we consider the case in which each  particle/agent in the system interacts only with a prescribed subset of the particles in the system (for example, those within a certain distance). In this non-mean-field regime the  influence between agents (i.e. the interaction term) can be normalised either by the total number of agents in the system (\textit{global scaling}) or by the number of agents with which the particle is effectively interacting  (\textit{local scaling}). We compare the behaviour of the globally scaled and the locally scaled systems in many respects, considering for each scaling both the PDE and the corresponding particle model.  In particular we observe that both the locally and globally scaled particle system exhibit pattern formation (i.e. formation of travelling-wave-like solutions) within certain parameter regimes, and generally display similar dynamics. The same is not true of the corresponding PDE models. Indeed,   while both PDE models have multiple stationary states, for the globally scaled PDE such (space-homogeneous) equilibria are unstable for certain parameter regimes, with the instability leading to travelling wave solutions, while they are always stable for the locally scaled one, which never produces travelling waves. This observation is based on a careful numerical study of the model, supported by further analysis.
	\smallskip

\end{abstract}

\section{Introduction}
The study of Interacting Particle Systems (IPSs) and related kinetic equations  has attracted the interest of the mathematics and physics communities for decades. Such interest  is kept alive by the continuous successes of this framework in modelling of a vast range of phenomena, in diverse fields such as biology, social sciences, control engineering, economics, statistical sampling and simulation \cite{pareschiInteractingMultiagentSystems2013, bianchiDynamicsBiologicalSystems2019}. Recent applications to e.g. neural networks \cite{sirignanoMeanFieldAnalysis2020} are only some of the examples of how the study of IPSs is proving able to be exploited in a growing number of research areas.   While such a large body of research has undoubtedly produced significant progress over the years, many important questions in this field remain open.

One is often not interested in a detailed description of the IPS itself, but rather in its collective behaviour. The standard methodology in statistical mechanics is to look for simplified equations, obtained by letting the number $N$ of particles in the system tend to infinity; in the limit one finds a  Partial Differential Equation (PDE) for the density of particles.\footnote{Depending on the nature of the IPS at hand, the macroscopic limit might of course lead to a Stochastic PDE, but we will not consider such models in this paper.}  Hence, when modelling IPSs, one has typically two levels of description: one given by the original IPS, also commonly referred to as the {agent-based} model, and the other given by the macroscopic PDE model. In this context, particle systems with mean-field interactions, i.e. `all-to-all' type interactions,  have been extensively studied and the standard paradigm refers in fact to this case (though even in the mean-field case several technical challenges are still to be solved, with some long standing ones having been only recently brought to a close \cite{breschMeanfieldLimitQuantitative2020}).
A  question which naturally arises in applications concerns the case in which each particle interacts only with certain other particles in the system, for example with those within a given range of interaction. This non-mean-field regime\footnote{In this paper we say that the interaction between particles is of mean-field type if every particle interacts with every other particle, so that, in the limit, every particle is interacting with infinitely many other particles. We say that the interaction is non-mean-field when each particle interacts only with a subset of the particles in the system and we emphasize that, in the limit,  each agent might still be interacting with infinitely many other agents. A better name for this regime would be `local mean field'; we do not use it only to avoid confusion with the local/global scaling nomenclature which we will introduce. } is the one on which we focus in this paper. In different scenarios from the one we consider here the study of the non-mean-field regime has also been tackled in \cite{barreFastNonmeanfieldNetworks2021, bayraktarStationarityUniformTime2022, motschNewModelSelforganized2011,krauseDiscreteNonlinearNonautonomous2000, canutoEulerianApproachAnalysis2012, goddardNoisyBoundedConfidence2022}, without any claim to completeness of references.

Let us explain the questions  tackled in this paper for a toy model (we defer the full presentation of the models that we will consider to Section \ref{sec:model}); to this end, consider a standard  model for $N$ diffusing particles,   of the form
$$
	\dif X_t^{i,N} = \frac{1}{N}\sum_{j=1}^N K(|X_t^{i,N}-X_t^{j,N}|) \dif t + \sqrt{2\sigma} \,\dif B_t^i \, ,
$$
where $X_t^{i,N} \in \R^d$ denotes the position of the $i$-th particle (or agent) at time $t$, $|\cdot|$ is the Euclidean norm,   $K:\R_+ \rightarrow \R^d$ is the so-called interaction kernel (here taken to depend on the distance between particles just to fix ideas, but could be more general, and indeed this fact is not enforced in our model), the $B^i_t$'s are $d$-dimensional independent standard Brownian motions and $\sigma>0$ is the diffusion constant. If $K$ is non-zero (almost) everywhere (and,  say,  bounded) then every particle  interacts with every other particle and we are in the mean-field case;  if  interested in the many particles regime, we would let $N\rightarrow \infty$ and the pre-factor $1/N$ is the natural and correct correct scaling needed to ensure a well posed limit exists,  otherwise the first term on the RHS may simply explode \cite{motschNewModelSelforganized2011}. However, when the interaction function is not supported on the whole $\R_+$  -- for example, if it is of the form $\mathbf{1}_{[0,R]}(|x|)K(|x|)$ where $\mathbf{1}_{[0,R]}$ is the characteristic function of the interval $[0,R]$ -- then each particle only interacts with other particles  within distance $R$. In this case it is no longer obvious how the interaction term should be normalised: one could either still normalise by $1/N$ ({\em global scaling}) or instead normalise by the number of particles contained in the  ball of radius $R$ and centered at $X^{i,N}_t$, i.e. consider the following model
\begin{equation}\label{intro:re}
	\dif X_t^{i,N} = \frac{1}{\# \{j: |X^{i,N}_t-X^{i,N}_t|\leq R\}}\sum_{j=1}^N \mathbf{1}_{[0,R]}(|X^{i,N}_t- X^{j,N}_t|){K}(|X^{i,N}_t- X^{j,N}_t|) \dif t + \sqrt{2\sigma} \dif B_t^i\,.
\end{equation}
This time the interaction is scaled by the number of particles with which particle $i$ is effectively interacting  at time $t$ ({\em local scaling}).
From a strictly mathematical point of view, under appropriate assumptions on $K$ both normalizations make sense and one can obtain a well-posed limiting PDE either way (though in general the rigorous justification of the limit in the local scaling case is a lot harder \cite{debusscheExistenceMartingaleSolutions2020}). However, as noted in \cite{motschNewModelSelforganized2011} in the context of Cucker-Smale type models, from a modelling perspective the global scaling has the drawback that the motion of every agent is influenced by the total number $N$ of agents even when, in principle, each agent should only be influenced by a subset of the other particles (in our running example those within distance $R$).  In   \cite[Section 2.1]{motschNewModelSelforganized2011} the authors give a nice heuristic explanation of this fact, on which we  will elaborate  in Note \ref{note:heuristics}.  Simplistically,  a starling in Edinburgh would be influenced by a starling in Rome, when it should only be influenced by its fellow flock members.

In models with the local scaling many questions become generally quite difficult to answer analytically,   and this is mostly due to the fact that the scaling factor, i.e. the denominator in front of the sum on the RHS of \eqref{intro:re} is time-dependent and, in general, complicated to keep track of analytically; we will return to this point for the specific model that we will consider, but we point out for now that this is usually circumvented by considering interaction functions with fast decay rather than with compact support (see e.g. \cite{motschNewModelSelforganized2011}).

In this paper we consider an interacting particle system of Vicsek-type \cite{vicsekNovelTypePhase1995}, under both global and local scaling; moreover, for each scaling, we consider the corresponding PDE dynamics as well. We use the acronyms LS IPS/GS IPS and LS PDE/GS PDE for Locally Scaled/Globally Scaled Interacting Particle system and PDE, respectively.
The GS IPS  we consider here has already been studied in \cite{garnierMeanFieldModel2019}; there the authors simulated the GS IPS and found that, in certain parameter regimes, the GS IPS exhibits pattern formation, i.e. the particles organise in  non space-homogeneous configurations (in particular they observed the formation of travelling bands). Linear stability analysis on the macroscopic GS PDE model yielded  sufficient conditions under which the steady states of the PDE -- which are homogeneous in space -- are stable. Such sufficient conditions were shown in \cite{garnierMeanFieldModel2019} to be compatible with simulations of the GS IPS, in the sense that  the GS IPS generates  travelling bands in the parameter regime in which the sufficient condition for the stability of the (space-homogeneous) steady states of the PDE is violated. Summarising,  the results of \cite{garnierMeanFieldModel2019} imply that the formation of travelling bands in the GS IPS is due to an instability of the  steady states (of the PDE) --  which are, we iterate, space-homogeneous.

In this paper we don't repeat the simulations of the GS IPS and  for such simulations we rely on \cite{garnierMeanFieldModel2019}. However we do simulate both the LS IPS and the GS and LS PDEs. Let us emphasise that the macroscopic models we consider are kinetic PDEs; they are not in gradient form, non-elliptic and non-linear, and, due to the nature of the non-linearity and to the lack of gradient flow structure, the analytical study of the long-time behaviour of such models (the LS PDE in particular, see Section \ref{sec:background and literature}) is known to be difficult \cite{benedettoKineticEquationGranular1997}.
So, to avoid simplifying the nature of the interaction function we first rely, in the first part of this paper, on a careful numerical study of the mentioned PDEs and IPSs. The simulation of the GS and LS PDEs is itself a non-trivial task (because of all the difficulties mentioned above), demanding state of the art pseudospectral methods (comments on the choice of numerical scheme can be found in Section \ref{sec:preliminaries}).
The results of our simulations show that while the behaviour of the LS/GS particle systems is deceptively similar, at least in the sense that both IPSs seem to not homogenise in space in certain parameter regimes, the behaviour of the LS PDE is qualitatively quite different from the behaviour of the GS PDE. In particular, within the Vicsek-type models that we consider, both PDEs exhibit multiple stationary states (and indeed the same set of stationary states, modulo some caveats, see Section \ref{sec:model}) but, while in the GS PDE such stationary states are unstable for certain parameter regimes, with the instability leading to formation of travelling wave solutions,  in the LS PDE model the stationary states are always observed to be stable, irrespective of the choice of parameter regime. In other words, while both the LS and the GS IPS exhibit pattern formation as a result of the instability of the space-homogeneous distribution, such instability is retained (in the limit $N\rightarrow\infty$) by the GS PDE but is lost in the LS PDE (see Figure \ref{fig:kde_plots} for a visual recap). This is coherent with the generally  known fact that one should be careful about deducing properties of PDE models from properties of the corresponding IPS or vice versa, see e.g. \cite{dawsonCriticalDynamicsFluctuations1983, arielLocustCollectiveMotion2015}.

To further understand these numerical results we first performed a linear stability analysis of the LS PDE, yielding a sufficient condition for the stability of the (space-homogeneous) steady states. However, in contrast to what happens with the GS models, even in the region of parameter space in which this stability condition  is broken we never observed any travelling waves or other space-inhomogeneous patterns. To get some further understanding we then performed a formal calculation of the spectrum of the (linearised) LS and GS PDEs. Such a calculation shows that the spectrum of the (linearised) GS PDE can possess positive and negative eigenvalues (again, depending on parameter choices), while the real part of the eigenvalues of the linearised LS PDE is never positive. We will be more precise in stating this result in Section \ref{sec:spectrum}.

Finally, as we have already mentioned,  \textit{ modulo some caveats}, the GS and LS PDEs have the same set of stationary states. A brief comment on this: while the set of stationary states for the GS PDE is relatively simple to determine analytically, the same is not true for the LS PDE. In particular, it is easy to see that the steady states of the GS PDE are also steady states of the LS PDE, but it is difficult to prove that for the LS PDE these are the only steady states \cite{buttaNonlinearKineticModel2019, garnierMeanFieldModel2019}. One of the purposes of this paper is to create numerical evidence for this conjecture and hence inform further analytical developments in this direction. More comments on this in Section \ref{sec:background and literature}.

\begin{figure}
	\centering
	\begin{minipage}{0.25\linewidth}
		\centering
		\includegraphics[width=\linewidth]{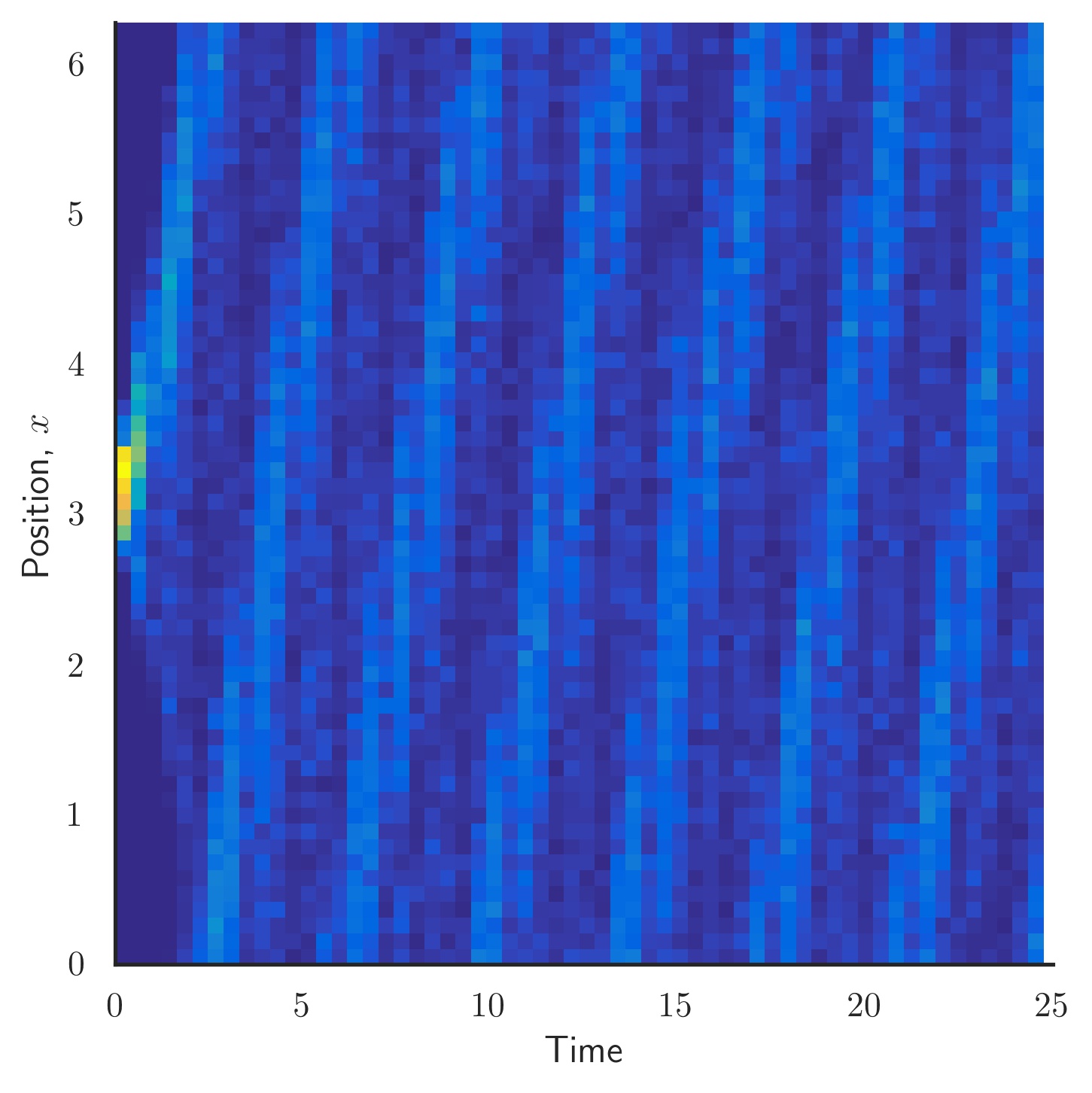}
		\subcaption{GS IPS}
	\end{minipage}%
	\begin{minipage}{0.25\linewidth}
		\centering
		\includegraphics[width=\linewidth]{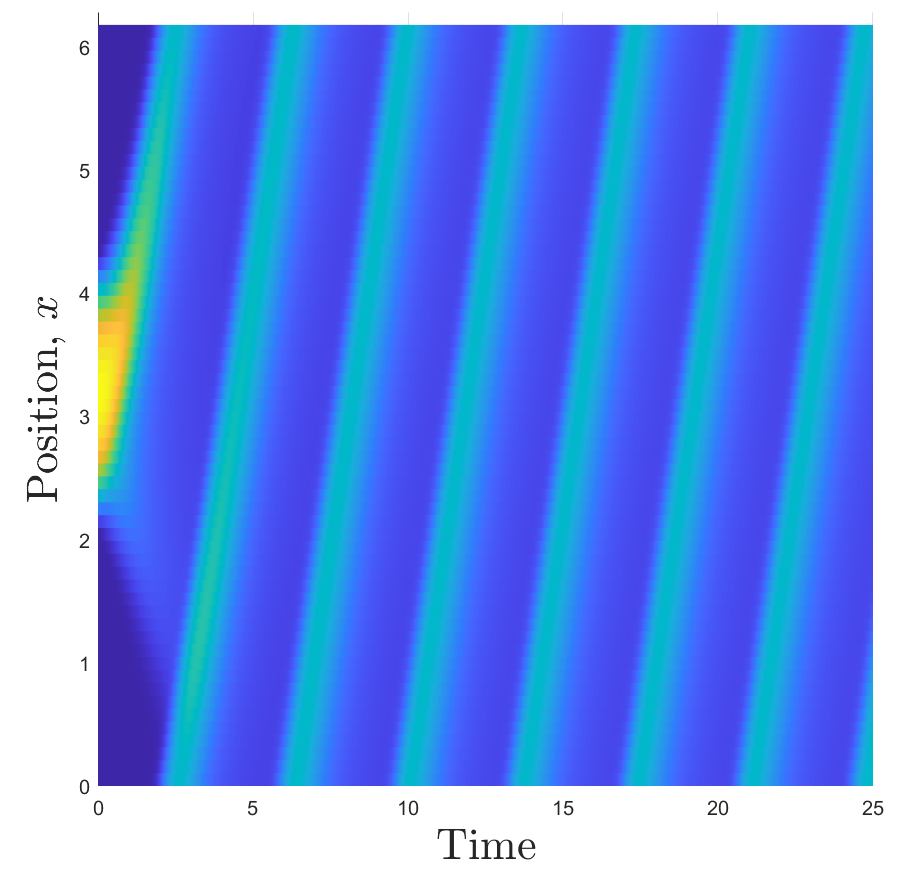}
		\subcaption{GS PDE}
	\end{minipage}\\
	\begin{minipage}{0.25\linewidth}
		\centering
		\includegraphics[width=\linewidth]{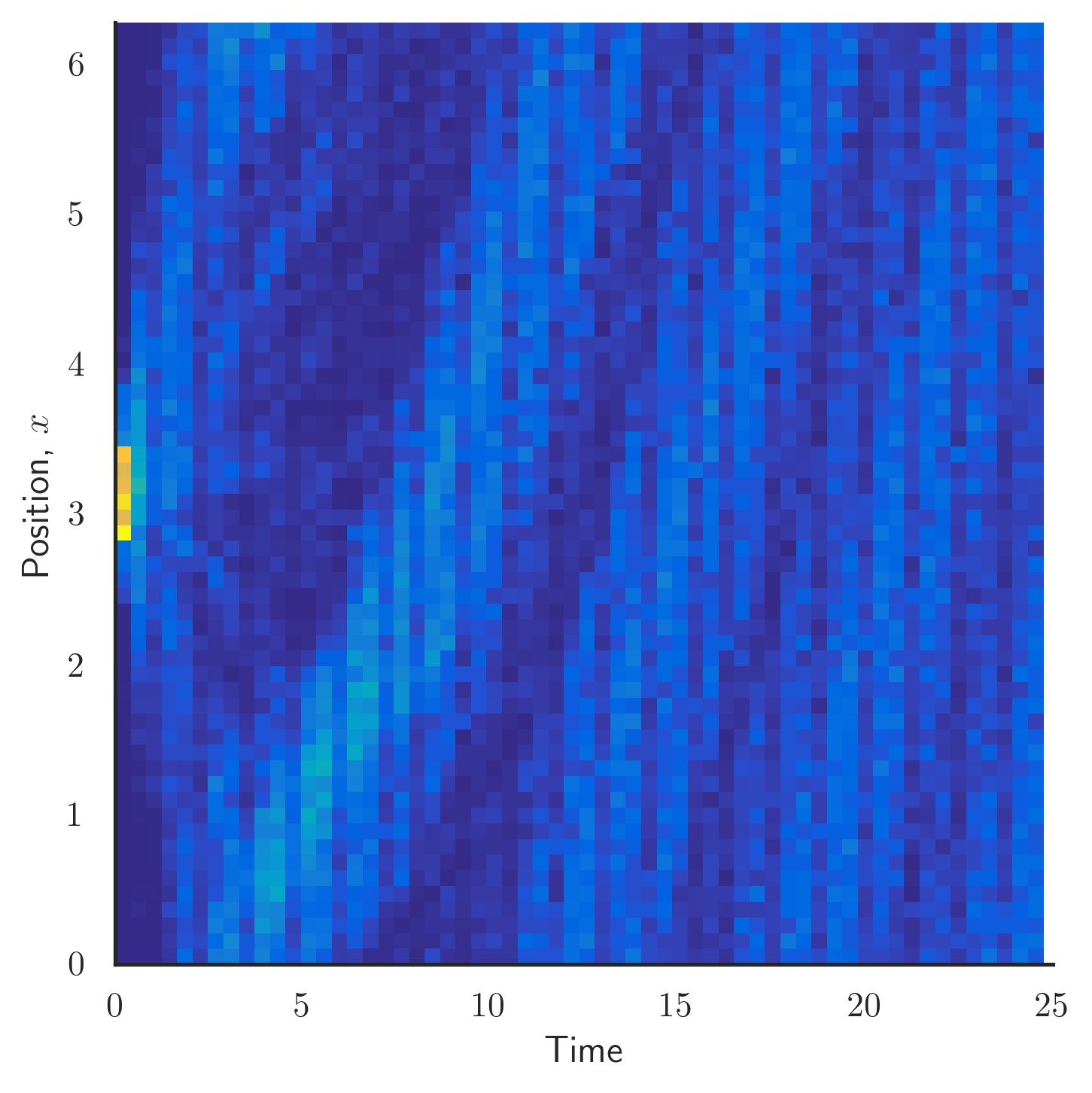}
		\subcaption{LS IPS}
	\end{minipage}%
	\begin{minipage}{0.25\linewidth}
		\centering
		\includegraphics[width=\linewidth]{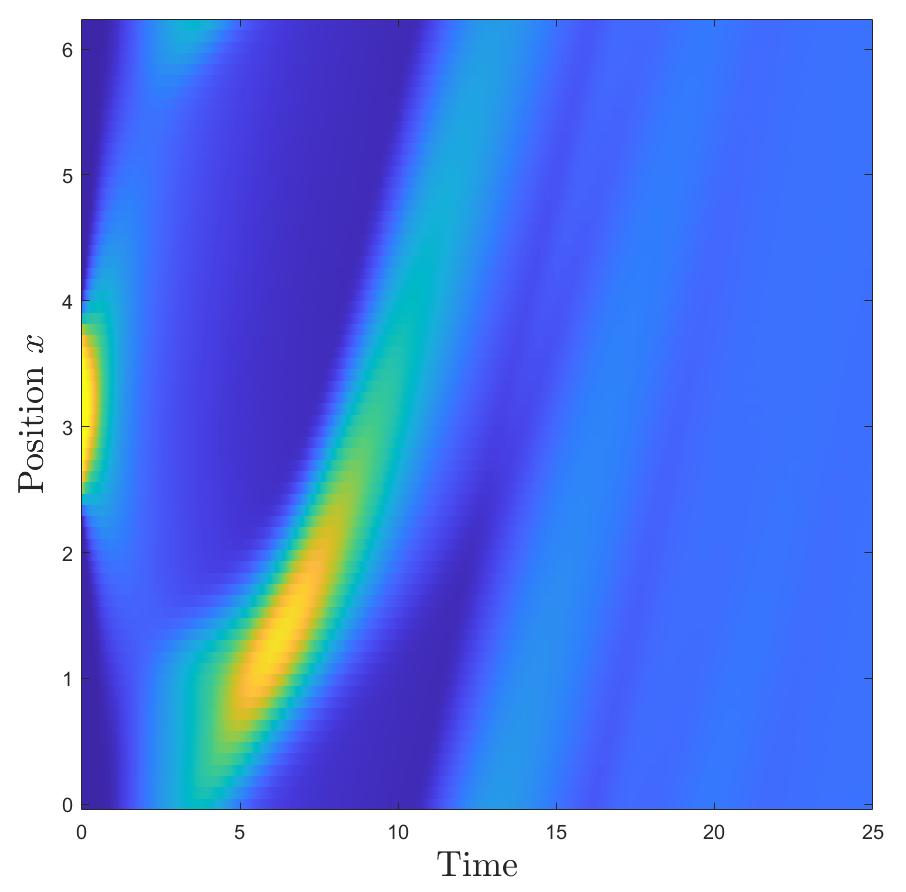}
		\subcaption{LS PDE}
	\end{minipage}
	\caption{For appropriate parameter regimes, the GS IPS \eqref{Garnierparsys1}-\eqref{Garnierparsys2}, LS IPS \eqref{parsys1}-\eqref{parsys2} and GS PDE \eqref{nlGarnier} exhibit pattern formation (travelling-wave-like behaviour) while the LS PDE \eqref{nl} homogenises in space. }
	\label{fig:kde_plots}
\end{figure}

The literature on IPSs and their associated PDEs is incredibly vast; a full literature review is out of reach hence in the above discussion we have mentioned only the works which are closest to this paper. For further works on emergence of travelling bands in various types of consensus formation models see e.g. \cite{aceves-sanchezLargeScaleDynamicsSelfpropelled2020, bouinTravellingWavesCane2014, calvezNonlocalCompetitionSlows2022, eftimieModelingGroupFormation2007, saragostiMathematicalDescriptionBacterial2010, hendersonPropagationSolutionsFisherKPP2016, wangNoisyHegselmannKrauseSystems2017,krukTravelingBandsClouds2020, eigentlerMetastabilityCoexistenceMechanism2019, eigentlerSpatialSelforganisationEnables2020} and references therein, to mention just a few.

The paper is organised as follows.
In Subsection \ref{subs:model} we describe the model considered in this paper and in Subsection \ref{sec:background and literature} we state our main results and put them in the context of existing literature.
In Section \ref{sec:preliminaries} we gather the necessary preliminary information regarding the numerical methods of choice; Section \ref{sec:ParticleSystem} and Section \ref{sec:TimeDependentPDESim}, respectively,  are  devoted to simulations of the involved particle system and PDEs, respectively. Section \ref{sec:linearstabilityanalysis} is devoted to the linear stability analysis for the LS PDE and, finally, Section \ref{sec:spectrum} is devoted to a formal calculation of the spectrum of the linearization of both  the GS and LS PDEs. We emphasize that  the analysis of Section \ref{sec:spectrum} is not a complete spectral analysis, but a calculation of (part of) the spectrum of the involved operators.  We give full details of why this is the case at the beginning of Section \ref{sec:spectrum} and in Note \ref{note:spectrumchoiceofspace}.

\section{Model and Main Results}\label{sec:model}
In Subsection \ref{subs:model} we present the models that we consider; in Subsection \ref{sec:background and literature}  we give more background on the introduced models,  recalling what is known and what are the main challenges in studying them, present our main results and put them in context of relevant literature.
\subsection{Description of the model}\label{subs:model}
We consider a  system of $N$ interacting particles evolving on the one-dimensional torus $\mathbb T$ of length $L$.  For each $i\in \{1, \dots, N\}$,  we denote by  $x_t^{i,N} \in \mathbb{T}$ and $v_t^{i,N} \in \mathbb{R}$, respectively, the position and velocity of the $i$-th particle at time $t$; the dynamics  of  $x_t^{i,N}$ and $v_t^{i,N}$ is described by the following system of SDEs
\begin{align}
	\mathrm{d} {x}_t^{i,N} & = v_t^{i,N} \, \mathrm{d} t\,,  \label{parsys1} \\\
	\mathrm{d} {v}_t^{i,N} & = - v_t^{i,N} \mathrm{d} t +
	G\left(\frac{ \sum_{j\neq i}^N \varphi (x_t^{i,N} - x_t^{j,N}) v_t^{j,N}}{ \sum_{j\neq i}^N \varphi (x_t^{i,N} - x_t^{j,N})}\right) \mathrm{d} t + \sqrt{2\sigma} \mathrm{d} B_t^{i}\,,\label{parsys2}
\end{align}
where the $B_t^i$'s are independent one-dimensional standard Brownian motions and  $\sigma>0$ is a given parameter. The functions $\varphi: \mathbb T \rightarrow \R_+$ and $G:\R \rightarrow \R$, are the so-called {\em interaction function} and  {\em herding function}, respectively.

When $\varphi$ is identically equal to one or strictly positive this is a mean-field type model; that is, every particle interacts with every other particle on the torus.  If the support of $\varphi$ is not the whole torus then the interaction is local, i.e.  particle $i$ interacts only with the particles $j$ within the support of $\varphi$. We will explain the modelling role of $G$ and $\varphi$ in more detail below; for the time being we clarify that, while the (support of the) function $\varphi$ determines the ``neighbourhood" (in space) over which each particle interacts, the function $G$ describes the way in which particle $i$ adjusts its velocity to the average velocity of its neighbours and, ultimately, it is responsible for consensus formation. Consistently with  the physics literature \cite{buhlDisorderOrderMarching2006, czirokCollectiveMotionSelfPropelled1999} we  assume that $G$ is of the form
\begin{equation}
	\label{herdingG}
	G(u) = - G(-u)\,, \qquad \left\{\begin{array}{ll} G(u) > u & \textrm{ if } 0<u<1\,, \\ G(u) < u & \textrm{ if } u>1\, , \end{array}\right.
\end{equation}
and in particular that $G$ has fixed points at $\pm 1$, i.e. $G(\pm 1) = \pm1$. {Functions of this form may have a jump at zero; to avoid this we take smooth versions of $G$, which have a third fixed point at the origin (this is not very dynamically relevant, we will explain later why).}

In \cite{buttaNonlinearKineticModel2019} it was shown that, as $N\rightarrow \infty$, the empirical measure of the particle system \eqref{parsys1}-\eqref{parsys2} converges to the continuum evolution for the particles' density $f(t,x,v)=f_t(x,v): \R_+\times \mathbb T \times \R \rightarrow \R$ described by the following nonlinear PDE\footnote{Throughout the paper we will deal with time dependent quantities. The notations $f(t)$ and $f_t$ will be used interchangeably to denote time-dependence.}
\begin{equation}
	\label{nl}
	\partial_t f_t(x,v)=-  v\,\partial_x f_t(x,v) - \partial_v\big\{\big[G(M_{f_t}(x)) - v\big] f_t(x,v) \big\} + \sigma\,\partial_{vv}f_t(x,v)\,,
\end{equation}
where the nonlinearity $M_f$ is given by
\begin{equation}
	\label{Mnl}
	M_f(x) :=  \frac{\int_{\mathbb{T}}\!\mathrm{d} y \int_{\mathbb{R}}\!\mathrm{d} w\, f(y,w)\,\varphi(x-y)\, w}{\int_{\mathbb{T}}\!\mathrm{d} y \int_{\mathbb{R}}\!\mathrm{d} w\, f(y,w)\,\varphi(x-y)}\,.
\end{equation}
The model \eqref{nl} was introduced in \cite{buttaNonlinearKineticModel2019} and it constitutes a continuum variant of the model introduced in the seminal works
\cite{buhlDisorderOrderMarching2006,czirokCollectiveMotionSelfPropelled1999}. Moreover,  the PDE \eqref{nl} can also be seen as the Fokker-Planck equation for a non-linear Stochastic Differential Equation (SDE), where the non-linearity is intended in the sense of McKean, see \cite{buttaNonlinearKineticModel2019}.

The evolution of the dynamics \eqref{parsys1}-\eqref{parsys2} (or, analogously, \eqref{nl}-\eqref{Mnl}),  is governed by friction, velocity-alignment and noise (the three terms on the RHS of equation \eqref{parsys2}, respectively).  When $G = 0$  the model reduces to many independent Langevin dynamics \cite{ottobreAsymptoticAnalysisGeneralized2011};  in this case it is well known that noise, while  directly acting only  on the particles' velocities, is transmitted to the particles' positions (by hypoelliptic effect \cite{ottobreAsymptoticAnalysisGeneralized2011}) so that the dynamics admits a unique invariant measure, namely
\be\label{invzero}
\mu_0(x,v) \dif x \dif v = \mathcal{U}_{\mathbb T}(x) \, \otimes \,  \frac{1}{\sqrt{2 \pi \sigma }}e^{-v^2/(2 \sigma)},
\ee
where $\mathcal{U}_{\mathbb T}(x)$ is the uniform distribution on the torus, and furthermore convergence to the measure $\mu_0$ takes place irrespective of the initial condition \cite{ottobreAsymptoticAnalysisGeneralized2011}. This implies that in stationarity the particles will be uniformly spread around the torus with velocities pointing in every direction (as $\mu$ has mean zero in velocity). In this sense, when $G=0$ the stationary state is a {\em disordered state}.

When $G \neq 0$ the second term on the RHS of \eqref{parsys2} forces each particle to align its velocity to the average velocity of its neighbours (i.e. to  the average velocity of the particles which, at time $t$, are in the support of $\varphi$) and in particular to slow down/accelerate if the neighbours are slower/faster. Because of the choice of $G(s)$ (which has fixed points at $s=1$ and  $-1$), the particles tend to align their velocities to either plus one or minus one (see Subsection \ref{sec:background and literature}), corresponding to clockwise/counterclockwise rotation.
We emphasize that in this model the particles  tend to align only their velocities -- hence the name {\em herding function} for $G$ --  and there is no explicit  space-aggregation effects. Because of this heuristic presentation, one expects that, after a transient phase, the particles will be found to be rotating either clockwise or anticlockwise (i.e. with average velocity $+1$ or $-1$).  This  is also corroborated by plenty of experimental evidence \cite{bertrandClusteringOrderingCell2020a,vicsekNovelTypePhase1995} and it is due to the fact that the particles system's (unique) invariant measure  has bimodal velocity-marginal (see Figure \ref{fig:switch}), with mean zero (see Appendix \ref{app:PSMeanZero})  and with modes concentrated around the values $\pm 1$.  Consequently, the {\em ordered states} corresponding to  clockwise and anticlockwise rotation are metastable states for the IPS, leading to  the particles  switching direction of motion at random times; these switches of average velocity are rare events and the mean time separating them is exponentially large in $N$, see \cite{garnierMeanFieldModel2019} and references therein. Hence, one expects these velocity switches to disappear in the PDE model (obtained  after taking the limit $N \rightarrow \infty$), and that the dynamics \eqref{nl}-{\eqref{Mnl}}  will have two distinct stationary states, namely $\mu_{\pm}$, with densities
\begin{equation}\label{invmeas}
	{\mu}_{\pm}(x,v)= \mathcal{U}_{\mathbb T}(x) \otimes \frac{1}{\sqrt{2\pi \sigma}} e^{-\frac{(v\mp 1)^2}{2 \sigma}} \,.
\end{equation}
The densities\footnote{We are abusing notation and denoting the measure and its density with respect to Lebesgue with the same letter. Note also that when we consider smoothed out versions of $G$ such that also zero is a fixed point of $G$ then a third invariant measure is expected, namely the measure \eqref{invzero}, which has mean zero in velocity. We don't discuss this measure as this is always unstable and indeed never observed in simulations.} $\mu_{\pm}$ have average velocity $\pm 1$, respectively, and they are both space-homogeneous.  One of the purposes of this paper is to produce numerical evidence that such densities are indeed the only stationary solutions of \eqref{nl},  we will come back to this in the next subsection.

In this paper we investigate the behaviour of the agent based model \eqref{parsys1}-\eqref{parsys2} and of its continuum counterpart \eqref{nl}; in particular we compare the behaviour of such locally scaled dynamics with their globally scaled analogues, i.e. with the following particle system
\begin{align}
	\mathrm{d} {x}_t^{i,N} & = v_t^{i,N} \, \mathrm{d} t\,,  \label{Garnierparsys1} \\\
	\mathrm{d} {v}_t^{i,N} & = - v_t^{i,N} \mathrm{d} t +
	G\left(\frac1N { \sum_{j\neq i}^N \varphi (x_t^{i,N} - x_t^{j,N}) v_t^{j,N}}\right) \mathrm{d} t + \sqrt{2\sigma} \mathrm{d} B_t^{i}\,,\label{Garnierparsys2}
\end{align}
and corresponding continuum description,
\begin{equation}\label{nlGarnier}
	\partial_t f_t(x,v)=-  v\,\partial_x f_t(x,v) - \partial_v\big\{\big[G(\tilde{M}_{f_t}(x)) - v\big] f_t(x,v) \big\} + \sigma\,\partial_{vv}f_t(x,v)\,,
\end{equation}
where this time the nonlinearity $\tilde{M}_f$ is given by
\begin{equation}
	\label{MnlGarnier}
	\tilde{M}_f(x) :=  {\int_{\mathbb{T}}\!\mathrm{d} y \int_{\mathbb{R}}\!\mathrm{d} w\, f(y,w)\,\varphi(x-y)\, w}\,.
\end{equation}
For both models we will always consider interaction functions (IF) $\varphi$ such that $\varphi\geq 0$ and $\frac 1L\int_{\mathbb T} \varphi =1$.

\subsection{Background,  relation with literature and main results}\label{sec:background and literature}

The difference between the two models introduced in the previous subsection  boils down to a different choice of normalization in the argument of $G$: in \eqref{nl}-\eqref{Mnl}  particle $i$ interacts with all the other particles within the support of $\varphi$ and the strength of the interaction is normalised by the number of particles with which agent $i$ is interacting (at time $t$); in contrast, in \eqref{nlGarnier}-\eqref{MnlGarnier} the strength of the interaction is always normalised by the total number of particles.

In  \cite{garnierMeanFieldModel2019} it is shown that the measures \eqref{invmeas} are the only stationary solutions of the global model \eqref{nlGarnier}-\eqref{MnlGarnier}, irrespective of the choice of $\varphi$ (even if $\varphi$ has compact support).

One can easily see that the measures \eqref{invmeas} are also stationary solutions of the local model \eqref{nl}-\eqref{Mnl}; however neither  the authors of \cite{buttaNonlinearKineticModel2019} nor those of \cite{garnierMeanFieldModel2019} were  able to prove that these are the {\em only} stationary solutions of \eqref{nl}-\eqref{Mnl}. The difficulty in proving this result analytically comes from the lack of gradient structure of equation \eqref{nl}.
The fact that the measures \eqref{invmeas} are the only stationary solutions of \eqref{nl}-\eqref{Mnl} has been proven analytically in \cite{buttaNonlinearKineticModel2019} only when the function $\varphi$ is chosen to be a small oscillation around  the constant function $\varphi \equiv 1$, i.e.
$$
	\varphi(x)= 1+ \lambda \psi(x),   \quad \mbox{where } 0<\lambda \ll 1, \int_{\mathbb T} \psi(y) \, \dif y =0 \,,
$$
or in the much simplified setting in which one considers only space-homogeneous solutions of \eqref{nl}. The latter scenario is quite helpful for general understanding, so we briefly recall what is known in that case.
If we consider only space-homogeneous solutions of \eqref{nl}, i.e. solutions of the form $f_t(v)$, then the LS PDE  \eqref{nl}-\eqref{Mnl} reduces to the (much) simpler dynamics\footnote{Note also that in this case the GS PDE and the LS PDE coincide. }
\begin{equation}\label{lin}
	\partial_t f_t(v)=- \partial_v\big\{\big[G(M(t)) - v\big] f_t(x,v) \big\} + \sigma\,\partial_{vv}f_t(x,v)\,,
\end{equation}
where $M(t):= \int_{\R}v\,  f_t(v) \, dv$ solves the one-dimensional ODE
\begin{equation}\label{Mlin}
	\dot M (t)= G(M(t))-M(t) \,.
\end{equation}
In other words, the nonlinearity $M_{f_t}(x)$ reduces to being the average velocity $M(t)$; because one can close an equation on $M(t)$, one can now regard $M(t)$ as being a time-dependent coefficient; hence, in this space-homogeneous case,  the nonlinear dynamics \eqref{nl} reduces to a linear, non-autonomous PDE. The equation  \eqref{Mlin}  satisfied by $M(t)$ is a simple one-dimensional autonomous ODE, hence its behaviour is straightforward to understand:  because of the assumption \eqref{herdingG} on the herding function (and in particular because $\pm 1$ are the only fixed points of $G$, i.e. $G(\pm 1)=\pm 1$,  so that $\pm 1$ are the only stationary points of \eqref{Mlin}),  from \eqref{Mlin} one can easily see that
\begin{align}\label{meanasymp}
	  & M(0)>0 \,\, \Rightarrow \,\, M(t)\stackrel{t\rightarrow \infty}{\longrightarrow} 1 \nonumber          \\
	  & M(0)<0 \,\, \Rightarrow \,\, M(t)\stackrel{t\rightarrow \infty}{\longrightarrow} -1 \,. \footnotemark
\end{align}
\footnotetext{ Again, if also $G(0)=0$ then $M(0)=0 \Rightarrow M(t)=0$ for every $t>0$. From now on we stop discussing this case.}
In this space-homogeneous case it is easy to prove (see \cite{buttaNonlinearKineticModel2019}) that the measures \eqref{invmeas} are the only stationary solutions of \eqref{lin}.
Note that such measures are the only  stationary solutions of the dynamics \eqref{lin}-\eqref{Mlin}, \textit{ irrespective of the value of $\sigma$}. This feature is unusual for these kind of models (if compared e.g. to McKean-Vlasov like equations, where the number of stationary solutions changes depending on the strength of the noise $\sigma$, see e.g. \cite{dawsonCriticalDynamicsFluctuations1983, carrilloLongTimeBehaviourPhase2019,herrmannNonuniquenessStationaryMeasures2010} and references therein), and it is sometimes referred to as {\em unconditional flocking}.  Moreover, again in the simplified case \eqref{lin},  if the initial profile $f_0$ has positive (negative, respectively) mean, i.e. if $M(0)>0$ ($M(0)<0$, respectively) then the dynamics converges exponentially fast to the stationary measure with positive mean, $\mu_+$ (negative mean, $\mu_-$, respectively), coherently with \eqref{meanasymp}, see \cite{buttaNonlinearKineticModel2019}. In other words, in this simplified case, knowing the sign of the initial mean  velocity is sufficient to determine the asymptotic behaviour of the dynamics. This is certainly not the case for the fully non-linear PDE \eqref{nl}-\eqref{Mnl}; we show this fact numerically in Section \ref{sec:TimeDependentPDESim}.

While the basin of attraction changes if we consider the space-homogeneous model \eqref{lin} versus the non-space-homogeneous  PDE \eqref{nl}, in this paper we conjecture that the stationary states of \eqref{nl} are the same as those of the models \eqref{lin} and \eqref{nlGarnier}, namely they are given by the measures \eqref{invmeas}. However, while such stationary states are unstable (in an appropriate parameter regime) for the GS PDE \eqref{nlGarnier}, we numerically observe them to be always stable for the LS PDE \eqref{nl} (and they are known to be stable for the linear equation \eqref{lin}, see \cite{buttaNonlinearKineticModel2019}). Let us comment on these two facts in turn, i.e. on understanding the set of steady states for the LS PDE and on determining their stability, starting from the first matter.

One way of showing that the measures \eqref{invmeas} are the only stationary states for \eqref{nl},  is demonstrating that the dynamics \eqref{nl}-\eqref{Mnl} always homogenizes in space (we refer to this as ``mixing in space"). Indeed, if the stationary solution of \eqref{nl} is space-homogeneous, then it can only be one of the measures \eqref{invmeas} (this is straightforward, see  Lemma \ref{lem:whatweknow}.) For this reason many of our simulation scenarios are aimed at showing that the LS PDE will always mix in space. Indeed, according to the heuristic description of the local model given in the previous section,  one does not expect the formation of travelling-wave-like patterns in the LS PDE: one expects that clusters (profiles that are compactly supported in space) will not persist because of the effect of the noise, which tends to break them by making the particles' velocities diffuse in every direction,  so that initial data with compact support in space will eventually homogenize.
The same intuition does not hold for the globally scaled model and indeed, as explained in the Introduction,  our numerical experiments show that the GS PDE does exhibit travelling wave solutions in certain parameter regimes (and this is true both for the GS IPS \eqref{Garnierparsys1}-\eqref{Garnierparsys2} and for the GS PDE  \eqref{nlGarnier}). We give an intuition of why this is the case in Note \ref{note:heuristics} below.
\begin{note}\label{note:heuristics}\textup{
	To give an intuitive explanation as to why travelling waves  emerge in the GS PDE, but not the LS PDE, one can  use a variation of the heuristic reasoning used in
	\cite{motschNewModelSelforganized2011} to explain the drawbacks of the global scaling.  To this end, consider  the local average velocity, $\lav(t,x)$, namely
	\[
		\lav(t,x):= \int_{\R} v \, f_t(x,v) \, \dif v \,
	\]
	(local because it  is the average velocity at $x$ rather than over the whole torus).
	Multiplying \eqref{nl} by $v$, integrating in the velocity variable and (formally) integrating by parts,  \footnote{These integrations by parts can be fully justified, see \cite{buttaNonlinearKineticModel2019}.} one obtains
	\begin{equation}\label{motsch1}
		\pa_t \lav(t,x) =
		- \lav(t,x) +
		G \left( \frac{\iint f_t(y,w) \varphi(x-y) w \, \dif w \, \dif y}{\iint f_t(y,w) \varphi(x-y)  \, \dif w \, \dif y} \right)  -\sigma \, \pa_x \!\!\left( \int_{\R} v^2 \, f_t(x,v) \, \dif v\right) \,.
	\end{equation}
	Now suppose the initial space-distribution is given by two clusters, widely separated, and one containing most of the mass (i.e., most of the particles); suppose also the `bigger cluster' starts with higher average initial velocity than the smaller. To fix ideas, say  $f_0(x,v) = f_0^{(a)}(x) \pi_0^{(a)}(v)+ f_0^{(b)}(x)\pi_0^{(b)}(v)$ with $f_0^{(a)}(x)$ and $f_0^{(b)}(x)$ such that their support is disjoint, $\int_{\mathbb{T}} f_0^{(a)}(x)\dif x= 1-\epsilon$,  $\int_{\mathbb{T}} f_0^{(b)}(x) \dif x = \epsilon$, for some small $\epsilon$ and $\int_{\R}\dif v \, \pi_0^{(b)}(v) = \int_{\R}\dif   v \,  \pi_0^{(a)}(v) = 1$ but $0< v_b:=\int_{\R}\dif v  \, \pi_0^{(b)}(v) v< \int_{\R}\dif v \,  \pi_0^{(a)}(v) v =: v_a$.   Suppose also that the support of $\varphi$ is much smaller than the initial distance between the two clusters, so that such clusters will not initially interact. If we look at the variation of the  local average velocity at time zero at a point $x_2$ in the support of the smaller cluster $f_0^{(b)}$, from \eqref{motsch1} we obtain
	\begin{align}
		\pa_t\lav(t,x_2) \vert_{t=0} & =
		- \lav(0,x_2) +
		G \left( \frac{ v_b\int f_0^{(b)}(y) \varphi(x_2-y)   \, \dif y}{\int f_0^{(b)}(y) \varphi(x_2-y) \, \dif y} \right)  - \sigma \, \pa_x \!\!\left( \int_{\R} v^2 \, f_0(x,v) \, \dif v\right)(x_2) \, \label{motsch2}\\
		                             & =
		- \lav(0,x_2) +
		G \left( v_b\right)  - \sigma \, \pa_x \!\!\left( \int_{\R} v^2 \, f_0(x,v) \, \dif v\right)(x_2) \,. \nonumber
	\end{align}
	Had we used the normalization \eqref{MnlGarnier} appearing in \eqref{nlGarnier} rather than the local normalization \eqref{Mnl}, the second term on the RHS of the above would have been
	$$
		G \left( {v_b\int f_0^{(b)}(y) \varphi(x_2-y)   \, \dif y} \right).
	$$
	Since $f_0^{(b)} \simeq \epsilon$,  with the local scaling one has
	$$
		G \left( \frac{v_b\int f_0^{(b)}(y,w) \varphi(x_2-y)  \, \dif y}{\int f_0^{(b)}(y,w) \varphi(x_2-y)  \, \dif y} \right)  \simeq G (O(1))  \simeq O(1),
	$$
	versus
	$$
		G \left( {v_b\int f_0^{(b)}(y,w) \varphi(x_2-y) w \, \dif y} \right) \simeq G(O(\epsilon)) \simeq O(\epsilon),
	$$
	with the global scaling. The consequence of this is as follows: under the global scaling the contribution to the local average velocity that comes from the nonlinear term, for a point in the smaller cluster, is negligible and, in particular, is insufficient to counterbalance the effect of the damping term. Hence, in appropriate parameter regimes, it can happen that the small cluster tends to slow, at least until the bigger and faster cluster catches up with it. When the two clusters interact, the slower one can accelerate briefly, but, if the parameters are such that the interaction is not strong enough, it fails to keep pace with the faster cluster and it is left behind. This repeats every time the two clusters meet, giving rise to the periodic, travelling-wave like solutions that we observe.
	On the other hand, under the local scaling,  the contribution to the average velocity for a point in the smaller cluster coming from the nonlinear term is of order one.\footnote{According to this heuristic description one might think that if $G$ was non smooth, but still satisfying \eqref{herdingG} (to fix ideas take $G(u)=(u+1)/2$ for $u>0$ and $G(u)=(u-1)/2$ for $u<0$), then the difference between the two models should disappear. In other words one might think that the difference in behaviour is an artifact of the fact that we chose $G$ to be smooth (which gives $G(O(\epsilon)) \simeq O(\epsilon)$).  According to preliminary simulations, this is not the case and the difference between the two PDE models persists irrespective of the regularity of $G$. We do not show simulations with $G$ non smooth, as for them to be thorough they would require different numerical methods from the ones used here (both for the PDEs and for the IPSs) and we keep this point for further future investigation.} Consequently the smaller, slower cluster, increases its velocity and eventually catches up with the bigger one. The agglomerated cluster is eventually broken by noise, homogenising in space -- at least in the LS PDE, for some parameter choices inhomogeneity can persist in the LS IPS as well. The latter fact is due to the `finite particle count':  if the small cluster contains $N_b$ particles and the large one contains $N_a$ particles, say all of them started deterministically with velocities $v_b$ and $v_a$, respectively, when there are finitely many particles, the order of the ratios of interest according to the above heuristic, namely  $N_b v_b/(N_b+N_a)$ (for the GS IPS) and  $N_b v_b/N_b$ (for the LS IPS), may not differ significantly.}
	\hfill{$\Box$}
\end{note}

To better understand our numerical results on the GS and LS PDE  we use linear stability methods and a formal spectral analysis. Let us summarise the findings of our stability analysis here (see Proposition \ref{lemma:modestability} below); complete proofs and details are postponed to Section \ref{sec:linearstabilityanalysis}.

We linearise the dynamics \eqref{nl} around the two equilibria $\mu_{\pm}$, see equation \eqref{blablu}; we then Fourier transform the linear equation both in space and velocity. When doing so one can see that the Fourier modes of the solution of the linearised equation decouple, see  \eqref{eqnforgo}; in view of this fact we say that the system is  {\em linearly stable} around $\mu_{\pm}$, if and only if  all the Fourier modes of the linearised equation are stable (in a sense that we make precise in \eqref{stabilitycondition}). We give full details of such a procedure in Section \ref{sec:linearstabilityanalysis}, here we just report the result, which gives a simple condition to check linear stability.

\begin{proposition}
	\label{lemma:modestability}
	Let $\xi=\pm 1$ and, correspondingly, let $\mu_{\xi}$ denote $\mu_{\pm}$. The $0$-th Fourier mode of the linearization around $\mu_{\xi}$ of the nonlinear Fokker Planck equation \eqref{nl} is stable (in the sense defined in \eqref{stabilitycondition}) iff $G'(\xi) \varphi_0 < 1$. If
	$$
		\lv G'(\xi)\varphi_k \rv  a_k <1 \quad \mbox{for  } k\neq 0
	$$
	where

	$$
		a_k := \left(1+ \frac{3 \sqrt{\pi}}{\sqrt{\sigma} \lv D_k\rv}+ \frac{3 \lv \xi \rv }{\sigma \lv D_k\rv} + \frac{e^{-1}}{1+\sigma \lv D_k\rv^2 } + \frac{3}{\sigma \lv D_k\rv} \right), \quad D_k:=\frac{2\pi k}{L}\,,
	$$
	$\varphi_k:=  \frac 1L\int_0^L \varphi(x) e^{i 2\pi k x/L} \dif x$ is the $k$-th Fourier mode of $\varphi$,  then also the $k$-th mode is stable.
	Moreover, since $a_k\leq a_1$, if $G'(\xi) \varphi_0 <1$ and
	\begin{equation}\label{eq:localStabilityInequality}
		\lv G'(\xi)\varphi_k \rv a_1 < 1
	\end{equation}
	then all the modes are stable and we say that the equilibrium $\mu_{\xi}$ is linearly stable.
\end{proposition}

\begin{note}\textup{ According to the above proposition, given the herding function $G$ and the IF $\varphi$,  for $\sigma$ large enough the equilibria $\mu_{\pm}$  are  linearly stable for the LS PDE.
		A similar stability analysis was carried out in \cite{garnierMeanFieldModel2019} for the GS PDE \eqref{nlGarnier},  giving (in our notation) the following:  the $0$-th mode is linearly stable iff $\varphi_0 G'(\xi)< 1$. If
		\begin{equation}\label{eq:globalStabilityInequality}
			\lv G'(\xi)\varphi_k \rv  \tilde{a}_k <1 \quad \mbox{for every } k\neq 0
		\end{equation}
		where
		$$
			\tilde{a}_k = \left(1+ \frac{3 \sqrt{\pi}}{\sqrt{\sigma} \lv D_k\rv}+ \frac{3 \lv \xi \rv }{\sigma \lv D_k\rv} + \frac{e^{-1}}{1+\sigma \lv D_k\rv^2 } \right)
		$$
		then also the $k$-th mode is stable, for every $k\neq 0$.
		We emphasize that the stability conditions \eqref{eq:localStabilityInequality} and \eqref{eq:globalStabilityInequality} are sufficient conditions; that is, for parameter choices satisfying \eqref{eq:localStabilityInequality}/\eqref{eq:globalStabilityInequality}  the LS PDE/GS PDE is linearly stable. However such conditions don't imply any stability statement about the region of parameter space where they are violated. In other words, when   \eqref{eq:localStabilityInequality}/ \eqref{eq:globalStabilityInequality} is violated, we can't say anything about the stability or instability of the LS PDE/GS PDE.
		With this in mind, let $\sigma_c^l$ be the critical value of $\sigma$ above which stability of the local model is guaranteed and $\sigma_c^g$ the analogous value for the global model. First of all note that such values, as obtained through the linearization procedure, are only estimates of the real critical values. Nonetheless,   because $\tilde{a}_k<a_k$, we have $\sigma_c^l>\sigma_c^g$. In other words, the estimated stability region for the LS PDE is ``smaller" than the one for the GS PDE. This is  counter intuitive and certainly not satisfactory in view of our numerical results, which show that the equilibria of the LS PDE are always stable (so that one would expect $\sigma_c^l=0$ or at the very least $\sigma_c^l<\sigma_c^g$).   Because   \eqref{eq:localStabilityInequality} and  \eqref{eq:globalStabilityInequality} are just sufficient conditions, the fact that $\sigma_c^l>\sigma_c^g$, which is obtained by using \eqref{eq:localStabilityInequality} and  \eqref{eq:globalStabilityInequality},   does not contradict any numerical findings, it simply does not help to give a satisfactory explanation of them. As we already said in the Introduction, this motivates the formal spectral analysis of Section \ref{sec:spectrum}. Figure \ref{figPDE:StabilityRegion} illustrates the regions where inequalities \eqref{eq:localStabilityInequality} and \eqref{eq:globalStabilityInequality} are satisfied.}
	\hfill{$\Box$}
\end{note}

We remark here that, as frequent for these models, while both the LS IPS \eqref{parsys1}-\eqref{parsys2}   and the GS IPS \eqref{Garnierparsys1}-\eqref{Garnierparsys2} admit a unique invariant measure, the corresponding  limiting PDEs have multiple stationary solutions. Hence the behaviour of the particle system can be misleading if one is interested in understanding the behaviour of the limiting PDE (see for example Figure \ref{fig:oneclustervarygamma} vs Figure \ref{figPDE:E1_VaryRatio_Local_lowNoise}). For this reason we simulate here both the PDE and the related particle system. What emerges in this case is that while both the GS IPS and the GS PDE exhibit pattern formation, the LS IPS seems to retain the pattern formation, while the limiting LS PDE seems to always converge to space-homogeneous solutions, as we have already discussed.

Besides studying stability and pattern formation, from a modelling perspective, we numerically investigate how the range of the interaction function $\varphi$ affects the dynamics. In particular, we aim to answer the following question: in order  for the dynamics to converge to equilibrium (i.e. to start flocking/create consensus) as fast as possible, is it more advantageous for every agent to interact with every other agent  or is it possible to reach consensus faster by purely local interactions? I.e. is convergence to equilibrium faster in the mean-field case or in the non-mean-field scenario? More importantly, does the answer to this question  change depending on whether we consider the locally scaled dynamics \eqref{nl}-\eqref{Mnl} versus the globally scaled one \eqref{nlGarnier}-\eqref{MnlGarnier}?  These questions are addressed in Section \ref{sec:ParticleSystem} and Section \ref{sec:TimeDependentPDESim}.

Finally, we mention for completeness that when $\varphi$ is uniformly positive, i.e. $\varphi(x)>\varepsilon$ for some constant $\varepsilon>0$, the well-posedness of the PDE \eqref{nl} has been studied in \cite{buttaNonlinearKineticModel2019}. When $\varphi$ is compactly supported well posedness is more involved, and it can be achieved via tools similar to those employed in \cite{debusscheExistenceMartingaleSolutions2020}. We don't do this in this paper and assume well-posedness even when $\varphi$ is compactly supported.

\section{Preliminaries on numerics} \label{sec:preliminaries}
To simulate the particle system \eqref{parsys1}-\eqref{parsys2} we use the explicit Euler-Maruyama method, namely,
\begin{align*}
	x^{i,N}_{n+1} & = x^{i,N}_n + v^{i,N}_n\Delta t \mod{2\pi}                                                                                                                                                                               \\
	v^{i,N}_{n+1} & = v^{i,N}_n + \left[G\left(\frac{ \sum_{j\neq i}^N \varphi (x_n^{i,N} - x_n^{j,N}) v_n^{j,N}}{ \sum_{j\neq i}^N \varphi (x_n^{i,N} - x_n^{j,N})}\right)-v^{i,N}_n \right]\Delta t + \sqrt{2\sigma\Delta t} \Delta B^i_n,
\end{align*}
where $x_n^{i,N}$ and $v_n^{i,N}$ are the numerical approximations of the positions and velocities of particle $i$ at time step $t_n = n\Delta t$, respectively,  $\Delta B^i_n \sim \mathcal{N}(0,1)$ is a standard Brownian increment, and $n=0,1, \dots$.  The time step $\Delta t$ depends on the system being simulated and is specified in the caption of each figure. The main concern when choosing $\Delta t$ is to resolve the interactions correctly; we calibrated the time step such that halving it did not appreciably change the results (as a rule of thumb, we observed that it is enough to let at least 10 timesteps occur during any particle interaction).

When $\varphi$ has compact support, the denominator of the interaction term can be zero. In this case the numerator is also zero (as the interaction function $\varphi$ is non-negative). To prevent division error a small positive number ($10^{-15}$) is always added to the denominator so that the interaction term can be calculated.
Aside from these small clarifications, the simulation of the LS IPS is straightforward. On the other hand, simulating the PDEs \eqref{nl} and \eqref{nlGarnier} is delicate, in particular due to the non-local, non-linear nature of the interactions and the (formally) infinite velocity domain. To simulate these PDEs  we have chosen a Pseudospectral Method (PSM); in particular, we have implemented the numerical scheme introduced in \cite{noldPseudospectralMethodsDensity2017}, with some (minor) modifications. Details are provided in Appendix \ref{app:PDEComp}; here we just motivate our choice and describe essential features of the scheme.

In implementing our PSM we have used the software \textit{2DChebClass} in MATLAB \cite{goddard2DChebClass2017}, which was originally designed for equations arising in Dynamic Density Functional Theory (DDFT). Such equations are typically highly non-linear, stiff and non-local.  Equation \eqref{nl} shares many of these features, and so the scheme is readily transferable to our case. Finite element methods (FEMs), while robust for these equations, are computationally expensive (primarily because the non-local terms destroy the sparsity of the resultant matrices). For problems with smooth solutions (which is the case here, see \cite{debusscheExistenceMartingaleSolutions2020, buttaNonlinearKineticModel2019}), PSMs provide exponential accuracy in the number of grid points, meaning significantly fewer points are required than in FEMs.
Furthermore, we observe that PDEs \eqref{nl} and \eqref{nlGarnier} preserve both mass and positivity; our simulation scheme preserves mass and positivity within a very low tolerance, for appropriate choices of the grid, which serves as a consistency check. See also \cite{russoFinitevolumeMethodFluctuating2021} and references therein.
We also recall that equation \eqref{nl}  is non-elliptic and indeed noise acts only on the velocity variable. Care must therefore be taken not to introduce any artificial diffusion into the space variable--a common pitfall of many standard techniques, such as upwind schemes in Finite Difference Methods (FDM). Doing so would surely `break' any clusters that may form under the dynamics, and hence mixing in space would be an artifact of the scheme. Figure \ref{fig:ArtificialDiffusion} illustrates a test case showing that our scheme does not introduce any artificial diffusion, compared to other standard schemes.   

Typical criticisms of PSMs, such as their loss of accuracy on complex geometries or in the presence of non-smooth solutions, do not apply here as the solution is smooth and the geometry is relatively simple \cite{trefethenSpectralMethodsMATLAB2000}.
The unbounded velocity domain, \(\R\), will be truncated to the finite domain \([-L_v,L_v]\) where \(L_v\) is chosen large enough to minimise mass loss. Note that \textit{2DChebClass} also allows for computation on an unbounded domain; in Appendix \ref{app:PDEComp} we will justify our choice to truncate the velocity domain nonetheless.

The numerical scheme used to simulate the dynamics \eqref{nl} depends on several parameter, all of which are described in Appendix \ref{app:PDEComp}. We call $\zeta$ the vector containing all such parameters (such as number of grid points, choice of truncation, etc.) and, correspondingly, if $f_t(x,v)$ is the analytical solution of \eqref{nl}, then $f^{\app}_t(x,v)$ is the approximation produced via our scheme of choice. The number of mesh points will be chosen so that the initial condition is well-interpolated and mass and positivity is conserved to high accuracy (50 points in each of position and velocity is usually sufficient).

\textbf{Parameter choices.} In the physics literature \cite{czirokCollectiveMotionSelfPropelled1999}, a typical choice of herding function $G$  is
\begin{equation}\label{eq:Gstep}
	G(u) = \begin{cases}\frac{u+1}{2} & \text{ for }  u>0    \\
		\frac{u-1}{2} & \text{ for } u<0 \,.\end{cases}
\end{equation}
The above has a jump at the origin, which could potentially cause numerical instabilities. As such we consider a smoothed version of the above, which preserves property \eqref{herdingG}, namely
\begin{equation} \label{eq:Gsmooth}
	G_{\alpha}(u) = \frac{\atan(\alpha u)}{\atan(\alpha)}, \qquad \alpha>0.
\end{equation}
In all simulations, unless specified otherwise in the caption, we choose $\alpha=1$ and fix the length $L$ of the torus to be $L=2 \pi$.
As far as the IF $\varphi$ is concerned, we will work either with the following compactly supported function
\begin{equation}\label{eq:phiIndicator}
	\varphi^{\gamma}(x) = \frac{1}{2\gamma} \mathbf{1}_{[0,2\pi\gamma]}(\|x\|), \qquad 0\leq\gamma\leq \frac{1}{2},
\end{equation}
where $\|\cdot\|$ denotes the distance on the torus, that is,
$$ \|x\| := \min(|x|, 2\pi-|x|), \quad \mbox{for any } x\in\mathbb{T} \,, $$
or with
the following strictly positive `bump function',
\begin{equation}\label{eq:phiBump}
	\varphi(x) = \mathcal{Z}^{-1} \exp\left(-\frac{1}{1 - \frac{\|x\|^2}{\pi^2}}\right),
\end{equation}
where $\mathcal{Z}$ is the appropriate normalising constant (such that $\int_0^{2\pi} \varphi = 2\pi$). Analogously,
in equation \eqref{eq:phiIndicator},
the prefactor in front of the indicator function is chosen so that   $\int_0^{2\pi} \varphi^{\gamma}(x) \,\dif x = 2 \pi$; for later purposes  we also note that  the $k$-th Fourier mode of $\varphi^{\gamma}$, i.e. $\varphi^{\gamma}_k :=  \frac{1}{2\pi}\int_0^{2\pi} \varphi^{\gamma}(x)\mathrm{e}^{i2 \pi kx/L} \,\dif x$, is given  by

\begin{align*}
	\varphi_k^{\gamma} & =  \frac{1}{2\pi}\int_0^{2\pi} \varphi^{\gamma}(x) \rme^{\rmi kx} \rmd x
	= \frac{1}{4\pi\gamma}\int_{-2\pi\gamma}^{2\pi\gamma} \rme^{\rmi kx}\rmd x\\
	                   & = \frac{1}{4\pi\gamma}\frac{1}{\rmi k}\left( \rme^{2\pi\rmi \gamma k} -  \rme^{-2\pi\rmi \gamma k}\right)
	= \frac{1}{2\pi \gamma k}\frac{1}{2\rmi} \left( \rme^{2\pi\rmi \gamma k} -  \rme^{-2\pi\rmi \gamma k}\right)\\
	                   & = \frac{\sin(2\pi\gamma k)}{2\pi \gamma k} = \operatorname{sinc} (2\pi\gamma k) \, ,
\end{align*}
where $\operatorname{sinc}(x)=\sin(x)/x$.
The parameter $\gamma$ is related to the \emph{interaction radius} $R$ by
$$ R=2\pi\gamma;$$
because we fix the length of the torus to $L=2 \pi$ in every simulation,  $2\gamma$ gives the fraction of the torus that can be seen by each particle.  The choice  $\gamma = 0.5$ corresponds to every particle interacting with every other particle on the torus and is equivalent to $\varphi\equiv 1$. We note here that, because of the lack of smoothness of the interaction function  \eqref{eq:phiIndicator},  one should be mindful of numerical instabilities developing in PDE simulations. We show in Appendix \ref{app:PDEComp} that, because of our choice of numerical scheme and the form of our equation, no such issues arise.
For the purpose of making fair comparisons, the same herding and interaction functions will be used in corresponding particle/PDE simulations.
{\bf Initial configurations.}  We gather here choices of  initial configurations. We describe both the discrete random initial conditions  used for the IPS and the analogous deterministic continuous version for the PDEs. We mostly used initial configurations in product form (i.e. space and velocity are decoupled).   First, the position distributions: one or multiple clusters centred at $a_1, \dots a_n$, each of width $2\pi w$. In the particle system, each cluster is simply
\begin{align*}\tag{IC1}\label{ic:PSonecluster}
	x \sim \mathcal{U}\left[a- \pi w, a+\pi w\right]\, ,
\end{align*}
i.e. a uniform distribution over the interval $\left[a- \pi w, a+\pi w\right]$.
For PDE simulations, discontinuities in the initial density pose a problem, due to constraints of the solver. To avoid this, to represent cluster-like initial data in the PDE we use either a `bump function' which is strictly positive or one with compact support, namely
\begin{align*}\tag{IC2a}\label{ic:PDEbump}
	h(x) = \exp\left(\frac{-1}{1-\frac{\|x-a\|^2}{(\pi w)^2}}\right)
\end{align*}
or
\begin{align*}  \tag{IC2b}\label{ic:PDEflatcluster}
	h_{a,w}(x) = \begin{cases}
		\exp\left(\frac{-1}{1-\frac{\|x-a\|^2}{(\pi w)^2}}\right)        & \text{ for } |x-a-\pi w| <\frac{\pi}{10} \text{ and }  |x-a+\pi w| <\frac{\pi}{10} \\[1.5em]
		\exp\left(\frac{-1}{1-\left(-1 + \frac{1}{10 w}\right)^2}\right) & \text{ for } |x - a| \leq \pi w - \frac{\pi}{10}                                   \\[1.5em]
		0                                                                & \text{ elsewhere, }
	\end{cases}
\end{align*}
respectively.
Here \(a\in\mathbb{T}\) gives the centre of the cluster and \(0\leq w \leq\frac{1}{2}\) controls the width.
In velocity we consider Gaussians and superpositions thereof.  Finally, we study one non-product distribution (in Figure  \ref{figPDE:E2_TwoClusters_SteepHerding}). This is chosen to mimic two clusters in position, of unequal density and with opposite velocities, namely
\begin{equation*}\tag{IC3}\label{PDEic:twoclusters}
	f_0(x,v) = \mathcal{Z}^{-1}\left[ \frac{1}{3}h_{\pi/2, 0.2}(x) \otimes e^{-\frac{(v+0.2 )^2}{2 \times 0.4}} + \frac{2}{3}h_{3\pi/2, 0.2}(x) \otimes e^{-\frac{(v-1.8)^2}{2 \times 0.4}} \right],
\end{equation*}
where $\mathcal{Z}$ is an appropriate normalising constant.

{\bf Simulation Metrics.}\label{subs:simmetrics}
To monitor convergence of the particle system/PDEs we will use various measures of convergence. These metrics are necessarily expressed differently depending on the system at hand (PDE or particle system),  we give here both formulations.
\begin{enumerate}[label=\textnormal{[C\arabic*]},ref={[C\arabic*]}]
	\item \label{C1}  Average velocity across all particles at each time, namely
	      \begin{equation}\label{Mn}
		      M^N(t) := \frac{1}{N}\sum_{j=1}^N v_t^{j,N},
	      \end{equation}
	      and the average of the above  across realizations,  denoted by \(\bar{M}^N(t)\). In the continuous setting this corresponds to the approximate average velocity $M_1^\app$ (`approximate' because calculated using the numerical solution $f_t^\app$),
	      \[M_1^\app(t) := \frac{\int_{\mathbb{T}}\int_{-L_v}^{L_v} v f_t^\app(x,v) \mathrm{d} x \mathrm{d} v}{\int_{\mathbb{T}}\int_{-L_v}^{L_v} f_t^\app(x,v) \mathrm{d} x \mathrm{d} v}. \]
	      Since the scheme is essentially mass preserving, the denominator will be independent of time to very high accuracy. It is noted here that the same observation holds for all the different quantities below, but we do not repeat this every time.

	\item \label{C2} In the continuous setting, we consider the velocity-variance, defined as \[ \operatorname{Var}(f_t^\app) := \frac{\int_{\mathbb{T}}\int_{-L_v}^{L_v} v^2 f_t^\app(x,v) \mathrm{d} x \mathrm{d} v}{\int_{\mathbb{T}}\int_{-L_v}^{L_v} f_t^\app(x,v) \mathrm{d} x \mathrm{d} v} - (M_1^\app(t))^2. \]

	\item \label{C3}  The $\ell^1$ distance between the uniform distribution and the empirical position distribution defined as
	      \begin{equation}\label{eq:l1metric}
		      \ell_1^N(t,x) = \sum_{j\in J} \left \vert \frac{\#\lbrace x_t^{i,N} :  \eta_j \leq \, x_t^{i,N} < \eta_{j+1}, 1\leq i \leq N \rbrace  }{N} - \frac{\eta_{j+1}-\eta_j}{2\pi} \right\vert,
	      \end{equation}
	      where \(\lbrace \eta_j\rbrace_{j\in J}\) forms a partition of \([0,2\pi]\). In practice this will be an equispaced partition such that \( \eta_{j+1} - \eta_j = 2\pi / 120 , \, \eta_0 = 0\).  If the particles are indeed distributed uniformly, this value will not be exactly zero except in the limit as $N\to \infty$. In figures containing this metric, a grey dashed line will show the $\ell^1$ distance for $N$ uniformly drawn random samples on the torus, where $N$ is the number of particles simulated. As in the velocity, we will also consider an averaged version of the metric. The averaged $\ell^1$ distance is the mean of the $\ell^1$ distance taken across many realisations of the particle system and is denoted \(\bar{\ell}^1\). In the continuous setting, the \(\ell^1\) distance between $f_t^\app$ and the uniform distribution on the torus is given by
	      \[\ell^1(f_t^\app,\mathcal{U}[0,2\pi]) :=  \int_{\mathbb{T}} \left| \int_{-L_v}^{L_v} f_t^\app(x,v)\mathrm{d}v - \frac{1}{2\pi}\right|\mathrm{d}x. \]

\end{enumerate}

The  metrics \ref{C1} and \ref{C2} are more suited to monitor convergence in the velocity variable, while \ref{C3} quantifies any non-uniformity in position distribution, i.e. mixing in space.

\section{Particle system}
\label{sec:ParticleSystem}
In this section we present and  discuss  simulations of  the LS IPS \eqref{parsys1}-\eqref{parsys2}.

{\bf Basic reality checks. }For each $N>0$ fixed, the LS IPS  has a unique  invariant measure.\footnote{For each fixed $N$, existence of the invariant measure is trivial as \eqref{parsys1}-\eqref{parsys2} is a bounded perturbation of a (stable) Ornstein-Uhlenbeck process. The process is, however, hypoelliptic so uniqueness of the measure is not a given. To prove uniqueness one needs to show that the process is irreducible and this can be done with observations completely analogous to those of \cite[Lemma 3.4]{mattinglyGeometricErgodicityHypoElliptic2002} or \cite[Lemma 3.4]{mattinglyErgodicitySDEsApproximations2002}, assuming $G$ and $\varphi$ are smooth and globally Lipschitz and  $\varphi$ is uniformly bounded below.} When $\varphi \equiv 1$ it is easy to check that such a measure is homogeneous in space and has mean zero average velocity (we do this in  Appendix \ref{app:PSMeanZero}) so
one expects to see the particles converging towards a space-homogeneous configuration, with the  average velocity $M^N(t)$ of the particles  flicking between the (metastable) values $\pm 1$ at random times.  Hence (as long as $G$  and $\varphi$ are antisymmetric/symmetric, respectively, which will always be the case in this paper) one expects that  time averages of $M^N(t)$, as well as averages over realizations of $M^N(t)$, i.e. $\bar{M}^N(t)$, will converge to zero. One further expects that away from the mean-field regime (that is, when $\varphi \not\equiv 1$), when analytic checks become difficult, the invariant measure of the LS IPS will still  have mean zero in velocity, hence in particular  that  $\bar{M}^N(t)$ will still converge (in time) to zero, irrespective of the choice of $\varphi$ . This is corroborated by the simulations in Figure \ref{fig:switch}. Note that the particle count in Figure \ref{fig:switch} is low, $N=100$, so switches in average velocity are still relatively frequent. In subsequent figures of this section the particle count is substantially higher, so these switches are not evident, as they will occur on longer time scales than those shown in the figures (we do not show the longer time behaviour in velocity in the experiments of this section because this is not our main concern here).  We observe that all the figures of this section as well as numerical experiments carried out corroborate the intuition that the IPS `thermalises' in position much faster than in velocity.  Convergence in space to the uniform distribution is a more delicate matter, on which we come to comment.

{\bf Space-Mixing (and apparent lack there of).}
In  Figure \ref{fig:ClusterConvergence}, we pick two different interaction functions, first strictly positive but non constant and then with compact support (hence, first in the mean-field regime and then outside it),  fix the value of the noise and control space-mixing under various initial data. Besides the observation that initial data farther from the uniform distribution take longer to converge in space, space-mixing always occurs in this set up.

In Figure \ref{fig:oneclustervarygamma} and Figure \ref{fig:oneclustervarynoise} we fix the initial datum and show parameter choices for which space-mixing does not occur.  More precisely,  in Figure \ref{fig:oneclustervarygamma} we fix the value of the noise and the initial datum and we consider the interaction function \eqref{eq:phiIndicator} for various values of $\gamma$; the noise is fixed at $\sigma=0.25$ and then at $\sigma=0.5$, the initial datum is a cluster centered at 0 with width $\pi/2$ and $\gamma$ is varied so that the interaction radius $R$ is $1/6, 1/3, 1/2,2/3,1$ and $4/3$ of the initial cluster size. In Figure \ref{fig:oneclustervarygamma}, lower interaction radii seem to result in persistence of a slightly non-uniform distribution.  To understand whether this phenomenon reveals actual new behaviour -- that is, non-convergence to the uniform in space distribution --  we increased the particle count and simulation time (to at least $t=10^4$, not shown in the figures) but the lack of mixing persisted (though becoming less evident as particle count increases, coherently with the observation at the end of Note \ref{note:heuristics}). However, Figure \ref{figPDE:E1_VaryRatio_Local_lowNoise} simulates the PDE \eqref{nl} (that is, the $N \rightarrow \infty$ limit) and suggests that this apparent lack of mixing is only present in the particle approximation: that is, with the same parameter choices as in Figure \ref{fig:oneclustervarygamma} the LS PDE does converge to a uniform distribution in space. This is an important difference with respect to the globally scaled model \eqref{nlGarnier}, where we will show that this lack of mixing, observed in \cite{garnierMeanFieldModel2019} in simulations of the GS IPS \eqref{Garnierparsys1}-\eqref{Garnierparsys2}, persists also in the GS PDE. We will come back to this in the next section.
In Figure \ref{fig:oneclustervarynoise} we fix  the interaction radius to the value which caused the most persistent non-uniformity in Figure \ref{fig:oneclustervarygamma}, keep the initial datum as in Figure \ref{fig:oneclustervarygamma} but vary the value for the noise.  For lower values of $\sigma$, the non-uniform space  distribution is persistent, while higher  noise produces fast convergence to the expected space-homogeneous equilibrium.
We note that, except when $\gamma$ is fixed to $\gamma=0.5$ in Figure \ref{fig:oneclustervarygamma}, all simulations of Figure \ref{fig:oneclustervarygamma} and Figure \ref{fig:oneclustervarynoise} are performed for parameter values outside of the stability region  \eqref{eq:localStabilityInequality} of the LS PDE, which is where we expect that space-mixing might fail. And again, to summarise the above, when the interaction radius $R$ is small space-mixing fails; when $R$ is large enough space-mixing always occurs, even if we are outside of the stability region (this relates to the fact that the stability condition \eqref{eq:localStabilityInequality} is only sufficient, not necessary).

\begin{figure}
	\centering
	\includegraphics[width=\linewidth]{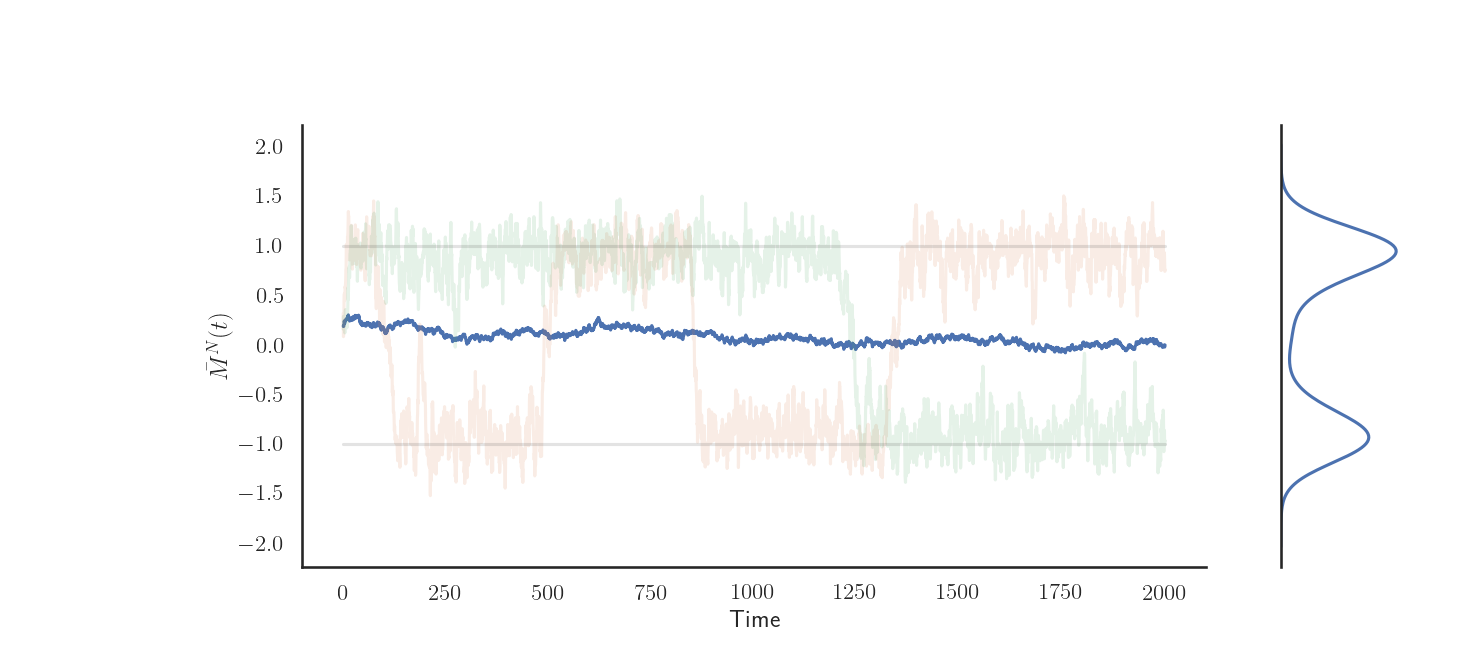}
	\caption{{\bf Bimodality of velocity distribution for the LS IPS \eqref{parsys1}-\eqref{parsys2}, when $\varphi \not\equiv 1$. } We show sample paths of the average velocity $M^N(t)$ (two of them shown in orange and green) and average across realizations $\bar M^N$ (blue line). Also away from the mean-field case the values $\pm 1$ are metastable states for the  single realizations of  $M^N(t)$  and $\bar M^N$  converges to zero.  The above results from simulating the LS IPS  with $N=100, \sigma=1, \Delta t = 0.01$,  IF \eqref{eq:phiIndicator} with $\gamma = 0.1$;   initial conditions  $(x^{i,N}_0, v^{i,N}_0 ) \sim \mathcal U \left[0,2\pi \right] \times \mathcal N(0.2, 2)$ for every  $1\leq i\leq N$.  The grey solid lines are at $\pm1$. The right panel shows the shape of a kernel density estimate at $t=1950$ across 100 realisations, showing expected bimodality in velocity distribution.}
	\label{fig:switch}
\end{figure}

\begin{figure}
	\centering
	\begin{minipage}{\linewidth}
		\centering
		\includegraphics[width=\linewidth]{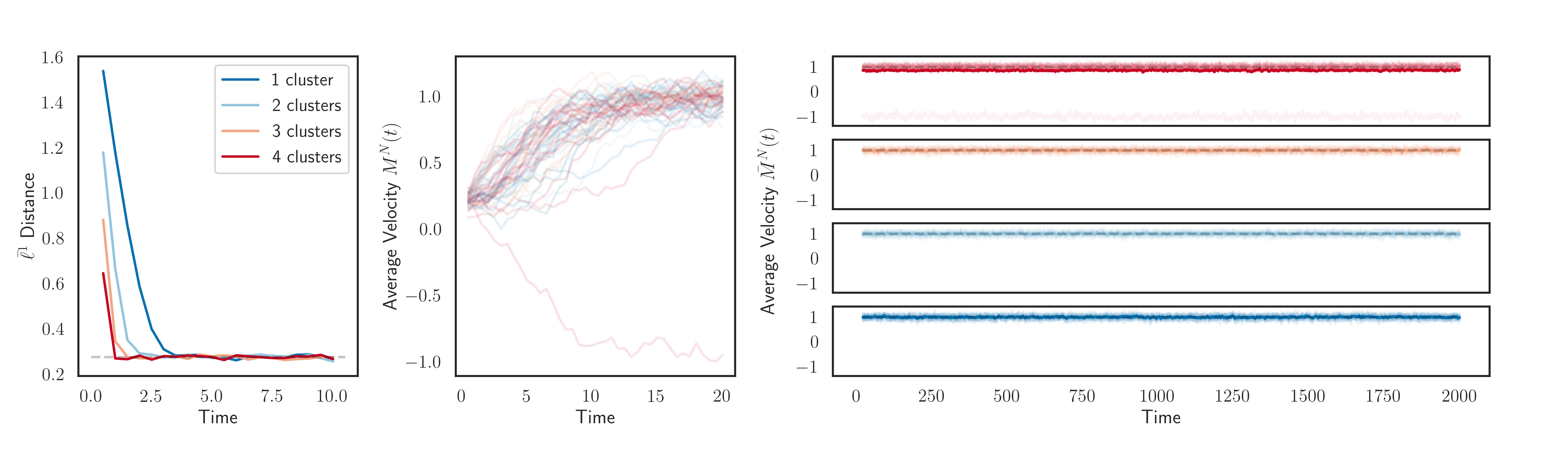}
		\subcaption{\(\varphi\) as in \eqref{eq:phiBump} }
	\end{minipage}\\
	\begin{minipage}{\linewidth}
		\centering
		\includegraphics[width=\linewidth]{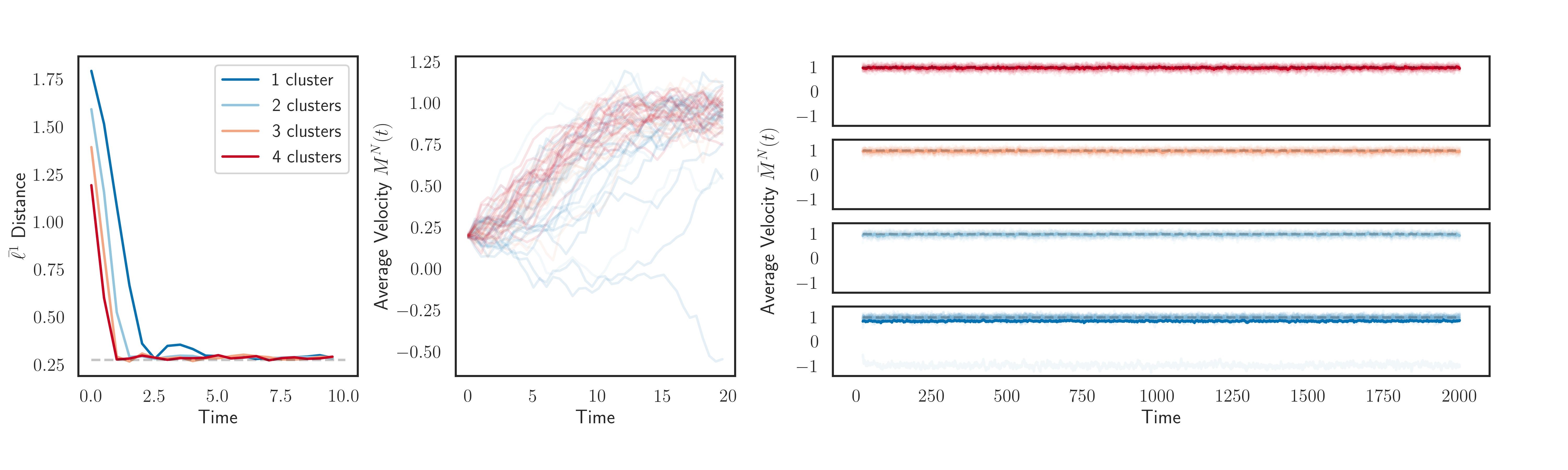}
		\subcaption{\(\varphi^{\gamma} (x), \gamma = 0.05 \)}
	\end{minipage}

	\caption{{\bf LS IPS \eqref{parsys1}-\eqref{parsys2} for various initial data, and two different interaction functions, all other parameters are fixed.} In this parameter regime, irrespective of initial data, space mixing always occurs.  We show 15 realizations for 4 initial conditions of  the  LS IPS \eqref{parsys1}-\eqref{parsys2}, evolved with $\sigma = 1$, $\Delta t = 0.01$; choices of IFs are in the subcaptions, in particular in row (a) the IF is positive everywhere, in row (b) it is  compactly supported.  We plot: in the left column $\bar{\ell}^1$; realizations of average velocity $M^N(t)$ up to time $t=20$ in the middle column; and both $M^N(t)$ and $\bar{M}^N(t)$ up to $t=2000$ in the right column. Different colours  distinguish different initial conditions:  the initial positions are clusters evenly spaced across the torus,  as in \eqref{ic:PSonecluster} and  all particles have deterministic velocity $0.2$. The clusters always break and the empirical distribution of positions converges quickly towards the uniform distribution. Note that the  $\ell^1$ error will never converge exactly to $0$ as we only simulate a finite number of particles, the grey dashed line indicates the expected $\ell^1$ error when drawing 480 uniform random samples on $[0,2\pi]$. As intuitive, the more clusters we start with,  the quicker space-mixing occurs.
	}
	\label{fig:ClusterConvergence}
\end{figure}

{\bf Speed of convergence}. Above we have commented on parameter values in Figure \ref{fig:oneclustervarygamma} and Figure \ref{fig:oneclustervarynoise} for which convergence to the space homogeneous distribution does not occur. Let us now focus on parameter values for which convergence to equilibrium does occur and comment on how the choice of interaction radius $R$ and noise $\sigma$ affects speed of convergence.  As we have already stated, in Figure \ref{fig:oneclustervarygamma}  $\gamma$ is varied so that the interaction radius $R$ is $1/6, 1/3, 1/2,2/3,1$ and $4/3$ of the initial cluster size. For  $R$ large enough, space mixing always occurs; however, while convergence in velocity is monotone in $R$, being fastest in the mean-field case, speed of space mixing is not monotone in $R$. In particular, for lower value of the noise ($\sigma=0.25$), the fastest convergence occurs when the interaction radius is as large as the initial cluster size; for the higher noise value ($\sigma=0.5$) this is no longer the case and the fastest space-mixing occurs when $R$ is larger than the initial cluster size. An analogous fact is observed for the LS PDE, as we will show in the next section.

\begin{figure}
	\centering
	\begin{minipage}{\linewidth}
		\centering
		\includegraphics[width=0.8\linewidth]{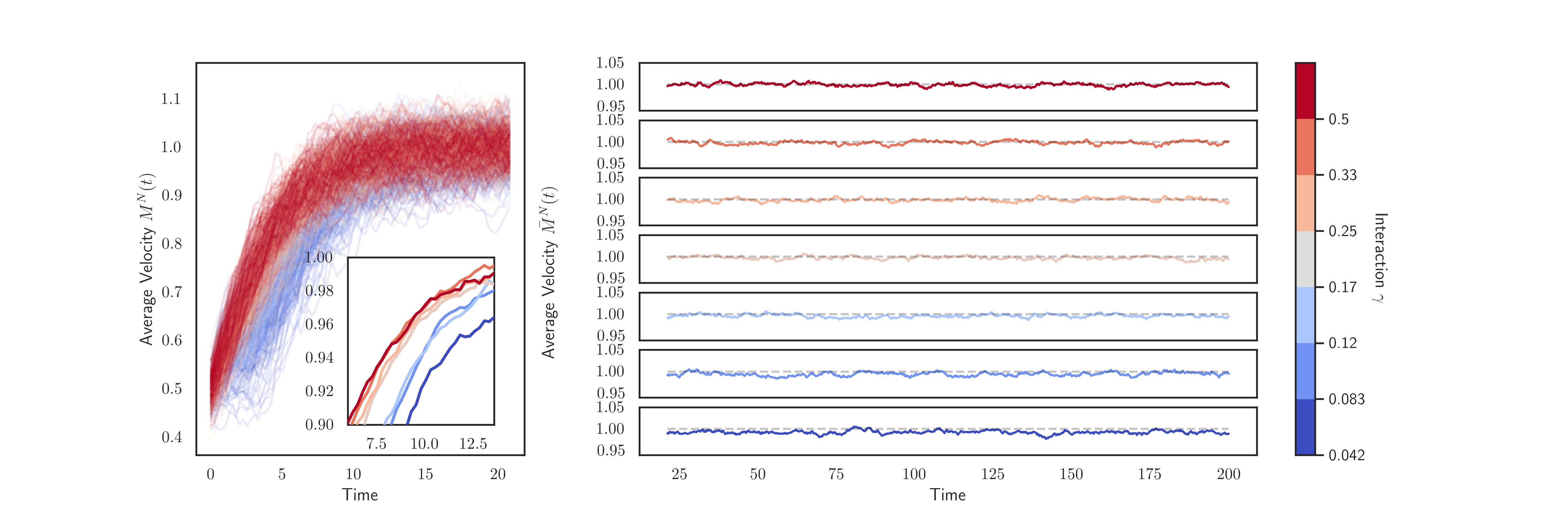}
	\end{minipage}\\
	\begin{minipage}{\linewidth}
		\centering
		\includegraphics[width=0.8\linewidth]{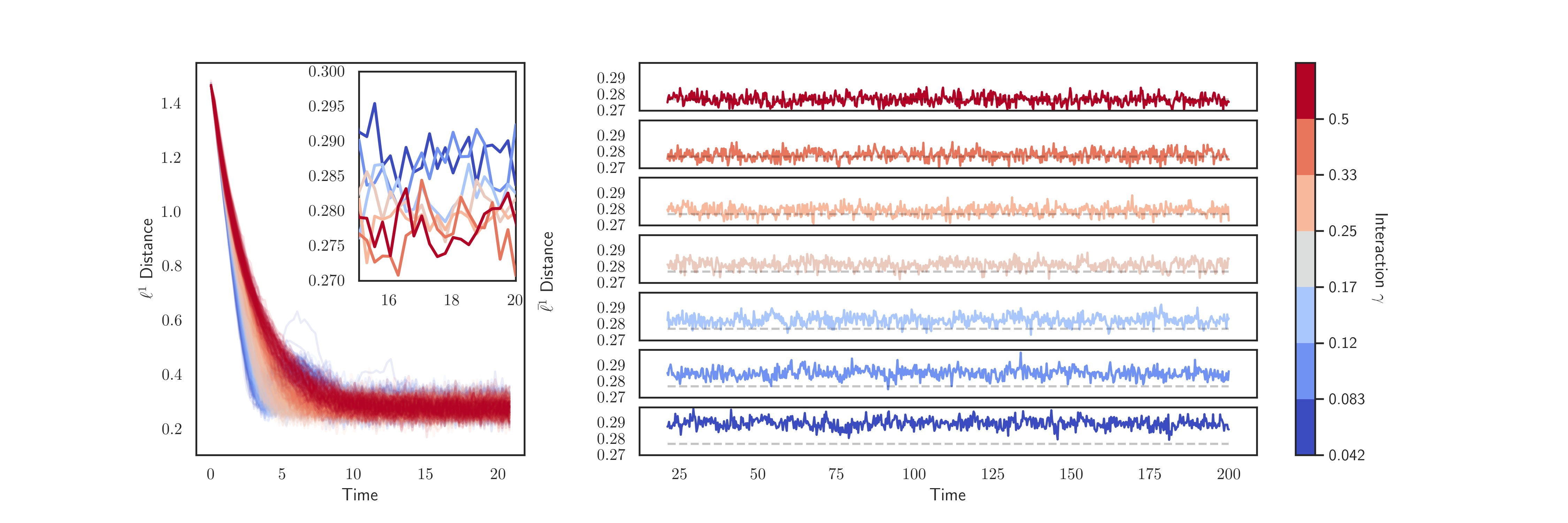}
		\subcaption{$\sigma = 0.25$}
	\end{minipage}\\
	\begin{minipage}{\linewidth}
		\centering
		\includegraphics[width=0.8\linewidth]{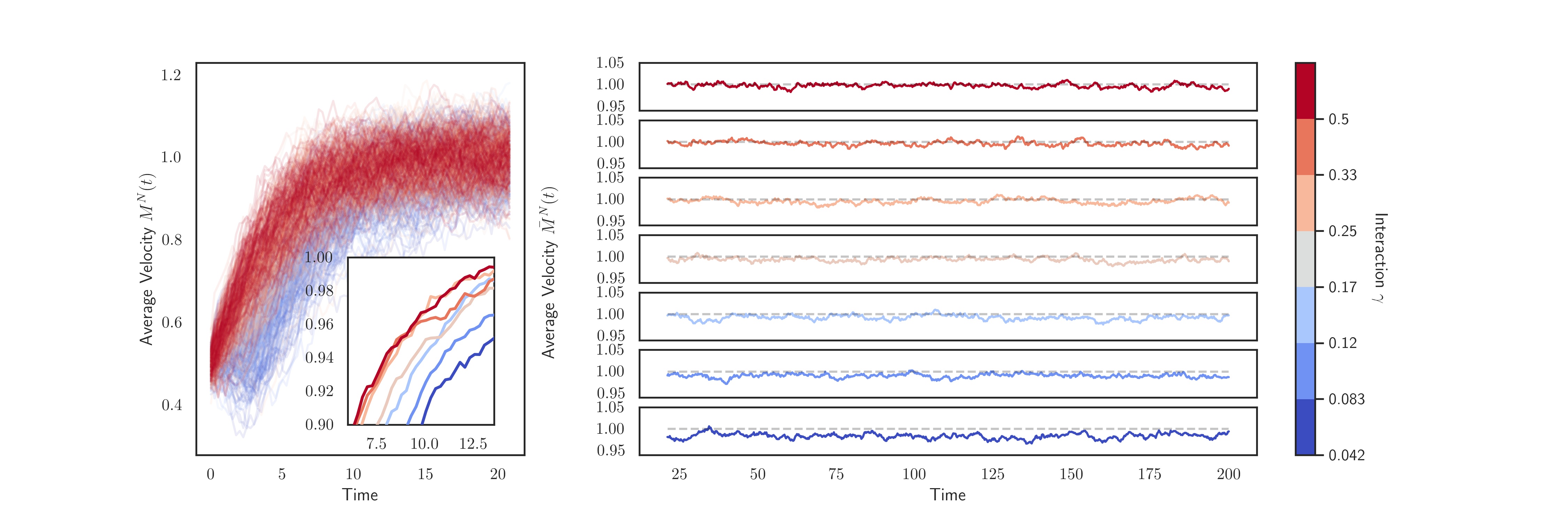}
	\end{minipage}\\
	\begin{minipage}{\linewidth}
		\centering
		\includegraphics[width=0.8\linewidth]{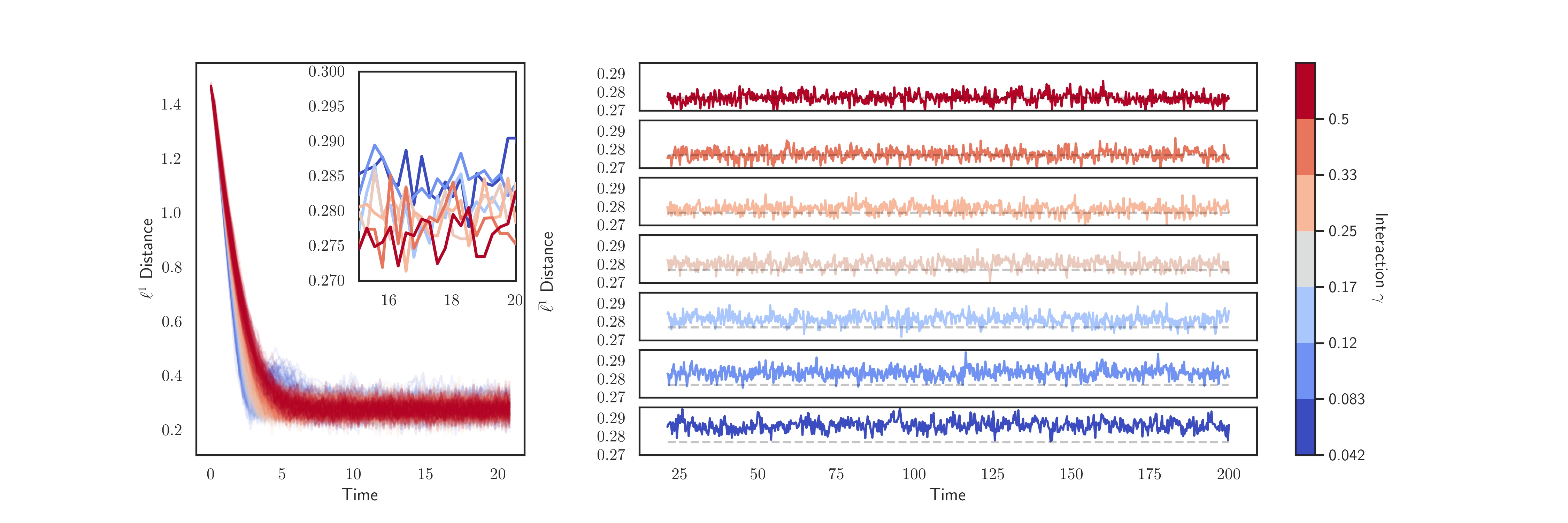}
		\subcaption{$\sigma = 0.5$}
	\end{minipage}
	\captionsetup{width=1.2\linewidth}
	\caption{{\bf LS IPS \eqref{parsys1}-\eqref{parsys2}: effect of interaction radius on space-mixing}. We evolve the LS IPS  with  $N=500$, $\sigma=0.25$, $\Delta t = 0.005$. At time $t=0$ the position distribution is one cluster,  \eqref{ic:PSonecluster} with  $a=0$, $w = \frac{1}{4}$,  and  velocity $\mathcal N(\frac{1}{2}, 0.4)$-distributed. The IF is \eqref{eq:phiIndicator}  and the radius $R= 2\pi \gamma$ of interaction is varied so that it is \(1/6, 1 / 3, 1 / 2, 2 / 3, 1, 4 / 3\) of the support of the initial condition. Mean-field case (\(\gamma = 0.5\)) is also shown. Colour indicates the size of $\gamma$,    as indicated in the colour bar on the right.  The bottom and top row show  behaviour of the $\ell^1$ distance \eqref{eq:l1metric} and of the average velocity, respectively. The left and right columns shows short and long time behaviour, respectively, of such quantities.  In particular, outer boxes on the left column contain plots of realizations, inner boxes contain plots of  average across 100 realisations of the respective quantities. Lower interaction radii lead to more clustered distributions after the initial faster dispersal and, for the lowest $\gamma$ value, the system is not mixing in space (even if the IPS is simulated for longer and for higher particle count, which we do not show here).   Except for $\gamma=0.5$, all other parameter choices fall outside of the stability region \eqref{eq:localStabilityInequality}.
	}
	\label{fig:oneclustervarygamma}
\end{figure}

In Figure \ref{fig:oneclustervarygamma}, for stronger noise levels, we  see the average velocity deviate farther from $+1$, reflecting that the system switches more readily between metastable states in velocity.
We also notice that, irrespective of the choice of $\gamma$, after a quick initial dispersion, the distribution of positions becomes less uniform for a short time (of course this is particularly evident for lower $\gamma$ values). This is due to the fact that at the beginning some particles will have positive velocity  while others will have negative velocity. When those moving clockwise and those moving counterclockwise collide, they form an apparent cluster. After this first collision, they either align velocities and then mix in space (as with higher interaction radii), or they are able to continue without alignment and collide again, causing a second rise in the $\ell^1$ error.

\begin{figure}
	\begin{minipage}{\linewidth}
		\centering
		\includegraphics[width=\linewidth]{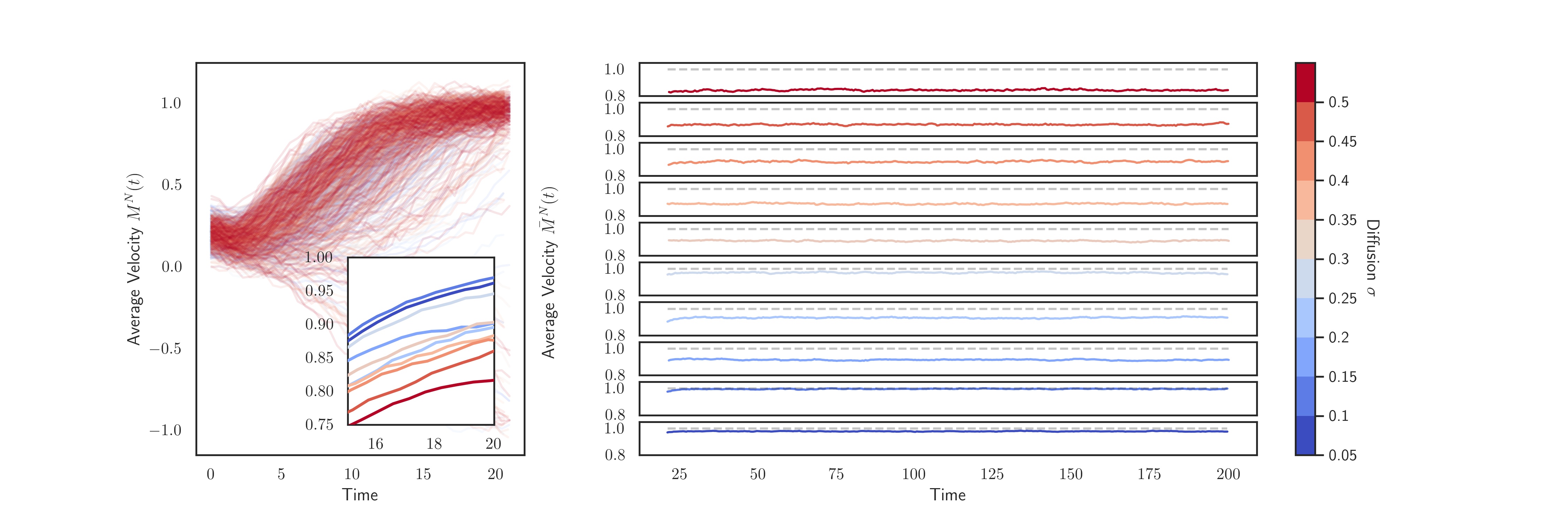}
	\end{minipage}\\
	\begin{minipage}{\linewidth}
		\centering
		\includegraphics[width=\linewidth]{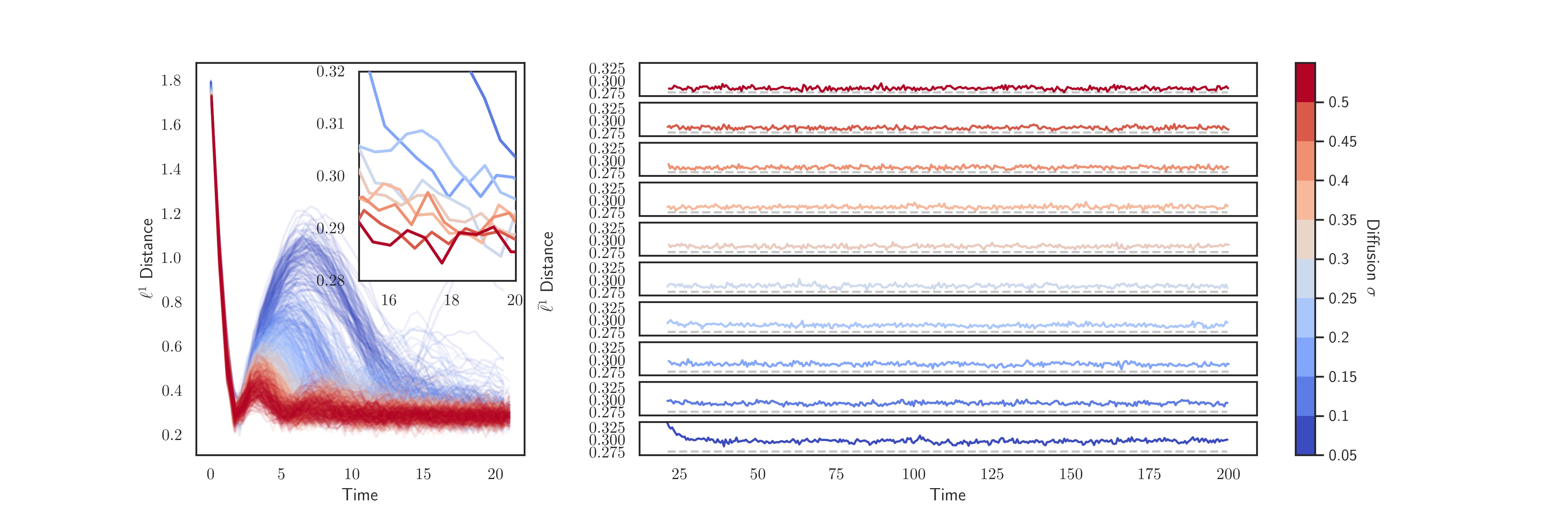}
	\end{minipage}
	\caption{{\bf LS IPS \eqref{parsys1}-\eqref{parsys2}: Effect of varying noise on space-mixing.} We fix the cutoff IF  as in \eqref{eq:phiIndicator} with  $\gamma = 0.05$ (the value which was giving most persistent non space homogeneous behaviour in  Figure \ref{fig:oneclustervarygamma}) and evolve  the dynamics   for values of the noise  $\sigma$, ranging between  0.05 and 0.5.   The initial condition for the particles is $\mathcal{N}(0.2,2)$ in velocity and one cluster in space \eqref{ic:PSonecluster} as in Figure \ref{fig:oneclustervarygamma}. The bottom panel shows the $\ell^1$ distance from uniform distribution for each realisation (left) and averaged across 100 realisations (right). The top panel shows the average velocity $\bar{M}^N(t)$ similarly. The timestep $\Delta t$ is chosen such that $\frac{\Delta t}{\sigma} = \frac{1}{10} $. {No simulations lie in the stability region \eqref{eq:localStabilityInequality}}.}
	\label{fig:oneclustervarynoise}
\end{figure}

\section{Simulation of the LS PDE \texorpdfstring{\eqref{nl}-\eqref{Mnl}}{of the PDE}} \label{sec:TimeDependentPDESim}
In this section we compare the behaviour of the LS PDE  \eqref{nl} with the behaviour of its approximating particle system \eqref{parsys1}-\eqref{parsys2} and with the behaviour of the GS PDE \eqref{nlGarnier}.

\begin{figure}
	\centering
	\begin{minipage}{0.4\linewidth}
		\centering
		\includegraphics[width=\linewidth]{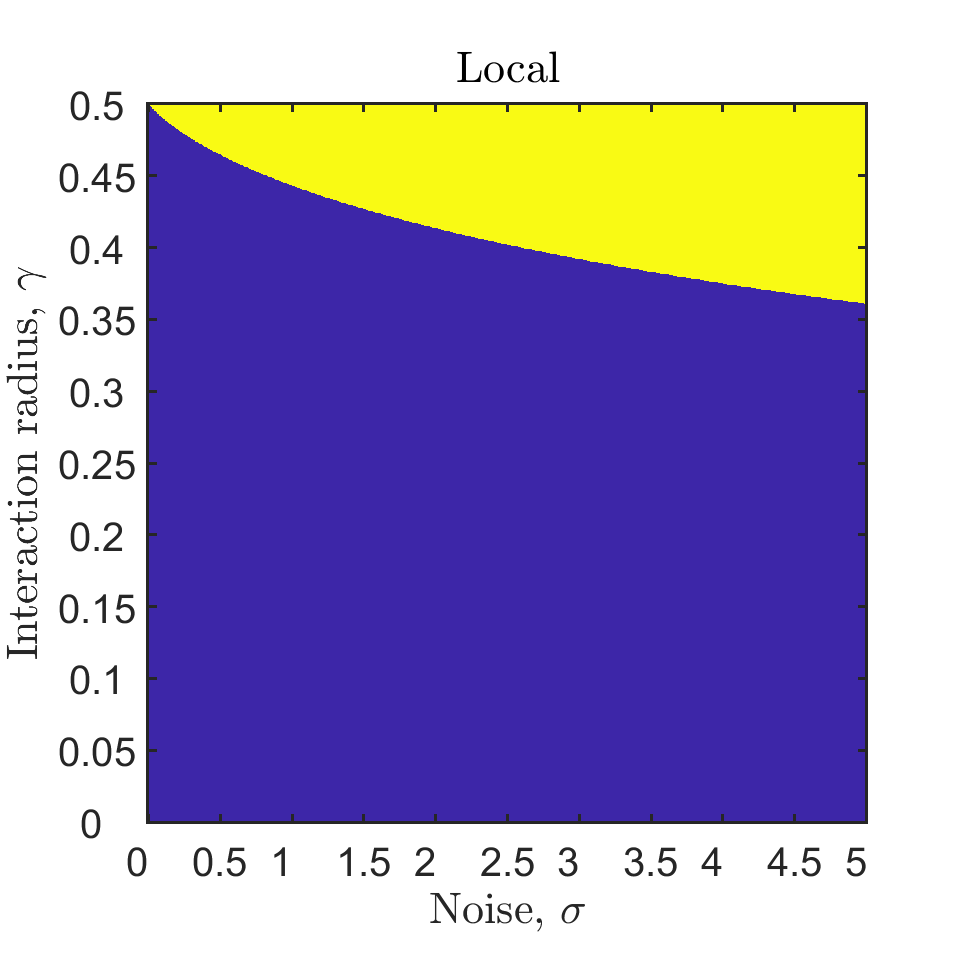}
	\end{minipage}%
	\begin{minipage}{0.4\linewidth}
		\centering
		\includegraphics[width=\linewidth]{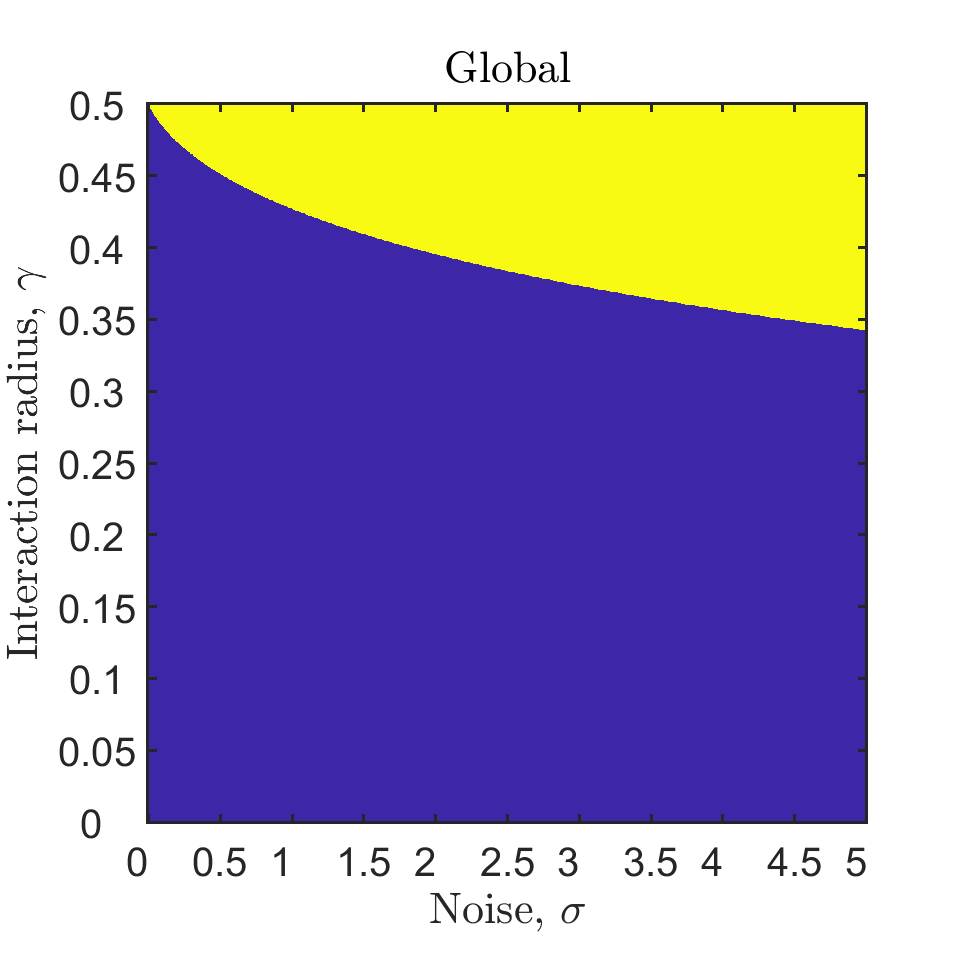}
	\end{minipage}
	\caption{{\bf Stability region for the LS and GS PDE.} For the IF \eqref{eq:phiIndicator} and herding function \(G\) \eqref{eq:Gsmooth} with \(\alpha = 1\), the stability condition for the LS PDE \eqref{eq:localStabilityInequality} (left) and GS PDE \eqref{eq:globalStabilityInequality} (right) are calculated for a variety of interaction radii and noise levels (at \(\xi=1\) with \(L=2\pi\)). Yellow indicates that the inequality is satisfied, thus the stationary state is stable, while blue indicates that the inequality is not satisfied, i.e. in the blue region the model could be either linearly stable or not.  }
	\label{figPDE:StabilityRegion}
\end{figure}

\subsection{\texorpdfstring{Comparison Between  the LS IPS \eqref{parsys1}-\eqref{parsys2}  and  the LS PDE \eqref{nl}-\eqref{Mnl}}{Comparison Between  Particle System and PDE}}

Figure \ref{figPDE:E1_VaryRatio_Local_lowNoise}  contains simulations of the PDE \eqref{nl}, showing the effect of changing the interaction radius for different levels of noise.  The parameter choices in Figure \ref{figPDE:E1_VaryRatio_Local_lowNoise}  are the same as those  in  Figure \ref{fig:oneclustervarygamma}, so that Figure \ref{figPDE:E1_VaryRatio_Local_lowNoise} can be viewed as the `PDE counterpart' of the particle simulations in Figure \ref{fig:oneclustervarygamma}. As opposed to what happens in the LS IPS, in  the LS PDE no non-uniformity is apparent in the position distribution, even at very low interaction radii. This is the only significant difference we notice between the LS IPS and the LS PDE. This is striking in view of the following observation: the simulations of the GS IPS produced in \cite{garnierMeanFieldModel2019} show that, outside stability region \eqref{eq:globalStabilityInequality}, the GS IPS may be unstable. In Figure \ref{figPDE:E1_VaryRatio_Global_lowNoise} we will simulate the GS PDE model for parameter values outside of its stability region and observe that the instability of the equilibria observed in particle simulations persists in the GS PDE (we will come back to this).  In Figure \ref{figPDE:E1_VaryRatio_Local_lowNoise} we simulate the LS PDE for various values of $\gamma$ and $\sigma$; all of these values (with the exception of $\gamma=0.5$) fall outside of the stability region. The uniform space distribution is non-stable in the particle system for such values, despite the invariant measure $\mu_+$ (similarly for $\mu_-$) always being stable for the LS PDE.

Aside from this difference, in the parameter regime in which space-mixing occurs the dynamics in Figure \ref{figPDE:E1_VaryRatio_Local_lowNoise} broadly agree with those of the IPS in Figure \ref{fig:oneclustervarygamma}. Indeed, regarding speed of convergence to equilibrium, we observe the following: speed of convergence to the uniform space- distribution in \(\ell^1\) distance  does not depend monotonically on the interaction radius;  instead there exists an optimum value, which seems to depend on the strength of the noise $\sigma$. As for the particle system, at lower noise levels ($\sigma=0.25$) the fastest convergence in \(\ell^1\) distance occurs when the radius of interaction is equal to the  width of the initial cluster (the simulation is started with an initial condition which is compactly supported in $x$, i.e. a `cluster' in position), i.e. for $\gamma=0.25$. On the contrary,  convergence in average velocity and velocity variance is monotonic in \(\gamma\), the fastest being the mean-field regime, irrespective of the value of $\sigma$.

\begin{figure}
	\makebox[\textwidth][c]{\includegraphics[width=1.1\linewidth]{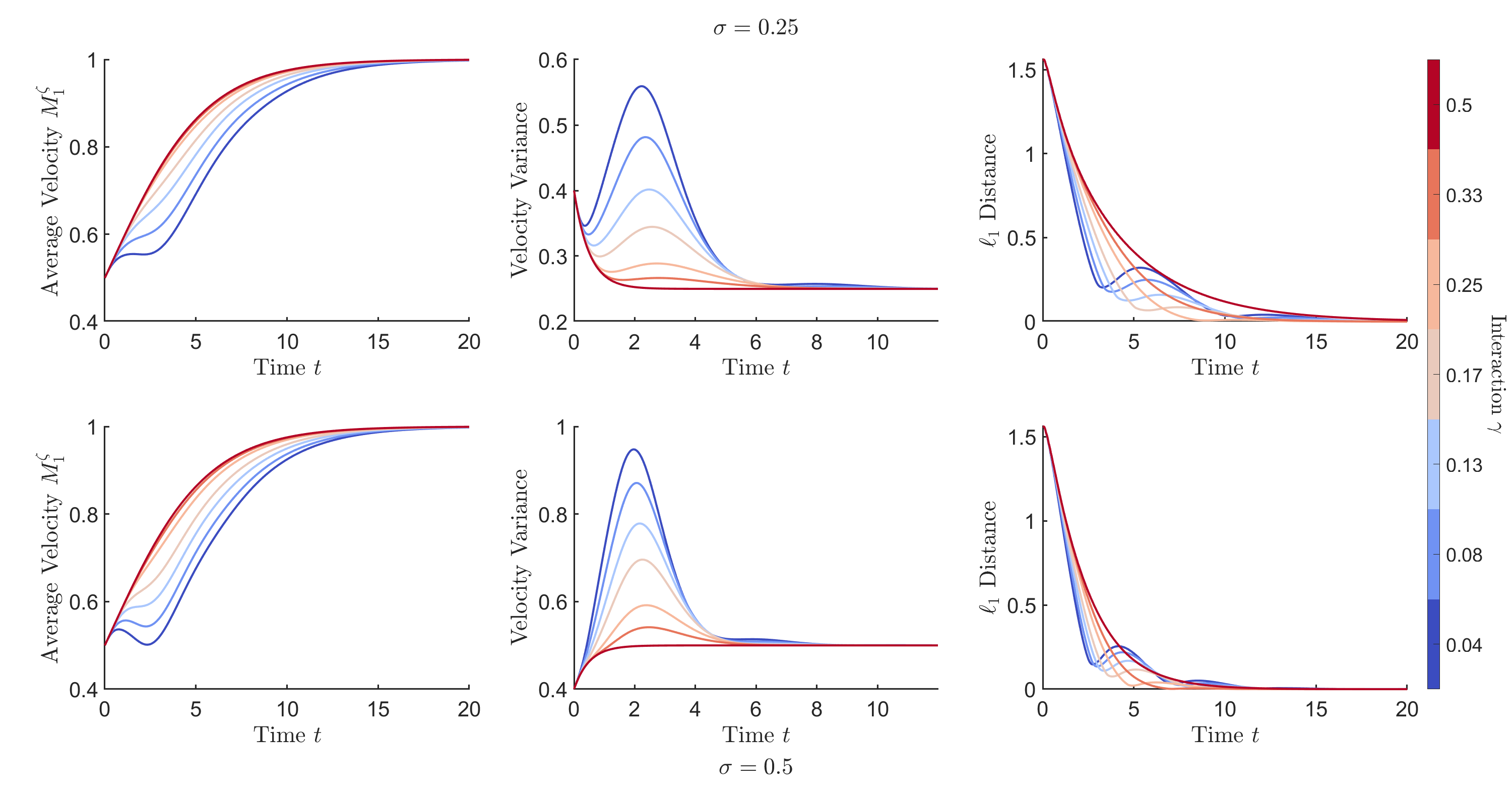}}
	\caption{{\bf LS PDE \eqref{nl}: effect of interaction radius on space-mixing.} All parameters as in Figure \ref{fig:oneclustervarygamma}: we simulate the LS PDE with IF \eqref{eq:phiIndicator}, initial distribution   \(\mathcal{N}(\frac{1}{2},0.4)\) in velocity and one cluster in space,  \eqref{ic:PDEflatcluster} with \(a=0, w = \frac{1}{4}\). The speed of convergence (in the sense of average velocity (\ref{C1}, left column), velocity variance (\ref{C2}, middle column) and \(\ell^1\) distance (\ref{C3}, right column) is shown. The interaction radii are chosen as in Figure \ref{fig:oneclustervarygamma}. Time axes are truncated after convergence to emphasise earlier dynamics. All simulations are outside of the stability region, except when $\gamma=0.5$.}
	\label{figPDE:E1_VaryRatio_Local_lowNoise}
\end{figure}

\subsection{Comparison Between  GS PDE and LS PDE}
In Figure \ref{figPDE:E1_VaryRatio_Global_lowNoise} we repeat the  experiments of Figure \ref{figPDE:E1_VaryRatio_Local_lowNoise},   this time simulating the GS PDE \eqref{nlGarnier}.
All the parameter regimes of  Figure \ref{figPDE:E1_VaryRatio_Global_lowNoise} fall outside of  the stability region for the GS PDE \eqref{eq:globalStabilityInequality}, except when $\gamma=0.5$.

{\bf Travelling wave solutions. } As we have already observed, outside of its stability region, the LS PDE remains stable, in the sense that  the dynamics always converges to either $\mu_{+}$ or $\mu_{-}$. As we can see from Figure \ref{figPDE:E1_VaryRatio_Global_lowNoise}, this is not the case for the GS PDE:  outside of  the  stability region \eqref{eq:globalStabilityInequality}, solutions of the GS PDE may not converge to  $\mu_{\pm}$.  Indeed in Figure \ref{figPDE:E1_VaryRatio_Global_lowNoise} we can see that, for small interaction radii,   the mean and variance of the dynamics do not converge to the mean and variance of either ($\mu_0$ or) $\mu_{\pm}$ and  the space distribution never converges to the space homogeneous configuration. It is important to recall that it is known analytically that $\mu_0$ and $\mu_{\pm}$ are the only stationary states of the GS PDE \eqref{nlGarnier} and,  because the solution of \eqref{nlGarnier} is tight, one expects that it should still converge somewhere.

In Figure \ref{figPDE:E3_marginal_waterfalls} we further investigate the GS PDE and observe  that, as for the GS IPS,  the instability of the  equilibria $\mu_{\pm}$ is due to an  instability of the space homogeneous distribution, and results in the formation of a travelling wave. Such a travelling wave is observed to be stable.  Indeed,  in Figure \ref{figPDE:E3_marginal_waterfalls} we simulate both the LS PDE (left column) and the GS PDE (right column);  in the bottom row we choose an initial datum which is stationary in velocity and a small perturbation of the uniform distribution in space, i.e. we start `very close' to the stationary state $\mu_+$. With this initial condition (and for parameter choices which are outside of the stability region of both models), we find that the initial small space inhomogeneity is amplified by the GS dynamics, which converges to a travelling wave. For the LS evolution, rather, the initial space inhomogeneity is forgotten and the solution converges to $\mu_+$. In the top row of Figure \ref{figPDE:E3_marginal_waterfalls}, the same experiment is repeated, this time starting with an initial condition further away from $\mu_+$, and again the GS PDE converges to a travelling wave (which, from simulation data,  seems to coincide with  the previously observed one), while the LS PDE converges to $\mu_+$.
Figure \ref{figPDE:E3_marginals} is produced with the same setup of the top row of Figure \ref{figPDE:E3_marginal_waterfalls} -- in particular with an initial distribution which is symmetric in position but bimodal in velocity, so that some `particles' start with positive velocity and some with negative velocity -- but allows a different visualization of the travelling wave, which informs the intuitive understanding and the heuristics which we have described in Note \ref{note:heuristics}.  In Figure \ref{figPDE:E3_marginals} we observe that under the GS PDE dynamics,   the initially symmetric position-cluster moves quicker than under the LS PDE,  whilst leaving behind a  `tail' of `slower particles' and hence losing immediately the initial symmetry;  under the LS PDE the position distribution  remains instead symmetric for all times and converges quickly to the uniform distribution.

\begin{figure}
	\makebox[\textwidth][c]{\includegraphics[width=1.2\linewidth]{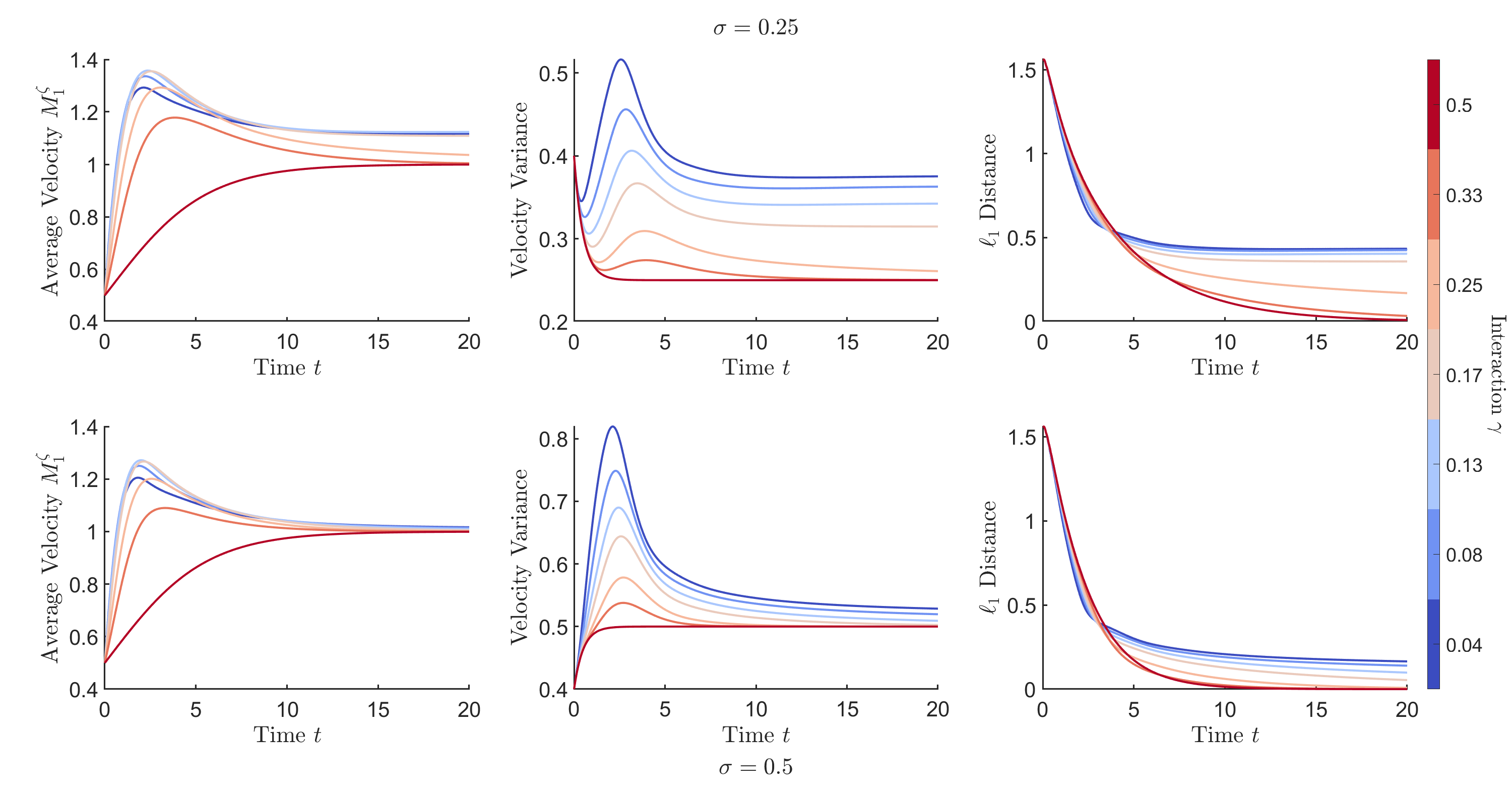}}
	\caption{{\bf GS PDE \eqref{nlGarnier}: effect of interaction radius on space-mixing.} We repeat the experiment of Figure \ref{figPDE:E1_VaryRatio_Local_lowNoise} for the dynamics \eqref{nlGarnier}-\eqref{MnlGarnier}. All  parameters are as described in Figure \ref{figPDE:E1_VaryRatio_Local_lowNoise}.  For low values of $\gamma$ the dynamics does not converge to  $\mu_{\pm}$ and the instability persists for long times (tested up to \(t=10000\), not shown). Except for $\gamma=0.5$, none of the parameter values are in the stability region \eqref{eq:globalStabilityInequality}.}
	\label{figPDE:E1_VaryRatio_Global_lowNoise}
\end{figure}

\begin{figure}
	\centering
	\begin{minipage}{0.5\linewidth}
		\centering
		\includegraphics[width=\linewidth]{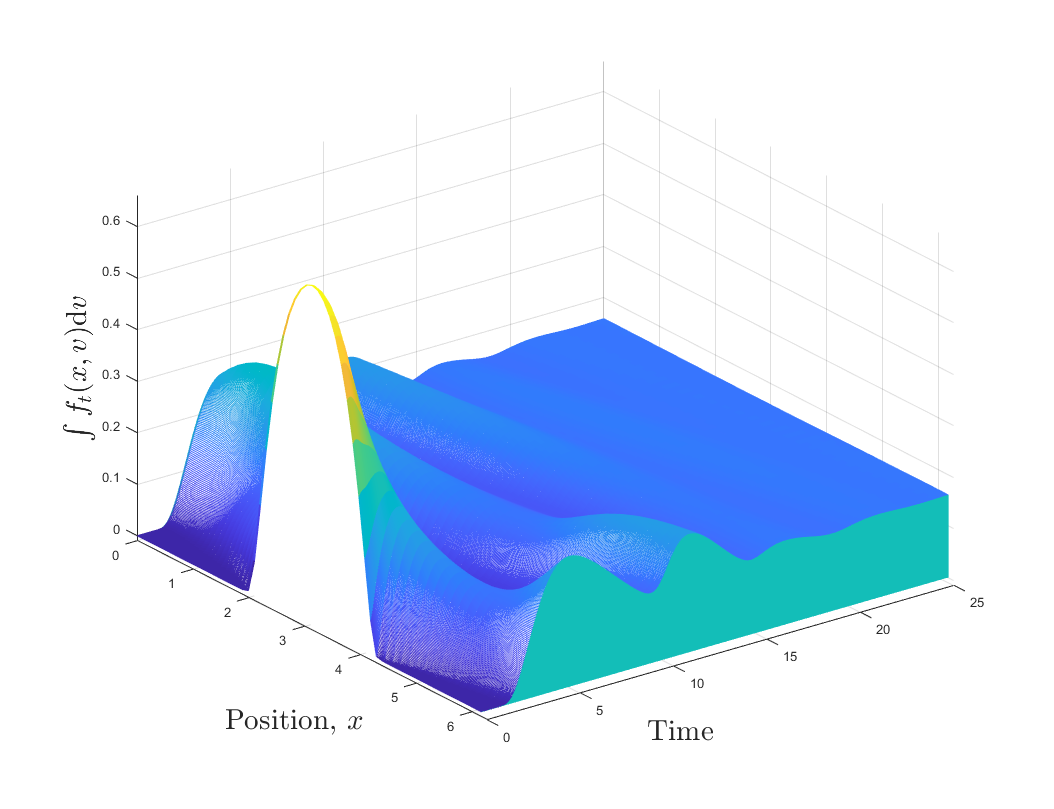}
		\subcaption{ }
	\end{minipage}%
	\begin{minipage}{0.5\linewidth}
		\centering
		\includegraphics[width=\linewidth]{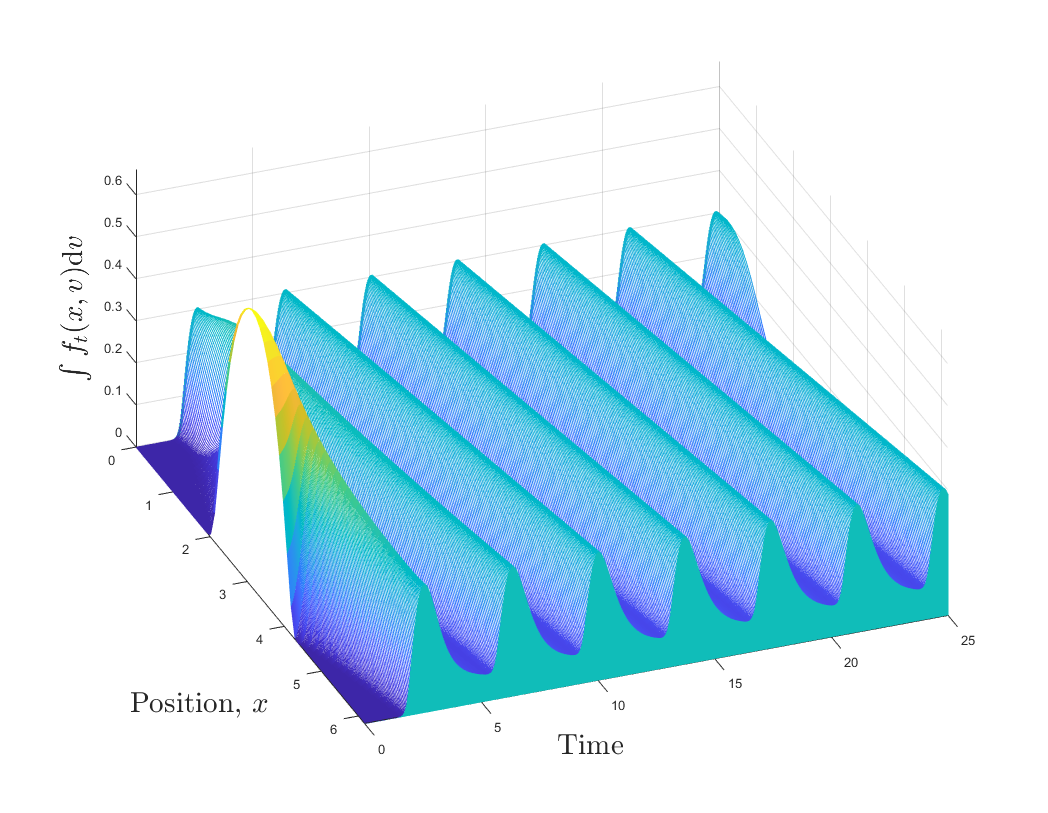}
		\subcaption{ }
	\end{minipage}\\
	\centering
	\begin{minipage}{0.5\linewidth}
		\centering
		\includegraphics[width=\linewidth]{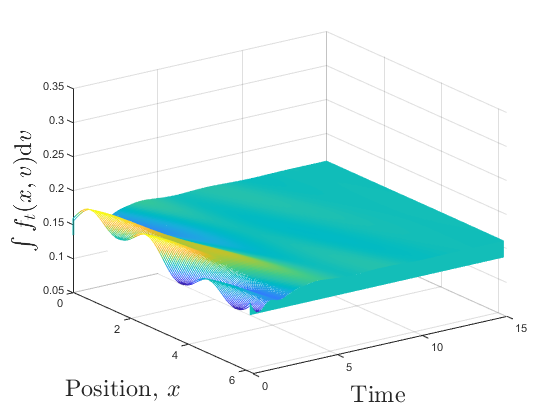}
		\subcaption{ }
	\end{minipage}%
	\begin{minipage}{0.5\linewidth}
		\centering
		\includegraphics[width=\linewidth]{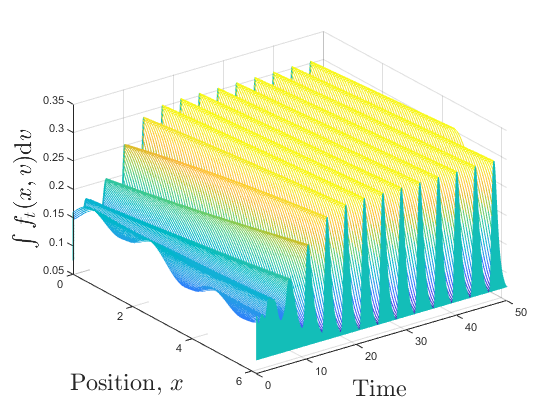}
		\subcaption{ }
	\end{minipage}
	\caption{{\bf Comparison between the GS and the LS PDE.} The dynamics are markedly different for the local scaling (left) compared to the global (right). The plots show the position marginal,  \(\int f_t(x,v) \, \dif v \),  evolved over time. The interaction function is given by \eqref{eq:phiIndicator} with \(\gamma = 0.01\) and \(\sigma=0.25\). We thus  lie outside of the stability region of both dynamics. Top and bottom row differ for choice of initial datum: in the top row the initial datum is a cluster,  given by \eqref{ic:PDEbump}  with $a=\pi, w=\frac{1}{5}$,   while the velocity distribution is a mixture of two Gaussians  \(\mathcal{N}(-0.4,0.3^2)\) and \(\mathcal{N}(0.6,0.3^2) \). In the bottom row, the position distribution is a perturbation around the uniform distribution, \(h(x) = (1+0.1\sin(x)+0.1\sin(3x))/2\pi\), while the velocity is \(\mathcal{N}(1,0.25)\) i.e. the velocity is in stationarity.}
	\label{figPDE:E3_marginal_waterfalls}
\end{figure}

{\bf Basin of attraction. }  In the space-homogeneous case (see \eqref{lin}-\eqref{meanasymp}), the sign of the initial average velocity determines the basin of attraction, i.e. if the initial configuration has positive/negative average velocity then the dynamics converges to $\mu_+$/$\mu_-$, respectively, see \eqref{meanasymp} and \cite{buttaNonlinearKineticModel2019}. Figure \ref{figPDE:E2_TwoClusters_SteepHerding} (left column) shows that, as expected,  this is not the case for the space-dependent case. Indeed, when initial configuration and parameters are appropriately chosen, initial data with positive average velocity can converge to $\mu_-$. This is true for both the GS and the LS PDE. Figure \ref{figPDE:E2_TwoClusters_SteepHerding} also shows that the basins of attraction of such models do not coincide.

{\bf Speed of convergence. }  We look again at Figure \ref{figPDE:E1_VaryRatio_Local_lowNoise} and Figure \ref{figPDE:E1_VaryRatio_Global_lowNoise}, this time focusing on the parameter values when convergence to the steady state does occur.  In the LS PDE the speed of convergence in average velocity and in velocity variance are monotone in $\gamma$, with the fastest case being the mean-field case,  while such a monotonicity is lost for the $\ell^1$ distance; in the GS PDE speed of convergence is instead observed to be monotone in $\gamma$ when measured using velocity variance and $\ell^1$ distance, and such a monotonicity is lost when we look at the average velocity instead. These features persisted for different values of $\sigma$ (not shown here).

\begin{figure}
	\makebox[\textwidth][c]{\includegraphics[width=1.2\linewidth]{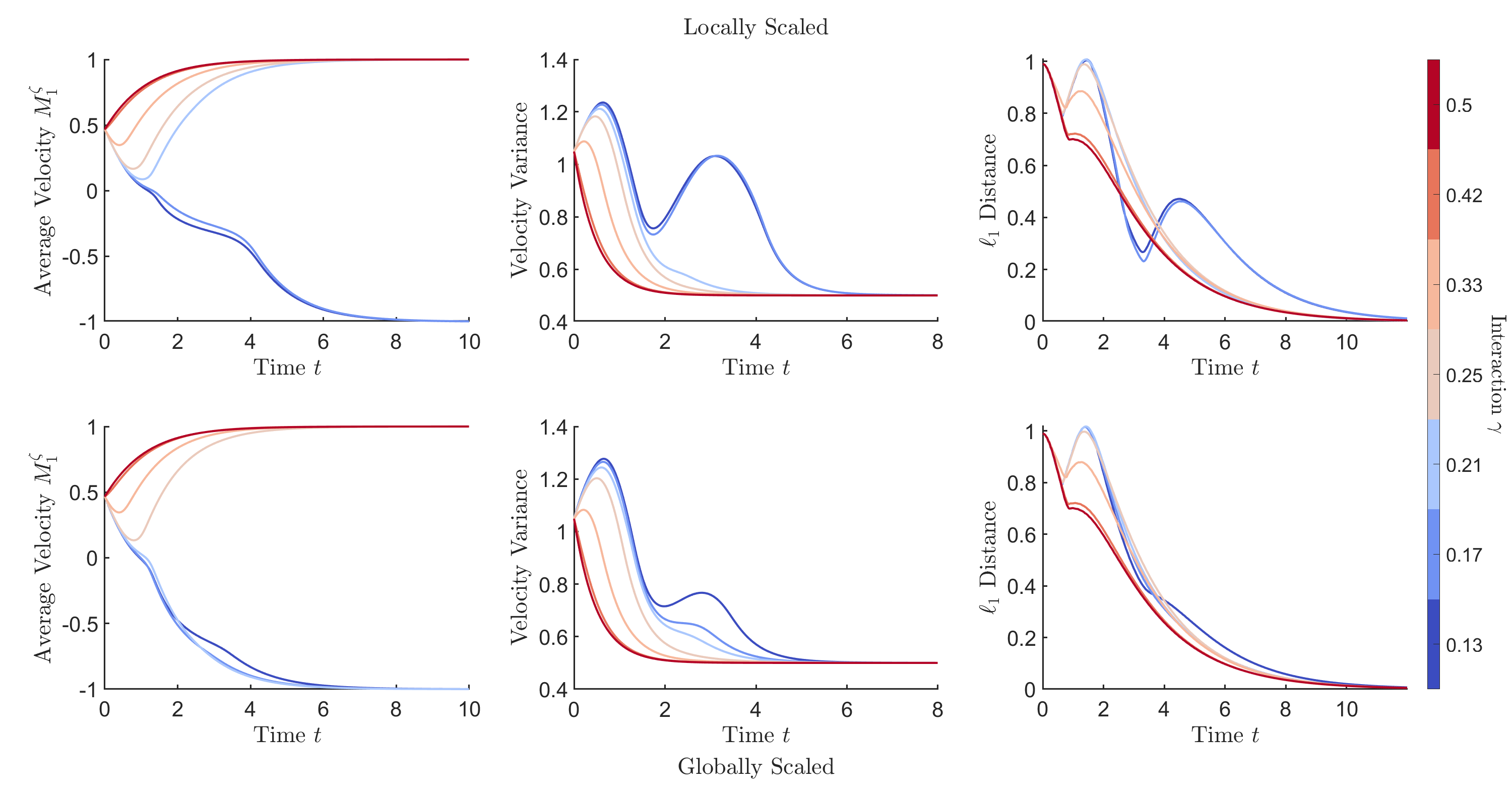}}
	\caption{{\bf GS PDE (bottom row) and LS PDE (top row): basin of attraction.} The basin of attraction of $\mu_{\pm}$ can differ  depending on the model, and is no longer governed by the initial average velocity like in the space homogeneous case. The herding function is \eqref{eq:Gsmooth} with \(\alpha = 10\). The top row shows \ref{C1}, \ref{C2}, \ref{C3} for the dynamics \eqref{nl} while the bottom row shows the same metrics for \eqref{nlGarnier}, both with \(\sigma = 0.5\).  The initial condition  is two clusters of uneven density travelling in opposite directions, \eqref{PDEic:twoclusters} }
	\label{figPDE:E2_TwoClusters_SteepHerding}
\end{figure}

\begin{figure}
	\centering
	\begin{minipage}{0.5\linewidth}
		\centering
		\includegraphics[width=\linewidth]{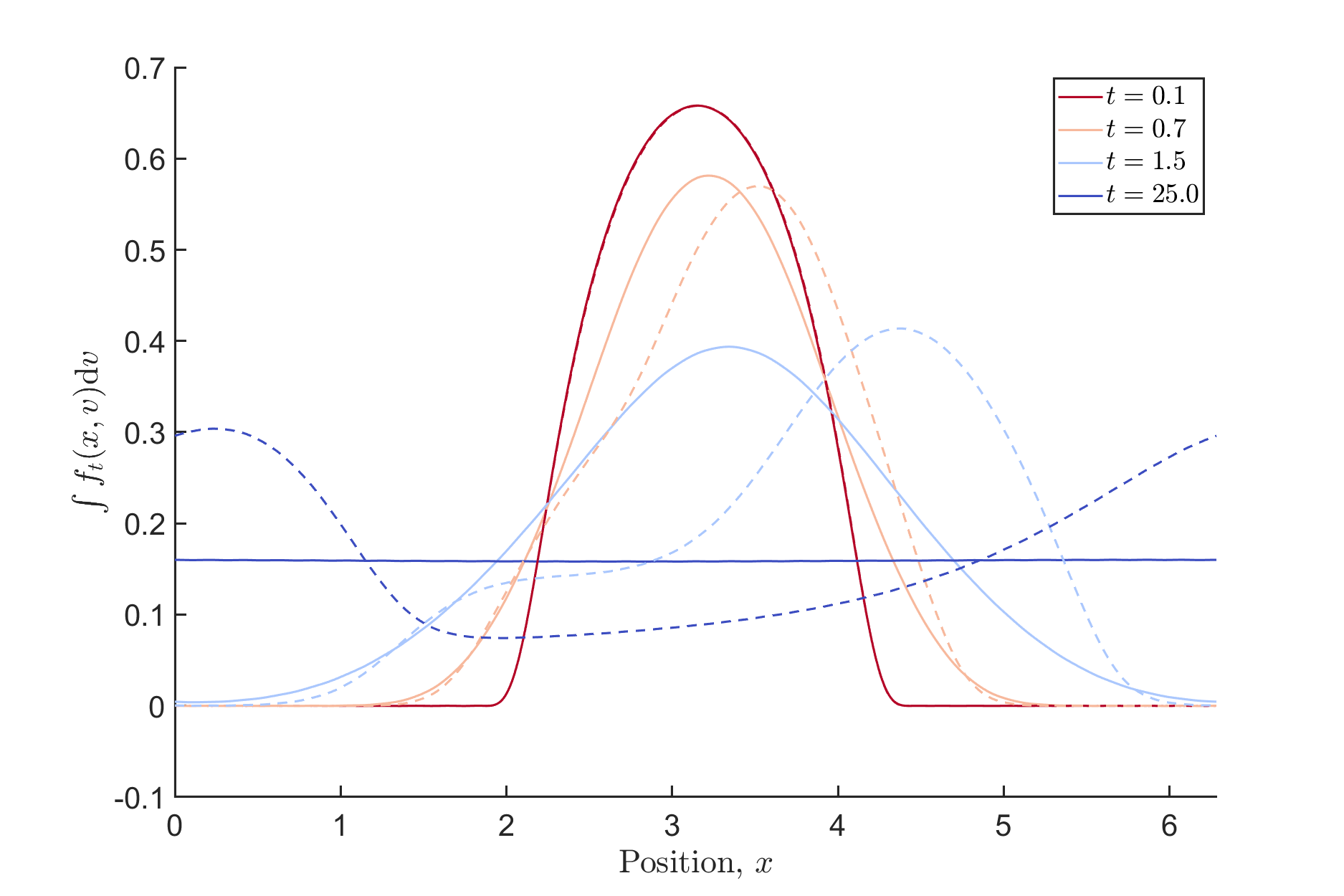}
		\subcaption{ }
	\end{minipage}%
	\begin{minipage}{0.5\linewidth}
		\centering
		\includegraphics[width=\linewidth]{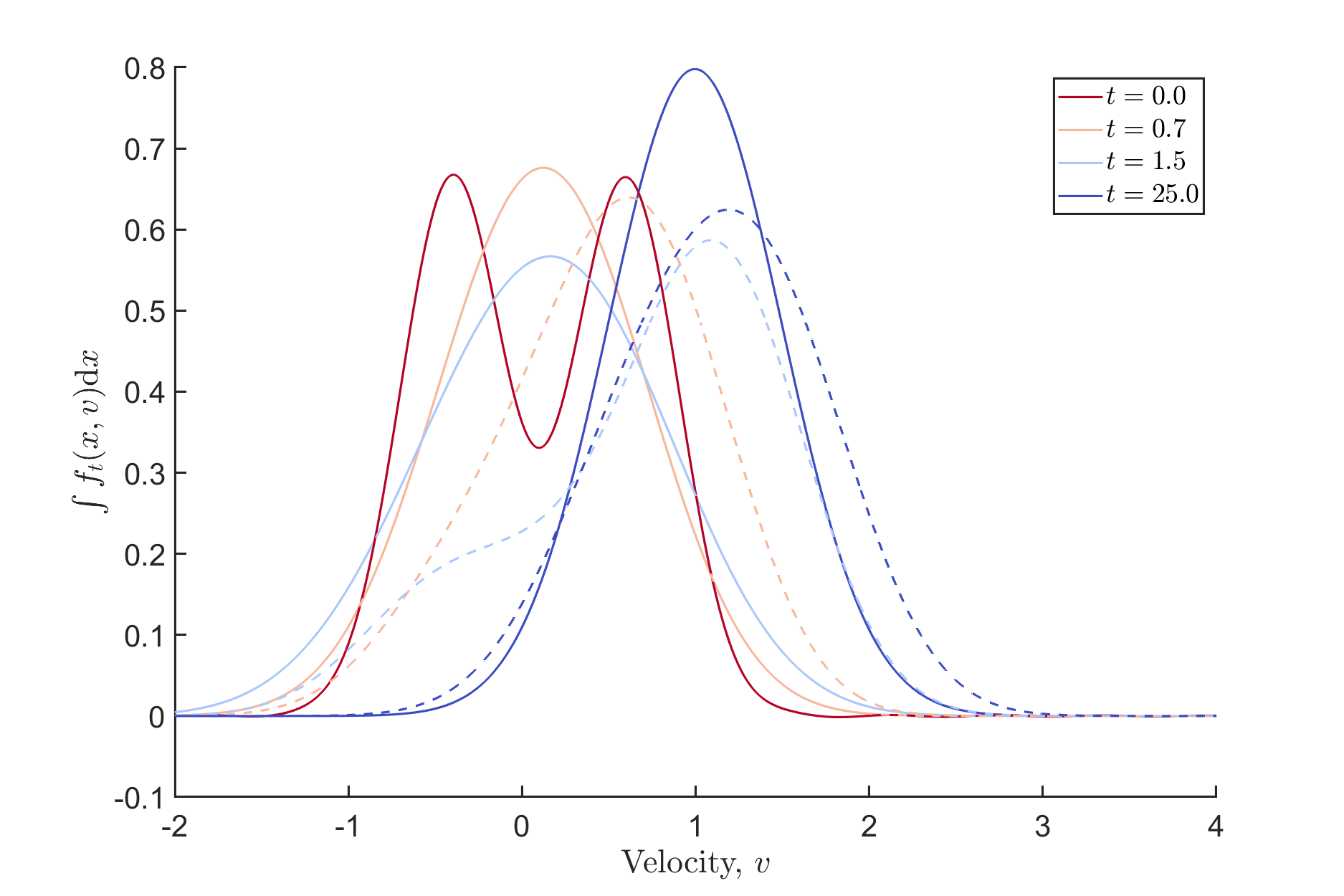}
		\subcaption{ }
	\end{minipage}
	\caption{{\bf GS and LS PDE: travelling wave versus space-mixing}. Starting from a single cluster in position,  \eqref{ic:PDEbump} with $a=\pi,w=\frac{1}{5}$ and a bimodal velocity distribution,  the dynamics are markedly different for the LS PDE (solid line) compare to the GS PDE (dashed line). The left plot (a) shows the position marginal \(\int f_t(x,v) \, \dif v \), the right plot shows the velocity marginal \(\int f_t(x,v) \, \dif x\) at representative time points. The IF is given by \eqref{eq:phiIndicator} with \(\gamma = 0.01\),  \(\sigma=0.25\). We thus are outside the stability region of both dynamics. The initial velocity distribution is a mixture of two Gaussians one with mean \(-0.4\), the other with \(0.6\), both with  variance \(0.3^2\).}
	\label{figPDE:E3_marginals}
\end{figure}

\section{Linear Stability  Analysis}\label{sec:linearstabilityanalysis}
Linearising the LS PDE \eqref{nl} around the equilibria $\mu_{\pm}$, which we express as
\be\label{eqn:equil}
\mu_{\pm}(x,v) = \mathcal U [0,L] \times  \Fx (v)
= \frac 1L  \Fx (v)\, , \qquad   \qquad \Fx(v) \frac{1}{\sqrt{\pi 2\sigma}}
e^{-\frac{(v-\xi)^2}{2\sigma}}\,,
\ee
where $\xi=\pm 1$, gives the following evolution:
\begin{equation}\label{blablu}
	\pa_t g = - v \pa_x g -\pa_v \left[(\xi - v) g\right]+ \frac{1}{L}G'(\xi) \left[\iint (\xi-w) \varphi(x-y) g(y,w) \dif w\, \dif y\right]\Fx'(v)+ {\sigma}\pa_{vv}g \,.
\end{equation}
We now assume $\varphi(x)=\varphi(\lv x \rv )$ and Fourier-transform the above equation in space. To this end, for every $k \in \mathbb Z$, let us set
$$
	e_k(x) = e^{-2 \pi k i x/L}
$$
and decompose $g$ in the Fourier basis $\{e_k\}_{k \in \mathbb Z}$:
\be\label{expansiong}
g(t, x,v) = \sum_{k \in \mathbb Z} g_k(t,v) e_k(x) \, ,
\ee
where  $g_k(t,v) = \frac{1}{L} \langle g, \bar{e}_k\rangle := \frac 1L \int_0^L
	g(t,x,v) e^{2 \pi k i x/L} \dif x$ (and $\bar{e}_k$ denotes complex conjugate). Then, for each $k \in \mathbb Z$, $g_k(t,v)$ satisfies
\be\label{eqnforgo}
\pa_t g_k = \frac{2 \pi ik}{L} v g_k-
\pa_v \left[ (\xi-v) g_k\right] +  G'(\xi) \varphi_k
\left[ \int_{\R} (\xi -w) g_k(t,w) \dif w\right]
\Fx'(v)+  \sigma \pa_v^2 g_k \,,
\ee
where $\varphi_k$ is the $k$-th Fourier mode of $\varphi$, i.e. $\varphi_k = \frac 1L\langle \varphi, \bar{e}_k\rangle$.
Considering now the Fourier transform of $g_k$ in the velocity variable, i.e.
$$
	\hat g_k(t, \eta):= \int_{\R} g_k(t,v)e^{-i\eta v} \dif v \, ,
$$
we find that $\hat g_k$ satisfies the following equation,
\be\label{eqn:fouriertrkthmode}
\pa_t \hat g_k+ \left( \frac{2\pi k}{L}+\eta \right)\pa_{\eta} \hat g_k=
\left( -i\xi \eta -\sigma\eta^2\right)\hat g_k
+ G'(\xi)\varphi_k \hat F_{\xi}(\eta) \eta \, \omega_k(t)
\ee
where
\be\label{eqn:omegak}
\omega_k(t):= \pa_{\eta}\hat g_k(t,0) + i \xi \hat g_k(t,0) \,.
\ee
Analogously to \cite{garnierMeanFieldModel2019}, we define the norm
$$
	\| \hat g_k(t, \cdot)\|_S:=
	\| \hat g_k(t, \cdot)\|_{L^1}+ \| \hat g_k(t, \cdot)\|_{L^{\infty}}+
	\| \pa_{\eta}\hat g_k(t, \cdot)\|_{L^{\infty}}
$$
and we define our notion of stability based on this norm. That is,  assuming the initial datum is such that
$\| \hat g_k(0, \cdot)\|_S< \infty$, we define the $k$-th order mode to be {\em stable} if
\be\label{stabilitycondition}
\sup_{t \geq 0} \| \hat g_k(t, \cdot)\|_S< \infty \,.
\ee
Noting  that the  equations \eqref{eqn:fouriertrkthmode} are decoupled in $k$,  we say that the system is linearly stable if all the modes are stable.

With this set up, we can now move on to the proof of Proposition \ref{lemma:modestability}.
\begin{proof}[Proof of Proposition \ref{lemma:modestability}]
	The proof of this Lemma follows very closely the strategy of \cite[Section 4]{garnierMeanFieldModel2019}, so we only highlight the main steps. We try to use similar notation to make comparison simpler. To begin with, equation \eqref{eqn:fouriertrkthmode} can be solved explicitly (with a similar method to \cite[Lemma 4.1]{garnierMeanFieldModel2019}) and the solution is given by
	\begin{align}\label{eqn:soleqnkthmode}
		\hat g_k(t, \eta) & =e^{-\gk(t, \eta)}
		\hat g_k(0, e^{-t}\eta -D_k(1-e^{-t})) \nonumber\\
		                  & +  G'(\xi) \varphi_k \int_0^t e^{-\gk(t-s, \eta)}\omega_k(s) H_{\xi}\left( e^{-(t-s)}\eta-D_k(1-e^{-(t-s)})\right)  \dif s \,,
	\end{align}
	where $H_{\xi}(\eta)=\eta \hat\Fx(\eta) =  \eta \hat{\Fx}(\eta) = \eta  e^{-i\eta \xi}e^{-\eta^2\sigma/2}$, $\omega_k$ has been defined in \eqref{eqn:omegak} and
	$$
		\gk(t,\eta):=\frac{\sigma}{2}(\eta+D_k)^2(1-e^{-2t}) + (i\xi - 2\sigma D_k)(\eta+D_k)(1-e^{-t})+
		(-i\xi D_k+\sigma D_k^2)t \,.
	$$

	{\em Step 1}. Noting that
	$\omega_0(t)=-im_1(t)+i \xi m_0(t)$ where
	$$
		m_0(t)= \int_{\R} g_0(t,v) \dif v, \quad m_1(t)= \int_{\R} v g_0(t,v) \dif v\, ,
	$$
	from \eqref{eqn:soleqnkthmode} with $k=0$, we obtain
	\begin{align*}
		\lv \hat g_0(t, \eta)\rv & \leq \|\hat g_0(0, \cdot)\|_{L^{\infty}}e^{-\sigma \eta^2 (1-e^{-2})/2}                                                  \\
		                         & +[\sup_{t\geq 0} (\lv m_0(t)\rv+\lv m_1(t)\rv)]  \lv G'(\xi) \varphi_0 \rv \left[\int_0^1 \lv H_{\xi}(\eta e^{-s})\rv ds
		+ \|\hat \Fx\|_{L^{\infty}} \lv \eta\rv e^{-\sigma \eta^2 (1-e^{-2})/2}  \right].
	\end{align*}
	Hence the $0-th$ mode is stable iff $\sup_{t\geq 0} (\lv m_0(t)\rv+\lv m_1(t)\rv)< \infty$. Using the equation for $g_0(t,v)$ (namely equation \eqref{eqnforgo} with $k=0$), we find
	$$
		\frac{\dif }{ \dif t}m_0(t) =0, \quad \frac{\dif}{\dif t}m_1(t) = \left(-1+ \varphi_0 G'(\xi)\right)m_1(t)+ \left(-1+  \varphi_0 G'(\xi)\right) \xi \,  m_0(t) \,
	$$
	hence the $0$-th mode is stable iff $ \varphi_0 G'(\xi)<1$.

	{\em Step 2}. We now turn to study the stability of the $k$-th mode for $k\neq 0$. Similarly to what we have done for the $0$-th mode, from the expression \eqref{eqn:soleqnkthmode} it is clear that if $\sup_{t\geq 0}\omega_k(t) <\infty$ then the $k$-th mode is stable. So we only need to study the time behaviour of $\omega_k(t)$. To this end, let us write a more explicit expression for $\omega_k$:
	$$
		\omega_k(t) = \psi_k(t)+\psi_k^*(t)+
		\int_0^t R_k(t-s)\omega_k(s) \dif s+ \int_0^t  R_k^*(t-s)\omega_k(s) \dif s \,
	$$
	where $\psi_k, R_k$ are defined in \cite[eqns (33)and (34)]{garnierMeanFieldModel2019} while
	\begin{align*}
		\psi^*_k(t) & := i \xi e^{-\gk(t,0)} \hat g_k(0,-D_k (1-e^{-t}))              \\
		R^*_k(t)    & :=  G'(\xi)\varphi_k  e^{-\gk (t,0)}H_{\xi}(-D_k(1-e^{-t})) \,.
	\end{align*}
	Because
	$$
		\sup_t \lv \omega(t) \rv \leq
		\sup_t \lv \psi_k(t)\rv +\sup_t \lv \psi_k^*(t)\rv+ \sup_t \lv \omega(t) \rv \int_0^{\infty} (\lv R_k(t)\rv  +\lv R_k^* \rv) \dif t \,,
	$$
	if $\sup_t \lv \psi_k(t)\rv<\infty$ and  $\sup_t \lv \psi_k^*(t)\rv<\infty$, $\int_0^{\infty}(\lv R_k(t)\rv  +\lv R_k^* \rv) dt <1$, then $\omega_t$ is bounded in time. The boundedness of $\psi_k(t)$ is proved in \cite[Lemma 4.2]{garnierMeanFieldModel2019}; the bound on $\psi^*_k(t)$ can be done analogously (and it is indeed quite straightforward just by the definition of $\psi_k^*$). So we turn to the last estimate. It is proved in \cite[Appendix A]{garnierMeanFieldModel2019} that the following estimate holds:
	$$
		\int_0^{\infty}\lv R_k(t)\rv \dif t \leq
		\lv G'(\xi) \varphi_k\rv
		\left( 1+ \frac{3 \sqrt{\pi}}{\sqrt{\sigma} \lv D_k\rv } + \frac{3 \lv \xi\rv}{\lv D_k\rv \sigma} + \frac{e^{-1}}{1+\lv D_k\rv \sigma }\right).
	$$
	One can do a similar computation to estimate the integral of $R_k^*(t)$; to this end first notice that the following equality holds
	$$
		\left\vert e^{-\gk (t,0)}H_{\xi}(-D_k(1-e^{-t})) \right \vert  = \lv D_k \rv (1-e^{-t}) e^{\sigma D_k^2 (1-e^{-t}-t)}.
	$$
	Now we use the following two estimates: we estimate $(1-e^{-t})$ from above by $t$ if $0\leq t \leq 1$ and by 1 if $t>1$; we estimate $e^{\sigma D_k^2 (1-e^{-t}-t)}$ from above by $e^{-\sigma D_k^2  t^2/4}$  if $0\leq t \leq 1$ and by  $e^{\sigma D_k^2  (1-t)} $ for $t>1$. Hence,

	\begin{align*}
		\int_0^{\infty}\lv R_k^*(t)\rv \dif t & \leq
		\lv G'(\xi)\varphi_k\rv\int_0^1 \lv D_k\rv  t e^{-{\sigma}D_k^2\frac{t^2}{4}} \dif t + \lv G'(\xi)\varphi_k\rv \int_1^{\infty} \lv D_k\rv
		e^{\sigma D_k^2(1-t)}\dif t \\
		                                      & \leq \lv G'(\xi)\varphi_k\rv \left(\frac{\sqrt{2\pi}}{\sigma \lv D_k\rv}\sqrt{\frac{2}{\pi }}+ \frac{1}{\sigma \lv D_k\rv}\right) = \lv G'(\xi)\varphi_k\rv \frac{3}{\sigma \lv D_k\rv}.
	\end{align*}
	This concludes the proof.
\end{proof}

\section{Formal calculation of the spectrum}\label{sec:spectrum}
In this section we consider the linearisation of both the GS PDE and LS PDE around the equilibria $\mu_{\pm}$ and we calculate the spectrum of the resulting operators. That is, linearising the GS PDE \eqref{nlGarnier} around the equilibria \eqref{eqn:equil} we obtain the equation
$$
	\pa_tf = -v \pa_x f - \pa_v\left[(\xi-v) f\right] - \frac{1}{L} G'(\xi) \left[\iint \varphi(x-y)f(y,w) w \, \dif y \dif w\right] \Fx'(v) +  \sigma \pa_{vv} f \,,
$$
where $\xi=\pm 1$.
In Section \ref{subs:spectrum GS} 	we look for $\lambda \in \mathbb C$ and functions $f$ (in a space which will be defined and motivated later, see Note \ref{note:spectrumchoiceofspace}) such that $\cL f = \lambda f$, where $\cL$ is the operator on the RHS of the above, i.e. the integro-differential operator which is defined (on smooth functions, which decay fast enough at infinity) as
$$
	\cL f : = -v \pa_x f - \pa_v\left[(\xi-v) f\right] - \frac{1}{L} G'(\xi) \left[\iint \varphi(x-y)f(y,w) w \, \dif y \dif w\right] \Fx'(v) +  \sigma \pa_{vv} f\,.
$$
In Section \ref{subsec:spectrum LS} we look to solve the analogous eigenvalue problem for the operator $\mathcal L_{loc}$	obtained by linearising the LS PDE around $\mu_{\pm}$.
We will show that in both cases if $\lambda \in \mathbb C$ solves an appropriate fixed point problem (FPP) (see Proposition \ref{prop:spectrum} and Proposition \ref{prop:spectrumlocal}) then it is an eigenvalue for $\cL$ ($\ccL$, respectively). For the globally scaled dynamics, we show that for $\sigma$ small enough such a FPP has at least one solution with positive real part, hence proving that the operator $\cL$ has at least one eigenvalue with positive real part, see Proposition \ref{prop:positive-eigenvalues-GLOB}. For the locally scaled dynamics our study of the FPP is not as conclusive as in the globally scaled case, as we resort to approximation (for $\sigma$ small) of the relevant FPP, see Section \ref{subsec:spectrum LS} for details.

\subsection{Spectrum for the globally scaled model}\label{subs:spectrum GS}
As in the previous section, we decompose the function $f$ in the Fourier basis $\{e_k\}_{k \in \mathbb Z}$:
\be\label{expansionf}
f(x,v) = \sum_{k \in \mathbb Z} f_k(v) e_k(x) \, ,
\ee
where again $f_k(v) = \frac 1L\langle f, \bar{e_k}\rangle := \frac 1L \int_0^L
	f(x,v) e^{2 \pi k i x/L} \dif x$ . By imposing $\cL f = \lambda f$, one can see that the equation satisfied by each one of the modes $f_k(v)$ is given by
$$
	\frac{2 \pi i k}{L} v f_k(v) - \pa_v \left[(\xi-v)f_k(v)\right]
	-  G'(\xi)\varphi_k \left(\int \dif v f_k(v) v\right)\Fx'(v) + \sigma \pa_{vv} f_k(v) = \lambda f_k(v) \,,
$$
where we assume that $\varphi$ is such that $\varphi_k \in \mathbb R$ for every $k \in \mathbb Z$ ($\varphi_k$ has been defined just after \eqref{eqnforgo}). Note that the equations for the modes $f_k$ are decoupled. If we Fourier-transform the above equation, we find that the Fourier transform of $f_k$, namely
$$
	\hat f_k(\eta) : = \int_{\R} f_k(v) e^{-i\eta v} \dif v \, ,
$$
satisfies the following equation
$$
	\left(- \frac{2\pi k }{L} - \eta \right)\pa_{\eta} \hat f_k
	- \left( i \xi \eta + {\sigma \eta^2}+\lambda\right)\hat f_k =-\frac{1}{i} u_k \widehat{\Fx'}(\eta)
$$
i.e.
\be\label{link1}
\left(- \frac{2\pi k }{L} - \eta \right)\pa_{\eta} \hat f_k
- \left( i \xi \eta + {\sigma \eta^2}+\lambda\right)\hat f_k = u_k \,  \eta\,  e^{-i\eta \xi}e^{-\eta^2\sigma/2}\, ,
\ee
(we recall $\hat{\Fx}(\eta) = e^{-i\eta \xi}e^{-\eta^2\sigma/2}$ and hence $\widehat{\Fx'}(\eta) = i \eta  \hat{\Fx}(\eta)$),
having set
$$
	u_k:= -i G'(\xi)\varphi_k \left(\int \dif v\,  f_k(v) v\right)  =  G'(\xi)\varphi_k \left(\frac{\dif}{\dif \eta}\hat f_k(\eta)\Big\vert_{\eta=0}\right),
$$
where the second equality comes from  $\int \dif v\,  f_k(v)v = -\frac 1i  \left(\frac{\dif}{\dif \eta}\hat f_k(\eta)\Big\vert_{\eta=0}\right)$.

Still using the notation $D_k=2\pi k /L$, for $\eta \neq -D_k$ the problem \eqref{link1} can then be rewritten as
\be\label{eqn:E}
\pa_{\eta} \hat f_k + a_k \hat f_k = b_k\,,
\ee
where
\[
	a_k(\eta) = \frac{ \left( i \xi \eta + {\sigma \eta^2}+\lambda\right)}{\left(D_k+ \eta \right)}\,, \qquad b_k(\eta) = \frac{\eta}{\left(D_k+ \eta \right)}u_k  e^{-i\eta \xi}e^{-\eta^2\sigma/2}\,, \qquad u_k =  G'(\xi)\varphi_k\hat f_k'(0)\,,
\]
and we are using interchangeably the notation $\pa_{\eta}\hat f_k(\eta)$ and $\hat f_k'(\eta)$.
Let us also define, for any $k\ne 0$,
$$
	A_k(\eta) = i\xi \eta + {\sigma}\left(\frac{\eta^2}{2} - D_k \eta\right)+
	\left({\lambda+{\sigma D_k^2} - i\xi D_k}\right)\log\frac{\lv\eta + D_k\rv}{\lv D_k\rv} \,,
$$
so that $A_k'(\eta) = a_k(\eta)$ and $A_k(0)=0$.

Before moving on to further investigate \eqref{eqn:E}, let us make some comments.
\begin{note}\label{note:spectrumchoiceofspace}
	Because the equations for the  $\hat f_k$'s   are decoupled, if $\hat f_k$ satisfies \eqref{eqn:E} for a given $\lambda$, then $f_k(v) e_k(x)$ is an eigenvector for the operator $\cL$, associated with the eigenvalue $\lambda$. Hence by solving \eqref{eqn:E} we determine a set of eigenvalues and eigenvectors for $\cL$. Compatibly with the stability notion given in Section \ref{sec:linearstabilityanalysis}, see equation \eqref{stabilitycondition}, we will look for eigenvectors $f_k$ such that $\hat f_k \in L^1 \cap L^{\infty}$ and $\pa_{\eta} \hat f_k \in L^{\infty}$ and denote by $S$ the space of functions $g=g(v)$ such that $\hat g \in L^1 \cap L^ {\infty}$,  $\pa_{\eta} \hat g \in L^{\infty}$ and both $\hat g$ and $\pa_{\eta}\hat{g}$ are continuous. The requirement that  $\hat{f}_k(\eta), \,\pa_{\eta}\hat{f}_k(\eta)$ should be continuous is in order for equation \eqref{eqn:E} to make sense   (note indeed that $u_k$ is defined via $\hat f_k'(0)$). It is important to emphasize that we are not working in an $L^2$ or other Hilbert space;  even if we were, the operator $\cL$ is not self-adjoint, hence  standard spectral theorems do not apply and a priori one does not know whether there exists a basis (of $L^2$) of eigenvectors for $\cL$. In particular the analysis of this section is partially formal because it does not guarantee that we are finding {\em all} the eigenvectors (and hence all the eigenvalues) of $\cL$. The rigorous study of the spectrum of $\cL$ ( and $\ccL$) is a problem in its own right and we don't tackle it in this paper. However, for the globally scaled dynamics our objective is to show that there exists at least one eigenvalue with positive real part. Because among the set of eigenvalues we determine through our analysis there is always at least one with positive real part, the results we produce for the globally scaled dynamics are conclusive. The same is not true of the results we present in Section \ref{subsec:spectrum LS} both because we will consider an approximation of the FPP solved by the eigenvalues of $\ccL$ and also because, even if we were to solve the original FPP and prove that all the solutions of such a problem have negative real part, the set of the solutions of the FPP is still only a subset of all the eigenvalues of $\ccL$. \hfill{$\Box$}
\end{note}
With these clarifications, if for a given  $\lambda \in \mathbb C$ one can find a solution $\hat f_k \in S$ of \eqref{eqn:E} (for some $k$), then $\lambda$ is an eigenvalue for the operator $\cL$. For expository purposes it is convenient to fix an arbitrary $k$ and, for that $k$, check for which $\lambda$'s in $\mathbb C$ equation \eqref{eqn:E} admits solutions $\hat f_k \in S$ (clearly, $\hat f_k$ depends on $\lambda$ as well, but we don't include this dependence explicitly in the notation).

\begin{proposition}\label{prop:spectrum} With the notation introduced so far, the following holds:

	\noindent
	i) When $k=0$, if $\lambda =0$ or $\lambda = G'(\xi)-1$  then equation \eqref{eqn:E} admits a solution $\hat f_0 \in S$. \\
	ii) When $k\neq 0$, if  $\lambda$  is such that $\Re(\lambda)< -\sigma D_k^2$, then  there exists a solution $\hat f_k$ of equation \eqref{eqn:E} in S. \\
	iii) When $k\neq 0$, if  $\lambda$  is such that $\Re(\lambda)> -\sigma D_k^2$ then a solution $\hat f_k$ of \eqref{eqn:E} exists in $S$ if and only if the following two conditions are satisfied
	\begin{align}
		  & {\lambda G'(\xi) \varphi_k } \int_0^1\! \dif z\, e^{+\sigma D_k^2 z} z (1-z)^{\lambda - i\xi D_k + {\sigma D_k^2} -1} = 1, \label{condlam1GLOB}                                                                        \\
		  & {\lambda G'(\xi) \varphi_k } \int_0^1\! \dif z\, e^{-\sigma D_k^2 z} (2-z) (1-z)^{\lambda - i\xi D_k + {\sigma D_k^2} -1}  =  \frac{\hat f_k(-2D_k)}{\hat f_k(0)} e^{-2i\xi D_k + 2\sigma D_k^2}. \label{condlam2GLOB}
	\end{align}
\end{proposition}

As a consequence of the points $i)$ and $ii)$ of the above proposition, the operator $\cL$ will always admit eigenvalues with negative real part.
Because we are interested in determining whether it is possible for $\cL$ to have eigenvalues with positive real part, after proving Proposition \ref{prop:spectrum} we will look more closely at condition   \eqref{condlam1GLOB} and \eqref{condlam2GLOB}, the former being the most relevant. Indeed observe that  equation \eqref{condlam1GLOB}  can be used to  find the eigenvalues $\lambda$; once the eigenvalues are (at least in principle) determined, substituting such values of $\lambda$ in \eqref{condlam2GLOB} one can  get  a condition on the ratio $\hat f_k(-2D_k) / \hat f_k(0)$, so that, as expected, one actually needs only one of the constants $\hat f_k(-2D_k)$ and  $\hat f_k(0)$ in order to determine the solution $\hat f_k$ of the ODE \eqref{eqn:E} (the latter comment will possibly be more clear in view of the proof of the above proposition).
\begin{proof}[Proof of Proposition \ref{prop:spectrum}]
	We study separately the case $k=0$ and the case $k\neq 0$.

	\medskip
	\noindent
	$\bullet$ Case $k=0$. The equation reads
	\begin{equation}\label{eqnk=0}
		\pa_{\eta} \hat f_0(\eta) + \frac{ \left( i \xi \eta + {\sigma \eta^2}+\lambda\right)}\eta \hat f_0(\eta) =  G'(\xi)\, \hat f_0'(0) \, e^{-i\eta \xi}e^{-\eta^2\sigma/2}\,,
	\end{equation}
	(recall that under our scaling $\varphi_0 =1$).  Let us first look for eigenvalues  $\lambda \ne 0$, we will check separately whether $\lambda=0$ is an eigenvalue.  Under the condition $\lambda\neq 0$,   the existence of $\hat f_0'(0)$ requires $\hat f_0(0) = 0$.
	Indeed, from \eqref{eqnk=0}, we have
	$$
		\hat f_0'(0) + i \xi \hat f_0(0) + \lambda \lim_{\eta \rightarrow 0}\frac{\hat f_0(\eta)}{\eta}  = G'(\xi)\hat f_0'(0) \,,
	$$
	hence
	$$
		(1-G'(\xi))\hat f_0'(0)+ i \xi \hat f_0(0)= - \lambda\lim_{\eta \rightarrow 0} \frac{\hat f_0(\eta)}{\eta} \,.
	$$
	Expanding $\hat f_0(\eta) = \hat f_0(0)+ \hat f_0'(0) \eta + {\eta}^2 \hat f_0''(0)/2+ \tilde{\eta}^2 \hat f_0'''(0)/6$, one can see that  $\hat f_0(0) = 0$ is always needed in order for the limit on the RHS to exist and the limit is always equal to $\hat f_0'(0)$ .
	In the case $\hat f_0'(0)=0$  the only solution of \eqref{eqnk=0} is $\hat f_0(\eta)=0$. So, to recap, when  $k=0$ we have:
	\begin{itemize}
		\item if $\hat f_0'(0)=0$ then  $\hat f_0(\eta)=0$, which is not an allowed eigenvector, hence in this case no eigenvalues are found;
		\item if $\hat f_0'(0)\neq 0$ then we find the eigenvalue $\lambda = G'(\xi)-1$, which is always negative because $\lv G'(\xi)\rv<1$.
	\end{itemize}

	Let us now check whether $\lambda =0$ is an eigenvalue. If $\lambda =0$  equation \eqref{eqnk=0} becomes
	\be\label{lam0k0}
	\pa_{\eta} \hat f_0(\eta) + \left( i \xi + \frac{\sigma^2 \eta}{2} \right) \hat f_0(\eta) =   G'(\xi)\, \hat f_0'(0) \, e^{-i\eta \xi}e^{-\eta^2\sigma/2}\,,
	\ee
	the solution of which is given by
	\[
		\hat f_0(\eta) =  \big[\hat f_0(0) +  G'(\xi)\hat f_0'(0) \eta \big] e^{-i\eta \xi}e^{-\eta^2\sigma/2}\,.
	\]
	Evaluating \eqref{lam0k0} at $\eta=0$ we then find
	\[
		\hat f_0'(0)= -\frac{i\xi}{(1-G'(\xi))} \hat f_0(0)
	\]
	hence
	\[
		\hat f_0(\eta) =  \hat f_0(0) \left(1 - {G'(\xi)}\frac{i\xi \eta} {1 - G'(\xi)} \right) e^{-i\eta \xi}e^{-\eta^2\sigma/2}\,,
	\]
	which is an admissible eigenfunction, in the sense that it belongs to the space $S$.

	\bigskip
	\noindent
	$\bullet $	Case $k\ne 0$. Evaluating  equation \eqref{eqn:E} at $\eta=0$ we find

	\be\label{de}
	\hat f_k'(0) + \frac{\lambda}{D_k} \hat f_k(0) = 0\,,
	\ee
	hence
	\be\label{ukloc}
	u_k = -\hat f_k(0) \frac{\lambda G'(\xi) \varphi_k }{D_k}\,.
	\ee

	To fix ideas take $k>0$ (the case $k<0$ is similar and we will explain why below). If $k>0$ then $0\in (-D_k,\infty)$, so the equation can be recast in integral form as follows:
	\[
		\begin{split}
			\hat f_k(\eta) & = \hat f_k(0) e^{-A_k(\eta)} + e^{-A_k(\eta)} \int_0^\eta\! \dif q\, e^{A_k(q)} b_k(q) \quad \text{for } \eta> -D_k\,, \\ \hat f_k(\eta) & = \hat f_k(-2D_k) e^{A_k(-2D_k)-A_k(\eta)} + e^{-A_k(\eta)} \int_{-2D_k}^\eta\! \dif q\, e^{A_k(q)} b_k(q) \quad
			\text{for } \eta < -D_k\,\,.
		\end{split}
	\]
	Using the expression \eqref{ukloc} for $u_k$, we then get
	\begin{align}
		\hat f_k(\eta) & = \hat f_k(0) e^{-A_k(\eta)} \left( 1 - \frac{\lambda G'(\xi) \varphi_k }{D_k}  \int_0^\eta\! \dif q\, e^{A_k(q)} \frac{q}{\left(D_k+q\right)}  e^{-iq \xi}e^{-q^2\sigma/2} \right) \nonumber                                                      \\
		               & = \hat f_k(0) \exp\left(-i \xi \eta - \frac{\eta^2\sigma}{2} + {\sigma D_k \eta} \right) \left(\frac{D_k}{\eta+D_k}\right)^{\lambda - i\xi D_k + \sigma D_k^2} \nonumber                                                                           \\
		               & \quad \times \left(1 - \frac{\lambda G'(\xi) \varphi_k }{D_k}  \int_0^\eta\! \dif q\, e^{-\sigma D_k q} \frac{q}{q+D_k} \left(\frac{q+D_k}{D_k}\right)^{\lambda - i\xi D_k + \sigma D_k^2 } \right) \quad \text{for } \eta > -D_k\,,\label{kposg1}
	\end{align}
	and, noting that $A_k(-2D_k) = -2i\xi D_k + 4 \sigma D_k^2$,
	\begin{align}
		\hat f_k(\eta) & = e^{-A_k(\eta)} \left(\hat f_k(-2D_k) e^{-2i\xi D_k + 4 \sigma D_k^2}  - \hat f_k(0) \frac{\lambda G'(\xi) \varphi_k }{D_k}  \int_{-2D_k}^\eta\! \dif q\, e^{A_k(q)} \frac{q}{\left(D_k+q\right)}  e^{-iq \xi}e^{-q^2\sigma/2} \right) \nonumber   \\
		               & = \exp\left(-i \xi \eta - \frac{\eta^2\sigma}{2} + {\sigma D_k \eta} \right) \left|\frac{D_k}{\eta+D_k}\right|^{\lambda - i\xi D_k + \sigma D_k^2} \bigg(\hat f_k(-2D_k) e^{-2i\xi D_k + 4 \sigma D_k^2} \nonumber                                  \\
		               & \quad - \hat f_k(0) \frac{\lambda G'(\xi) \varphi_k }{D_k}  \int_{-2D_k}^\eta\! \dif q\, e^{-\sigma D_k q} \frac{q}{q+D_k} \left|\frac{q+D_k}{D_k}\right|^{\lambda - i\xi D_k + \sigma D_k^2} \bigg) \quad \text{for } \eta < -D_k\,.\label{kposg2}
	\end{align}

	Let us compute the limits $\hat f_k(-D_k^\pm)$ of $\hat f_k(\eta)$ as $\eta \to - D_k^\pm$ in the case $\Re(\lambda) < -\sigma D_k^2$:

	\[
		\begin{split}
			\hat f_k(-D_k^+) & = - \hat f_k(0) \exp\left(i \xi D_k - \frac{3D_k^2\sigma}2\right) \frac{\lambda G'(\xi) \varphi_k }{D_k}  \\ & \quad \times \lim_{\eta\to -D_k^+}\left(\frac{D_k}{\eta+D_k}\right)^{\lambda - i\xi D_k + \sigma D_k^2} \int_0^\eta\! \dif q\, \frac{e^{-\sigma D_k q}q}{q+D_k} \left(\frac{q+D_k}{D_k}\right)^{\lambda - i\xi D_k + \sigma D_k^2}
		\end{split}
	\]
	\[
		\begin{split}
			\hat f_k(-D_k^-) & = - \hat f_k(0) \exp\left(i \xi D_k - \frac{3D_k^2\sigma}2\right) \frac{\lambda G'(\xi) \varphi_k }{D_k}  \\ & \quad \times \lim_{\eta\to -D_k^-}\left|\frac{D_k}{\eta+D_k}\right|^{\lambda - i\xi D_k + \sigma D_k^2} \int_{-2D_k}^\eta\! \dif q\, \frac{e^{-\sigma D_k q}q}{q+D_k} \left|\frac{q+D_k}{D_k}\right|^{\lambda - i\xi D_k + \sigma D_k^2}
		\end{split}
	\]
	Since the integrals are diverging for $\eta\to -D_k^\pm$ and the functions under integration are the same,  $\hat f_k(-D_k^+) = \hat f_k(-D_k^-)$ (regardless of the values of the constants $\hat f(0)$ and $\hat f(-2D_k)$). Indeed, if $g$ is any smooth function, $b \in \mathbb R$ and $\epsilon \in \mathbb C$ is such that $\Re(\epsilon)>0$, we have
	\[
		\begin{split}
			\lim_{\eta\to b} (b-\eta)^\eps \int_a^\eta\! \dif q\, \frac{g(q)}{(b-q)^{1+\eps}} & = \lim_{\eta\to b} (b-\eta)^\eps \int_a^\eta\! \dif  q\, \frac{g(q)-g(b)}{(b-q)^{1+\eps}} + g(b) \lim_{\eta\to b} (b-\eta)^\eps \int_a^\eta\!\frac{\dif q}{(b-q)^{1+\eps}} \\ & = 0 - \frac{g(b)}{\eps}\lim_{\eta\to b} (b-\eta)^\eps \left[\frac{1}{(b-\eta)^\eps} - \frac{1}{(b-a)^\eps} \right] = - \frac{g(b)}{\eps}.
		\end{split}
	\]
	Therefore
	\[
		\hat f_k(-D_k^+) = \hat f_k(-D_k^-) = - \hat f_k(0) \exp\left(i \xi D_k - \frac{D_k^2\sigma }2 \right) \frac{\lambda G'(\xi) \varphi_k }{(\lambda - i\xi D_k + \sigma D_k^2)}.
	\]

	Consider now the case $\Re(\lambda) > -\sigma D_k^2$. In this case $\Re( \lambda - i\xi D_k + {\sigma D_k^2} )>0
	$
	so that the integrals are now converging for $\eta\to -D_k^\pm$. With the change of variable $q=-D_k z$ in the first one and $q=D_k z-2D_k$ in the second one, they read
	\[
		\begin{split}
			\int_0^{-D_k}\! \dif q\, e^{-\sigma D_k q} \frac{q}{q+D_k} \left(\frac{q+D_k}{D_k}\right)^{\lambda - i\xi D_k + \sigma D_k^2} & =D_k \int_0^1\! \dif z\, e^{\sigma D_k^2 z} z (1-z)^{\eps-1+i\gamma}, \\ \int_{-2D_k}^{-D_k}\! \dif q\, e^{-\sigma D_k q} \frac{q}{q+D_k} \left|\frac{q+D_k}{D_k}\right|^{\lambda - i\xi D_k + \sigma D_k^2} & = e^{2\sigma D_k^2} D_k \int_0^1\! \dif z\, e^{-\sigma D_k^2 z} (2-z) (1-z)^{\eps-1+i\gamma}.
		\end{split}
	\]
	Therefore, a necessary and sufficient condition to have non singular limits of $\hat g(\eta)$ as $\eta\to -D_k^\pm$ is that \eqref{condlam1GLOB} and \eqref{condlam2GLOB} hold.
	{When  $k<0$ the analogous of \eqref{kposg1} and \eqref{kposg2} are
	\[
		\begin{split}
			\hat f_k(\eta) & = \hat f_k(0) e^{-A_k(\eta)} \left( 1 - \frac{\lambda G'(\xi) \varphi_k }{D_k}  \int_0^\eta\! \dif q\, e^{A_k(q)} \frac{q}{\left(D_k+q\right)}  e^{-iq \xi}e^{-q^2\sigma/2} \right) \\  & = \hat f_k(0) \exp\left(-i \xi \eta - \frac{\eta^2\sigma}2 + {\sigma D_k \eta} \right) \left(\frac{D_k}{\eta+D_k}\right)^{\lambda - i\xi D_k + \sigma D_k^2} \\ & \quad \times \left(1 - \frac{\lambda G'(\xi) \varphi_k }{D_k}  \int_0^\eta\! \dif q\, e^{-\sigma D_k q} \frac{q}{q+D_k} \left(\frac{q+D_k}{D_k}\right)^{\lambda - i\xi D_k + \sigma D_k^2} \right) \quad \text{for } \eta < -D_k\,,
		\end{split}
	\]
	and
	\[
		\begin{split}
			\hat f_k(\eta) & = e^{-A_k(\eta)} \left(\hat f_k(-2D_k) e^{-2i\xi D_k + 4 \sigma D_k^2}  - \hat f_k(0) \frac{\lambda G'(\xi) \varphi_k }{D_k}  \int_{-2D_k}^\eta\! \dif q\, e^{A_k(q)} \frac{q}{\left(D_k+q\right)}  e^{-iq \xi}e^{-q^2\sigma/2} \right) \\  & = \exp\left(-i \xi \eta - \frac{\eta^2\sigma}2 + {\sigma D_k \eta} \right) \left|\frac{D_k}{\eta+D_k}\right|^{\lambda - i\xi D_k + \sigma D_k^2} \bigg(\hat f_k(-2D_k) e^{-2i\xi D_k + 4 \sigma D_k^2} \\ & \quad - \hat f_k(0) \frac{\lambda G'(\xi) \varphi_k }{D_k}  \int_{-2D_k}^\eta\! \dif q\, e^{-\sigma D_k q} \frac{q}{q+D_k} \left|\frac{q+D_k}{D_k}\right|^{\lambda - i\xi D_k + \sigma D_k^2} \bigg) \quad \text{for } \eta > -D_k\,,
		\end{split}
	\]
	}
	respectively. With these expressions one can follow the same reasoning we have done above for the case $k<0$, from which we conclude that the necessary and sufficient conditions \eqref{condlam1GLOB} and \eqref{condlam2GLOB} remain unaltered also in the case $k<0$.
\end{proof}

{\bf The fixed point problem \eqref{condlam1GLOB}}. As explained after Proposition \ref{prop:spectrum}, in order to understand whether the operator $\cL$ has eigenvalues with positive real part we can focus on equation \eqref{condlam1GLOB}. Such an equation can be seen as a fixed point problem for $\lambda$. For $\sigma >0$ solving this equation explicitly seems difficult;  trying to solve it numerically has  also proved more challenging than expected as all the (elementary and less elementary) numerical methods we have used to solve it have proved to be quite unstable. Nonetheless one can prove the following.
\begin{proposition}\label{prop:positive-eigenvalues-GLOB}
	For $\sigma$ small enough, equation \eqref{condlam1GLOB} admits (at least one) solution $\lambda$ such that the real part of $\lambda$ is positive.  Hence the operator $\cL$ has eigenvalues with positive real part.
\end{proposition}
\begin{proof}[Proof of Proposition \ref{prop:positive-eigenvalues-GLOB}]
	For $\sigma =0$ equation   \eqref{condlam1GLOB} becomes
	\be\label{sigma=0fixedpointglob}
	{\lambda G'(\xi) \varphi_k } \int_0^1\! \dif z\, z (1-z)^{\lambda - i\xi D_k  -1} = 1 \,. \ee
	The above integral can now be calculated so,
	setting  $\dk:= G'(\xi)\varphi_k$ and taking   $\xi=1$ to fix ideas (analogous calculations can be done for $\xi=-1$), we have
	$$
		\lambda \dk \int_0^1 \dif z \, z (1-z)^{\lambda-1-iD_k} =  \frac{\lambda \dk}{( - iD_k + \lambda)(1 -iD_k + \lambda)} = 1 \,, \quad \Re(\lambda)>0\,.
	$$
	The above then boils down to the following simple system for the real and imaginary part of $\lambda$:
	\begin{align*}
		  & D_k^2-\rel +\dk \, \rel - (\rel)^2-2 D_k \iml+(\iml)^2=0 \\
		  & D_k+ 2 D_k \rel-\iml +\dk \,\iml-2\rel\iml=0 \,.
	\end{align*}
	There are  two families of real solutions of the above system, namely
	\begin{align}
		\Re(\lambda_k)         & = \frac{1}{2} \left((\dk -1)- \frac{1}{\sqrt 2}\sqrt{\rk+\sqrt{\rk^2+16 \dk^2 D_k^2}}\right)\label{negeigglob} \\
		\mathrm{Im}(\lambda_k) & =  \frac{1}{4\dk D_k}
		(4\dk D_k^2 + r_k - 2 \dk r_k+ \dk^2r_k-r_k^3) \nonumber
	\end{align}
	and
	\begin{align}
		\Re(\lambda_k)         & =\frac{1}{2} \left((\dk-1) + \frac{1}{\sqrt 2}\sqrt{\rk+\sqrt{\rk^2+16 \dk^2 D_k^2}}\right) \label{poseigglob} \\
		\mathrm{Im}(\lambda_k) & =  \frac{1}{4\dk D_k}
		(4\dk D_k^2 - r_k + 2 \dk r_k- \dk^2 r_k+r_k^3),  \nonumber
	\end{align}
	where $r_k:=(d_k-1)^2$ (and it is understood that the above solutions are valid only for those $k$'s such that $d_k$ is non zero).
	It is now easy to see that the solutions \eqref{negeigglob} are always negative, as
	$$
		\dk -1 = {G'(1)\varphi_k} -1 <0 \, ,
	$$
	(because $0<G'(1)<1$ and, by definition, $\lv \varphi_k\rv <1$) while the family of solutions \eqref{poseigglob} are always positive. Hence, in conclusion, when $\sigma =0$ there exist solutions $\lambda$ of \eqref{condlam1GLOB} such that the real part of lambda is positive. By the implicit function theorem we then have that for sigma small enough there will still be solutions of the same equation with positive real part. Such a theorem can be applied to the (analytic) function $F(\lambda, \sigma) = {\lambda d_k } \int_0^1\! \dif z\, e^{+\sigma D_k^2 z} z (1-z)^{\lambda - i\xi D_k + {\sigma D_k^2} -1} - 1 $  once we show that the complex derivative $\pa_{\lambda}F$ calculated at $\lambda$ as in \eqref{poseigglob} and $\sigma=0$ is different from zero. More precisely,  we need this derivative to be non-zero for at least one of the $\lambda_k$'s in \eqref{poseigglob}, not necessarily for all of them. To this end,
	$$
		\pa_{\lambda}F = {d_k } \int_0^1\! \dif z\, e^{+\sigma D_k^2 z} z (1-z)^{\lambda - i\xi D_k + {\sigma D_k^2} -1} +{ \lambda d_k } \int_0^1\! \dif z\, \log (1-z) e^{+\sigma D_k^2 z} z (1-z)^{\lambda - i\xi D_k + {\sigma D_k^2} -1} \,.
	$$
	When $\sigma=0$, both of the above integrals can be calculated explicitly, so we have (again taking $\xi=1$ just to fix ideas)
	\begin{align*}
		\pa_{\lambda}F\vert_{\sigma=0} & = \frac{ \dk}{( - iD_k + \lambda)(1 -iD_k + \lambda)} + {d_k \lambda}\left(-\frac{1}{(\lambda - iD_k)^2} + \frac{1}{(\lambda+1-iD_k)^2}\right) \\
		                               & = d_k \,  \frac{-\lambda^2+ (iD_k-1)iD_k}{(\lambda-iD_k)^2 (\lambda+1 - iD_k)^2} \,.
	\end{align*}
	Because $d_k$ can't be zero for every $k$,we just need to focus on the numerator of the above, the imaginary part of which is clearly non-zero, for any $k \in \mathbb{Z}\setminus \{0\}$. As for the real part of the numerator, when calculated at $\lambda_k$ as in \eqref{poseigglob}, it is given by
	$$
		-\lambda_k^2- D_k^2 = - D_k^2 +
		\frac{((d_k-2) (d_k-1)^4 + 4D_k^2)^2}{16 D_k^2} - \frac 14
		\left[  d_k-1 + \frac{1}{\sqrt 2}
			\sqrt{(d_k-1)^2 + \sqrt{(d-1)^4+ 16d_k^2 D_k^2}}
			\right]^2
	$$
	It is now easy to see that there exists always at least one $k$ (possibly dependent on $L$ and $\varphi$) such that the above is non-vanishing.
\end{proof}

To gain some more insight one can also carry out the following heuristic calculation: for   $\sigma$  positive but  small,  we  approximate the  integral appearing in \eqref{condlam1GLOB}  by approximating $e^{\sigma k^2} \simeq 1$ (clearly this approximation is only valid for $\sigma$ and $k$ small enough) and consider the simpler problem
\be\label{approxprobglobsigma}
{\lambda G'(\xi) \varphi_k } \int_0^1\! \dif z\,  z (1-z)^{\lambda - i\xi D_k + {\sigma D_k^2} -1} = 1 \,.
\ee
The integral on the LHS of the above is explicitly computable and, as in the proof of Proposition \ref{prop:positive-eigenvalues-GLOB},  one can reduce the above problem to a simple system of two equations for the real and imaginary part of $\lambda$.  Setting   $\ck:=D_k^2\sigma ,  \,  \dk:= G'(\xi)\varphi_k$ and taking again $\xi=1$ to fix ideas,  \eqref{approxprobglobsigma} becomes
$$
	\lambda \dk \int_0^1 \dif z \, z (1-z)^{\lambda-1-iD_k+\ck} =  \frac{\lambda \dk}{(\ck - iD_k + \lambda)(1+\ck -iD_k + \lambda)} = 1 \,, \quad \Re(\ck+\lambda)>0\,.
$$
The above equation then boils down to the following simple system for the real and imaginary part of $\lambda$:
\begin{align*}
	  & D_k+ 2\ck D_k + 2D_k \rel -\iml - 2\ck \iml+ \dk \iml -2\rel \iml=0           \\
	  & -\ck -\ck^2+D_k^2-\rel - 2\ck \rel +\dk \rel -\rel^2 - 2D_k\iml+\iml^2 =0 \,.
\end{align*}
\newcommand{\ak}{a_k}
Setting  this time $\rk:= \dk^2-4\ck\dk - 2\dk +1$ and $\ak: = \dk-2\ck-1$, the only real solutions of the above system are
\begin{align}
	\Re{(\lambda_k)} & = \frac{1}{2} \left(\ak\pm \frac{1}{\sqrt 2}\sqrt{\rk+\sqrt{\rk^2+16 \dk^2 D_k^2}}\right)\label{negeigglob1}
\end{align}
(we don't report the corresponding value of the imaginary part as the formula is lengthy and all we are interested in is the real part of the eigenvalue.)
Now note that the solution \eqref{negeigglob1} with the minus is always negative, as $\ak$ is always negative, irrespective of $k$; this is due to the fact that $\ck>0$ by definition.  As for the solution with the plus sign, this can be positive or negative depending on the value of $\sigma$. Indeed, start by noting that $\ak= -\sqrt{\rk+ 4\ck (1+\ck)}$. Hence, checking whether the solution \eqref{negeigglob1} with a plus can be positive boils down to checking whether
$$
	\frac{1}{\sqrt 2}\sqrt{\rk+\sqrt{\rk^2+16 \dk^2 D_k^2}}>\sqrt{\rk+4\ck(1+\ck)}.
$$
This is in turn equivalent to
$$
	{\sqrt{\rk^2+16 \dk^2 D_k^2}}>{\rk+8\ck(1+\ck)}\,.
$$
Because the constant $\ck$ is directly proportional to $\sigma$, for $\sigma$ small enough the above is satisfied, hence the presence of positive eigenvalues.

\subsection{Spectrum of the locally scaled model}\label{subsec:spectrum LS}
We now linearise the LS PDE \eqref{nl} around the equilibria \eqref{eqn:equil}, and obtain
$$
	\pa_t f = - v \pa_x f -\pa_v \left[(\xi - v) f\right]+ \frac{1}{L}G'(\xi) \left[\iint (\xi-w) \varphi(x-y) f(y,w) \dif w\, \dif y\right]\Fx'(v)+ \sigma\pa_{vv}f \,.
$$
As before, we consider the eigenvalue problem $\ccL f=\lambda f$, where $\ccL$ is the operator on the RHS of the above, and expand $f$ as in \eqref{expansionf}.
Then the equation satisfied by $f_k(v)$ is
$$
	\frac{2 \pi i k}{L} v f_k(v) - \pa_v \left[(\xi-v)f_k(v)\right]
	+  G'(\xi)\varphi_k \left( \xi\int \dif v f_k(v)  - \int \dif v f_k(v) v\right)\Fx'(v) + \sigma \pa_{vv} f_k(v) = \lambda f_k(v) \,,
$$
and the Fourier transform of $f_k$ satisfies
\be\label{link2}
\left(- \frac{2\pi k }{L} - \eta \right)\pa_{\eta} \hat f_k
- \left(i \xi \eta + {\sigma \eta^2}+\lambda\right)\hat f_k =  u_k \eta e^{-i\eta \xi}e^{-\eta^2\sigma/2}\, ,
\ee
where this time
\be\label{uklocal}
u_k =  i G'(\xi) \varphi_k \left(\xi \hat f_k(0) +
\frac 1i \frac{\dif}{\dif \eta} \hat f_k(\eta) \Big\vert_{\eta=0}\right)
\ee

One can carry out calculations completely analogous to those done in the proof of  Proposition \ref{prop:spectrum}, leading this time to the following result

\begin{proposition}\label{prop:spectrumlocal} With the notation introduced so far, the following holds:

	\noindent
	i) When $k=0$, if $\lambda =0$ or $\lambda = G'(\xi)-1$  then equation \eqref{link2} admits a solution $\hat f_0 \in S$. \\
	ii) When $k\neq 0$, if  $\lambda$  is such that $\Re(\lambda)< -\sigma D_k^2$, then  there exists a solution $\hat f_k$ of equation \eqref{link2} in S. \\
	iii) When $k\neq 0$, if  $\lambda$  is such that $\Re(\lambda)> -\sigma D_k^2$ then a solution $\hat f_k$ of \eqref{link2} exists in $S$ if and only if the following two conditions are satisfied
	\begin{align}
		  & -D_k \left[{G'(\xi)\varphi_k}\left(i\xi-\frac{\lambda}{D_k}\right)\right]
		\int_0^1 \dif z e^{\sigma D_k^2 z}
		z(1-z)^{\lambda - i\xi D_k + {\sigma D_k^2} -1}=1 \label{condlam1LOC}\\
		  & D_k\left[{G'(\xi)\varphi_k}\left(i\xi-\frac{\lambda}{D_k}\right)\right]
		\int_0^1 \dif z e^{-\sigma D_k^2 z}
		(2-z) (1-z)^{\lambda - i\xi D_k + {\sigma D_k^2} -1} =- \frac{\hat f(-2D_k)}{\hat{f}_k(0)}
		e^{-2i\xi D_k+2\sigma D_k^2} \,. \label{condlam2LOC}
	\end{align}
\end{proposition}

With a reasoning similar to the one done before we can focus on equation \eqref{condlam1LOC}. The integral in \eqref{condlam1LOC} is convergent only for $\rel > -\sigma D_k^2$ hence,  if we set $\sigma=0$, it is only convergent for $\rel >0$. With this in mind, setting $\sigma=0$ in \eqref{condlam1LOC} and acting analogously to what we have done for the globally scaled case, we find that the only solution is given by $\rel = G'(\xi) \varphi_k-1 <0, \iml=D_k$. As such a solution has negative real part it needs to be discarded (the integral diverges so it cannot be a solution), leading to the conclusion that for $\sigma=0$ the problem \eqref{condlam1LOC} admits no solutions with positive real part (and indeed no solutions at all).

Of course in this case this does not imply anything on what happens for $\sigma$ small.
Again, for $\sigma$ small enough, we can again heuristically approximate the problem   \eqref{condlam1LOC} with the simpler
$$
	-D_k \left[{G'(\xi)\varphi_k}\left(i\xi-\frac{\lambda}{D_k}\right)\right]
	\int_0^1 \dif z \,
	z(1-z)^{\lambda - i\xi D_k + {\sigma D_k^2} -1}=1 \,.
$$
We act as in the global case and find the following four pairs of solutions (this time we need to report the formulas for the real and imaginary part)
\begin{align*}
	  & \Re(\lambda_1)=\frac{\ak}{2}, \qquad \qquad \quad\mathrm{Im}(\lambda_1)= \frac 12 (2D_k-\sqrt{-\rk})   \\
	  & \Re(\lambda_2)=\frac{\ak}{2},  \qquad \qquad \quad \mathrm{Im}(\lambda_2)= \frac 12 (2D_k+\sqrt{-\rk}) \\
	  & \Re(\lambda_3)=\frac{\ak-\sqrt{\rk}}{2}, \qquad \mathrm{Im}(\lambda_3)= D_k                            \\
	  & \Re(\lambda_4)=\frac{\ak-\sqrt{\rk}}{2}, \qquad \mathrm{Im}(\lambda_4)= D_k \,.
\end{align*}
While $a_k<0$, $r_k$ does not have a sign. So, if $r_k$ is positive then the first two solutions must be discarded (as the real and imaginary parts need to be real numbers), leaving us with the third and fourth solutions, both of which have negative real part. Vice versa, if $r_k<0$, then the third and fourth solutions need to be discarded, leaving us with the first two, both of which, again, have negative real part.

\appendix
\section{Auxiliary results}\label{app:PSMeanZero}

  Consider the stationary problem associated with \eqref{nl}-\eqref{Mnl}, i.e.  the solutions of the equation
\begin{equation}\label{statpb}
    v\pa_x f(x,v)+ G(M_f(x)) \pa_vf(x,v) = \pa_v{(vf(x,v))}+\sigma \pa_{vv}f(x,v),   \quad (x,v) \in \Tt\times\R
\end{equation}
where $M_f(x)$ is defined in \eqref{Mnl}.
\begin{lemma}\label{lem:whatweknow}
Let $f(x,v)$ be a smooth solution of \eqref{statpb}-\eqref{Mnl} which decays fast at infinity (together with its derivatives). Then
\begin{description}
\item[i)]  $\int \dif v \, v f(x,v) =c$, where $c$ is some constant independent of $x$;
\item[ii)] if $f(x,v)$ is independent of $x$, then $f$ can only be one of the three densities $\mu_{\pm}, \mu_0$ in \eqref{invmeas}. 
\end{description}
\end{lemma}
\begin{proof}
The claim i) is obtained by simply integrating the whole equation in the velocity variable; point ii) is obtained by simply observing that if $f$ is independent of $x$ then \eqref{statpb} boils down to 
$$
G(M_f(x))\partial_v f(x,v) = \partial_v (vf(x,v)) + \sigma\partial_{vv} f(x,v), \quad M(t) = \int \dif v \,v f(x,v)
$$  
which is explicitly solvable.
\end{proof}

\noindent 
{\bf Invariant measure of Particle System \texorpdfstring{\eqref{parsys1}-\eqref{parsys2}}{}}
Consider the evolution described by \eqref{parsys2}, with $\varphi\equiv 1$, that is the space homogeneous particle system with uniform interaction:
\begin{align}
    \mathrm{d} {v}_t^{i,N} & = - v_t^{i,N} \mathrm{d} t + G\left(\frac{ \sum_{j=1}^N  v_t^{j,N}}{N}\right) \mathrm{d} t + \sqrt{2\sigma} \mathrm{d} B_t^{i}, \quad i = 1,\dots N.
    \label{eq:spacehomPS}
\end{align}
Notice that this is equivalent to the globally scaled particle model of \cite{garnierMeanFieldModel2019} when $\varphi\equiv1$. We will show that this equation has an invariant measure with mean zero. To this end, let $G$ be smooth and define $\bar{v}_t = N^{-1} \sum_{i=1}^N  v_t^{i,N}$. Then, by summing over $i$ in \eqref{eq:spacehomPS}, $\bar{v}_t$ solves
\begin{equation}\label{eq:avgvelPS}
    \mathrm{d}\bar{v}_t = \left[-\bar{v}_t+G(\bar{v}_t)\right]\mathrm{d}t + \sqrt{2\sigma}\mathrm{d}B_t, \quad \bar{v}_t \in \mathbb{R}.
\end{equation}
Let $V(x) = \frac{x^2}{2} + \tilde{V}(x)$, where $\tilde{V}(x)$ is such that $G(x) = -\tilde{V}'(x)$ so that
\[
    -V'(\bar{v}) = -\bar{v} + G(\bar{v}).
\]
Equation \eqref{eq:avgvelPS} can be rewritten as
\begin{equation}\label{eq:avgvelPSrewrote}
 \mathrm{d}\bar{v}_t = -V'(\bar{v})\mathrm{d}t + \sqrt{2\sigma}\mathrm{d}B_t, \quad \bar{v}_t \in \mathbb{R}.
\end{equation}
Now $\tilde{V}(\bar{v}) = - \int_{-\infty}^{\bar{v}} G(u) \mathrm{d} u $ by definition, so that if $G$ is bounded then $\tilde{V}(x)$ grows at most linearly. Hence, $V(\bar{v}) \to +\infty$ as $|\bar{v}|\to \infty$ so that $\mathrm{e}^{-V(\bar{v})}$ is integrable. Thus \eqref{eq:avgvelPSrewrote} admits an invariant measure, which is 
\[
    \rho(\bar{v}) = \frac{1}{Z}\mathrm{e}^{-V(\bar{v})},
\]
where $Z$ is a normalising constant that we neglect from now on. Note that $\tilde{V}(\bar{v})$ is an even function as $G(\bar{v})$ is odd. This implies that
\[
    V(\bar{v})=\frac{\bar{v}^2}{2}+ \tilde{V}(\bar{v})
\]
is even, and so 
\[
    \int_{\mathbb{R}} \bar{v}\mathrm{e}^{-V(\bar{v})} \mathrm{d}\bar{v}  = 0.
\]
That is, as $t\to \infty$, $\mathbb{E}[\bar{v}_t] \to 0$.

\section{Computational Method for the PDE}\label{app:PDEComp}
When solving the PDE models \eqref{nl},\eqref{nlGarnier}, we have chosen a pseudospectral method, in particular the scheme described in \cite{noldPseudospectralMethodsDensity2017}. This has been implemented in MATLAB and made freely available \cite{goddard2DChebClass2017}. The only adjustments we have made are the introduction of 2D periodic-Chebyshev domains, and the associated use of a different mapping from the computational domain to the physical space. In this appendix we describe further our modifications and and justify our choices of parameters of the scheme using convergence plots. Throughout we will follow the notation of \cite{noldPseudospectralMethodsDensity2017} where applicable. For brevity, we will not define all objects mentioned and instead refer to \cite{noldPseudospectralMethodsDensity2017}. 

As mentioned in Section \ref{sec:preliminaries}, when simulating \eqref{nl} there are two main difficulties: the hypoellipticity of the diffusion term and the nonlinearity \eqref{Mnl}. Many numerical schemes can introduce artificial diffusion when applied to such problems \cite{mortonNumericalSolutionPartial2005}. If such a scheme was used here, any clusters potentially forming in space would be dispersed. Figure \ref{fig:ArtificialDiffusion} shows the dispersal seen using a pseudospectral method compared to an explicit Lax-Wendroff (finite difference) scheme. 
\begin{figure}
\begin{minipage}{\linewidth}
    \begin{minipage}{0.45\linewidth}
        \centering
        \includegraphics[width=\linewidth]{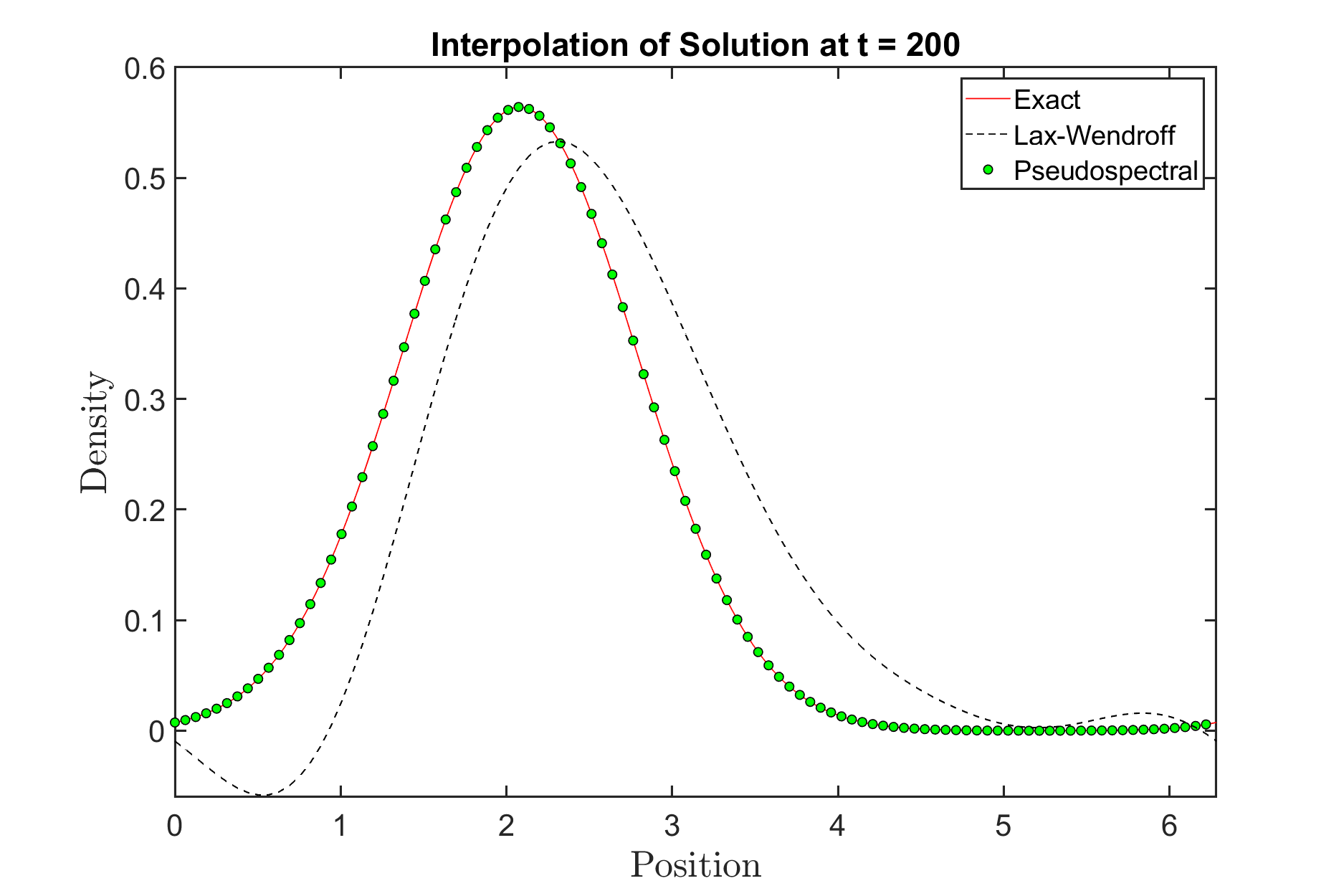}
    \end{minipage}%
    \begin{minipage}{0.45\linewidth}
        \centering
        \includegraphics[width=\linewidth]{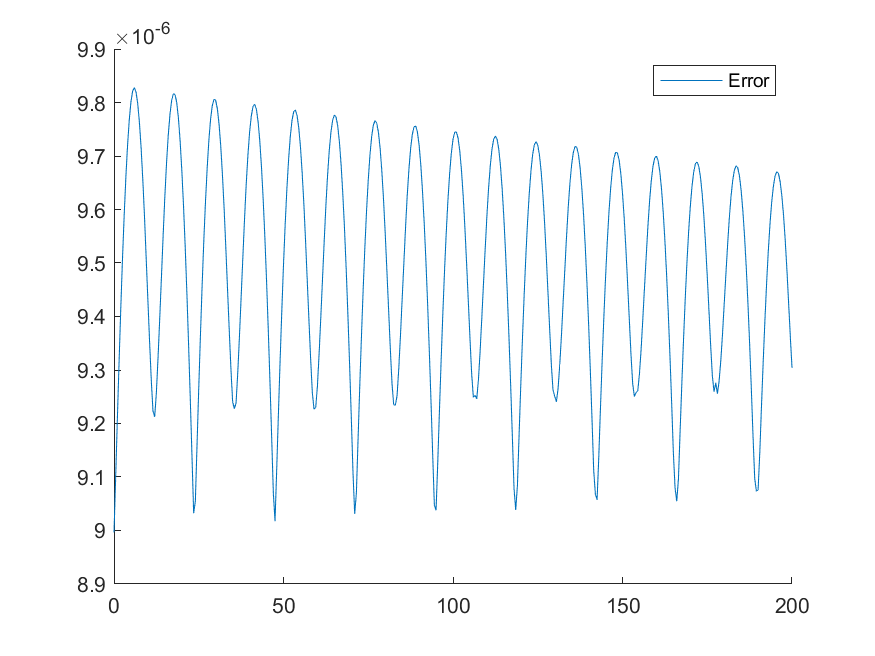}
        
    \end{minipage}\end{minipage}
    \caption{Solving the advection equation \(f_t = -f_x\) with \(f(0,x) = \exp(-(x-\pi)^2)\) on a periodic domain \([0,2\pi)\) using the pseudospectral method described here (green circles) and using a Lax-Wendroff scheme (dashed line). Both use 100 grid points in space. The left shows the exact and approximate solutions at $t=200$ while the right shows the \(\ell^1\) difference between the approximate pseudospectral and exact solution, \(\int_{[0,2\pi)} \|f^{\app}(t,x) - f(0,x+t) \| \dif x\). The Lax-Wendroff scheme shows a greater error even with a smooth initial profile. To get similar accuracy requires many more grid points, a problem which is only compounded in higher dimensions. The Lax-Wendroff approximation also becomes negative in certain regions, which is an impossibility for the probability density we study here.}
    \label{fig:ArtificialDiffusion}
\end{figure}

We will describe the scheme separately in each dimension (in this case position and velocity) as combining them into the 2D problem is straightforward -- one can simply take the Kronecker tensor product of the 1D domains, and create the associated operators as described in~\cite{noldPseudospectralMethodsDensity2017}. A key choice in this approach regards the mapping from the `computational' space (typically $[0,1]$ or $[-1,1]$) to the `physical' space of the problem; this choice has a large impact on the accuracy of the scheme. In position, the solution, \(f\) lies on the periodic domain \(\mathbb{T}\). The function \(f\) is mapped to the computational space \([0,1]\) using the obvious linear map. It is then represented by its values at the collocation points \(x_n\), which here form an equally spaced grid, and an interpolating Fourier series. This is a standard technique, see \cite{trefethenSpectralMethodsMATLAB2000}.

As the velocity domain is not periodic, we instead use Gauss-Lobatto-Chebyshev collocation points and interpolate using Chebyshev polynomials \cite{trefethenSpectralMethodsMATLAB2000}. Doing so clusters the points at the boundary thereby avoiding the Runge phenomenon. The computational space here is \([-1,1]\). The problem \ref{nl} is on an unbounded domain (\(\R\)), which can be represented using a square root map or similar -- see for example \cite{noldPseudospectralMethodsDensity2017}. However, in this case the map was found to cause inaccuracies as the solution is very small far from the centre of the domain. As a consequence, calculating advective terms or moments involving the product of \(v\) and \(f_t(x,v)\) caused oscillations to appear around the expected mean. It was thus decided to truncate the domain to \([-L_v,L_v]\) for some \(L_v\) large enough that the majority of the mass is contained within the interval. The standard approach for solving on the truncated domain uses the obvious linear map, preserving the spacing of the Gauss-Lobatto-Chebyshev points. As we expect the solution to decay quickly and the mass to be in the centre of the domain, this is not ideal in our case. To correct this, the tangent mapping of \cite{baylissMappingsAccuracyChebyshev1992} is introduced.  

\begin{align}\label{eq:fulltanmap}
&\mathcal{A}:[-1,1] \to [-L_v, L_v], & v\mapsto -L_v + L_v\left(s_0 + \frac{\atan(\alpha_1(v-\alpha_2))}{\lambda})\right)
\end{align}
where
\begin{align*}
    & s_0 = \frac{\kappa-1}{\kappa+1}, & \kappa = \frac{\atan(\alpha_1(1+\alpha_2))}{\atan(\alpha_1(1-\alpha_2))},  &&\lambda = \frac{\atan(\alpha_1(1-\alpha_2))}{1-s_0}.&\\
\end{align*}
Using this mapping shifts the points to be clustered around \(\alpha_2\), similar to the clustering seen around the origin when using the square root map in an unbounded domain. The tangent mapping provides the best of both worlds, giving high accuracy in the centre like the square root map but without the spurious oscillations. The parameter \(\alpha_2\) is set to zero so as not to bias the solution towards either \(\pm 1\). If, after solving, a more accurate solution is desired, one could rerun the scheme with \(\alpha_2\) at the observed final mean of the previous run. It could also be adaptively changed over time, although the computational cost incurred in recalculating the differentiation matrices is likely to outweigh any accuracy benefit. In practice, we found sufficient accuracy and reasonable computation time without changing \(\alpha_2\). Setting \(\alpha_2 =0\) simplifies the above map \eqref{eq:fulltanmap} considerably to 
\begin{align}\label{eq:tanmap}
    &\mathcal{A}:[-1,1] \to [-L_v, L_v], & v\mapsto -L_v + L_v\left( \frac{\atan(\alpha_1 v)}{\atan(\alpha_1)}\right)
\end{align}

Using a truncated domain introduces a boundary to the problem, and thus requiring an artificial boundary condition. Here the most natural boundary condition is a Dirichlet zero condition as the solution is expected to be rapidly-decaying. For large enough \(L_v\), the mass loss by introducing the zero condition will be below the accuracy of the numerical scheme. For example, the mass contained below \(-8\) for a Gaussian distribution with mean \(-1\) and variance \(1\) is \(6.2210\times 10^{-16}\). The imposition of the boundary condition transforms the ODE into a differential-algebraic equation (DAE), necessitating the use of a specialised DAE solver, such as MATLAB's \emph{ode15s}. 

The convolution term \eqref{Mnl} is similar to that seen in many DDFT equations and schemes such as \emph{2DChebClass} have been shown to accurately compute them \cite{noldPseudospectralMethodsDensity2017}. The non-smoothness of the interaction function \eqref{eq:phiIndicator} does not cause an issue in these calculations as it is not interpolated globally -- instead, the solution \(f\) is interpolated on the region in which \eqref{eq:phiIndicator} is nonzero and the convolution calculated there.  An analogous approach (in 1D) has been used in recent work for models in opinion dynamics, which also feature convolution kernels with finite support~\cite{goddardNoisyBoundedConfidence2022}.

We now explain our choice of free parameters, given in Table \ref{tab:schemeparameters}, which are the contents of the vector \(\zeta\) used in Section \ref{subs:simmetrics}. 
\begin{table}
\centering
\begin{tabular}{c|c}
    \textbf{Name} & \textbf{Description}\\
    \hline 
     \(N_1\) & Number of grid points in \(v\) \\
     \(N_2\) & Number of grid points in \(x\) \\
     \(L_v\) & Truncation point in the velocity domain\\
     \(\alpha_1\) & Concentration of points around \(\alpha_2 = 0\)\\
     \emph{Rel. Tol.} & Relative Tolerance of the underlying ODE solver\\
     \emph{Abs. Tol.} & Absolute Tolerance of the underlying ODE solver\\
\end{tabular}
\caption{ Free parameters in the numerical scheme, collectively named \(\zeta\).}
\label{tab:schemeparameters}
\end{table}
The tolerances of the ODE solver (MATLAB's \emph{ode15s}) are set at \( 10^{-9}\). This value was chosen as a balance between computation time and accuracy. The remaining parameters depend on the initial condition and the magnitude of \(\sigma\) in \eqref{nl}. The number of grid points in velocity \(N_1\) must be sufficient to interpolate, differentiate, and integrate accurately the initial condition and the expected stationary state (a Gaussian centred at \(\pm1\)). As the velocity domain is the one with the tangent mapping and the truncation, the concentration parameter \(\alpha_1\) and the truncation point \(L_v\) are inherently linked to the choice of \(N_1\). Figure \ref{fig:NvsLPlots} shows the error incurred when integrating the stationary distribution \(\mu_-\) with \(\sigma = 0.5\) for a variety of values of \(N_1, L_v\) and \(\alpha_1\). From these plots \(L_v\) was set to be 8 and \(\alpha_1 = 8\). As we expect the velocity distributions to decay like a Gaussian, this value of \(L_v\) allows the domain to contain the vast majority of the mass. In general, for convenience we set \(\alpha_1 = L_v\). Figure \ref{fig:ConvergenceNTan} shows how the error changes with this choice of \(L_v\) and \(\alpha_1\) for increasing \(N_1\). The tangent mapping allows for at least an order of magnitude more accurate computation than the standard linear mapping, while not affecting computation time. From these plots an \(N_1\) value of 50 was chosen. This gives errors two orders of magnitude below that of the ODE solver, while maintaining a reasonable computation time. In the very low noise regime, reducing the truncation parameter \(L_v\) was preferred to increasing the number of grid points. Figure \ref{fig:DynamicNvsLPlots} shows the error in the final time step after evolving the dynamics \eqref{nl}. This shows that the parameters chosen from the static integration are still good choices for the dynamic problem. In all of these figures, the position distribution was uniform. For problems where the initial condition was not uniform, the number of points was chosen on an \emph{ad hoc} basis, ensuring the initial condition was well interpolated. If after evolving the dynamics the solution was not accurate (if for example there were points at which it was negative), \(N_2\) was increased and the simulation repeated; a value of \(N_2=50\) was sufficient for most simulations. 

\begin{figure}
    \centering
    \begin{minipage}{0.3\linewidth}
    \centering
    \includegraphics[width=\linewidth]{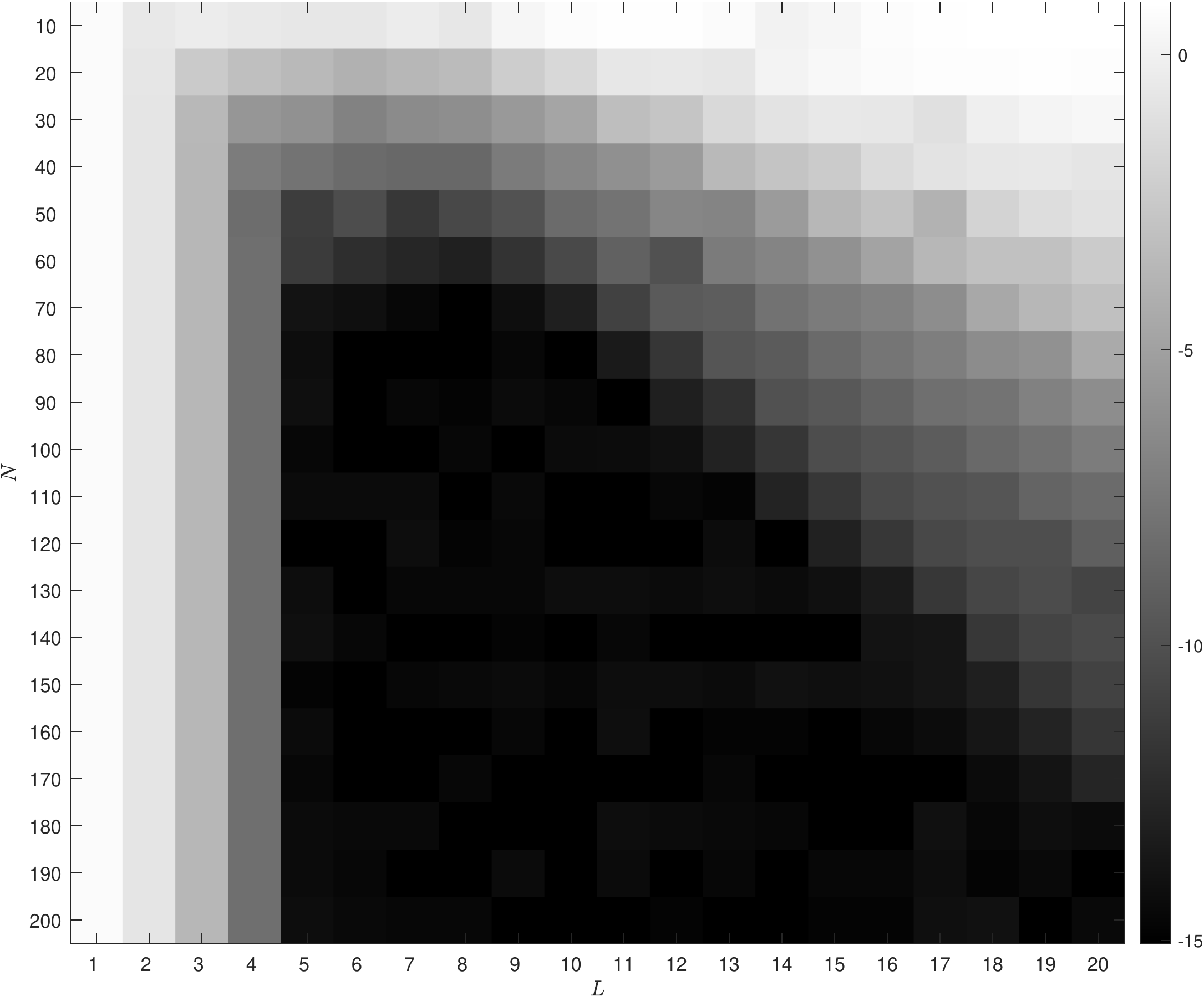}
    \subcaption{ }
    \end{minipage}%
    \begin{minipage}{0.3\linewidth}
    \centering
    \includegraphics[width=\linewidth]{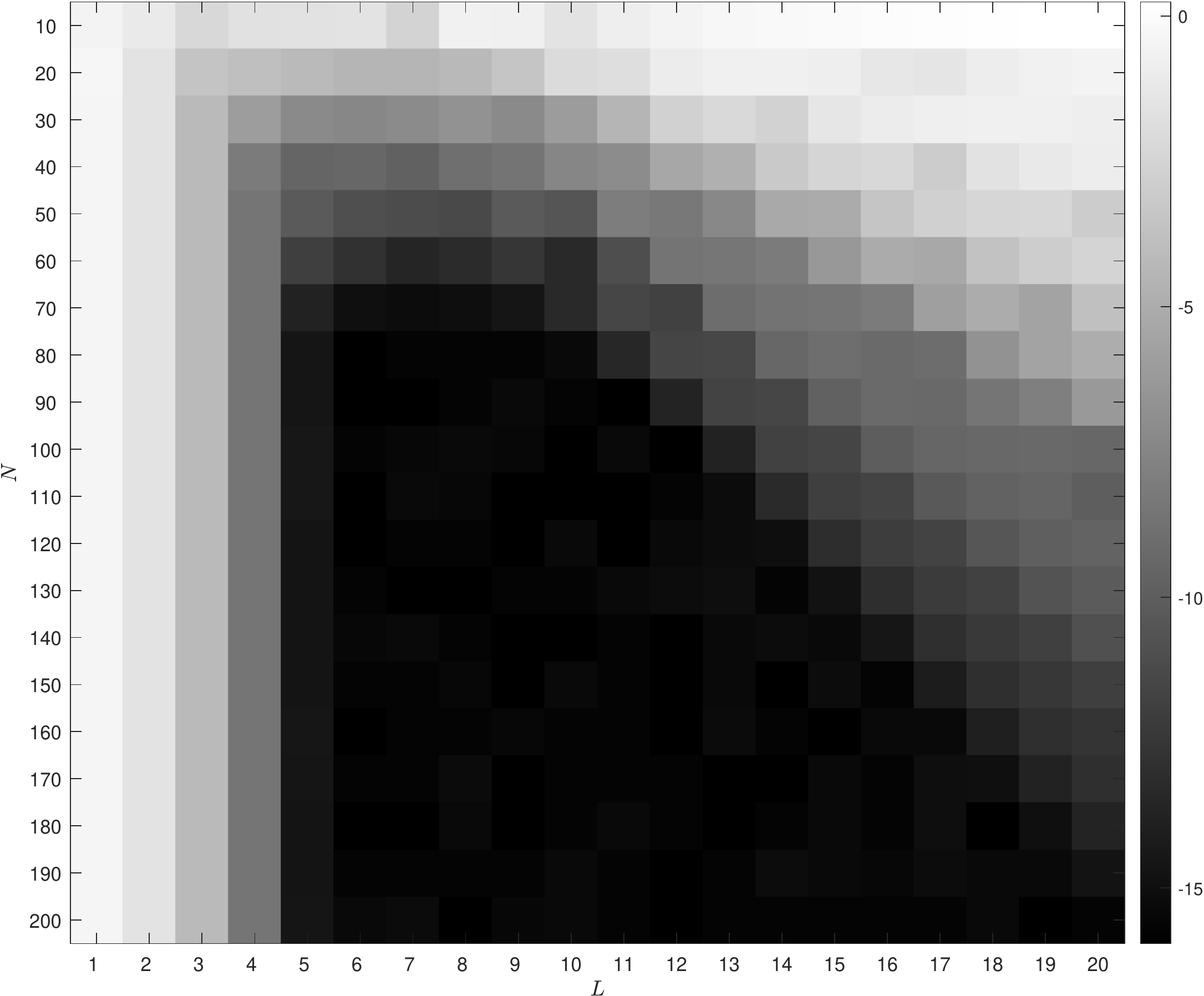}
    \subcaption{ }
    \end{minipage}%
    \begin{minipage}{0.3\linewidth}
    \includegraphics[width=\linewidth]{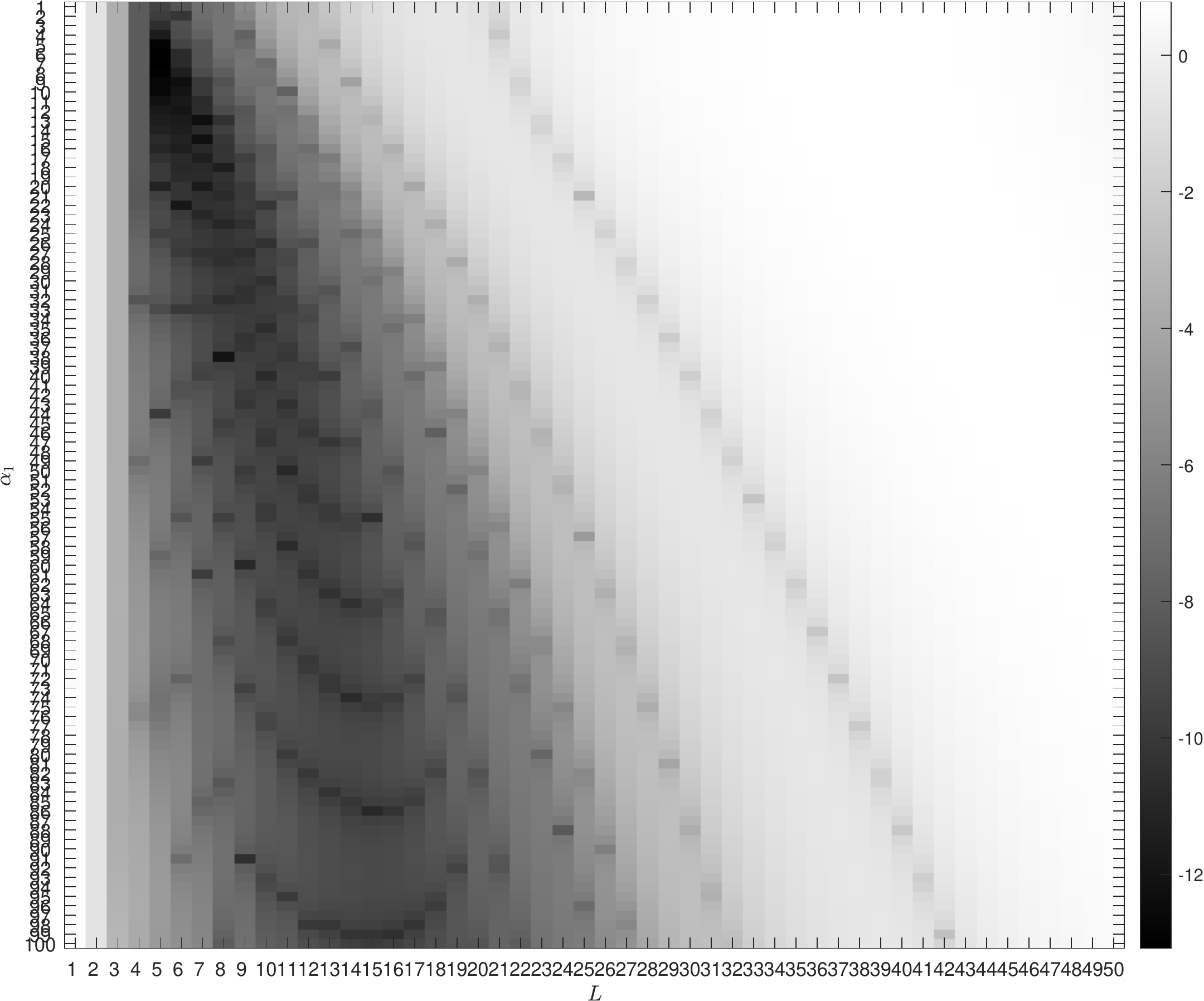}
    \subcaption{ } 
    \end{minipage}
    \caption{Behaviour of the base 10 logarithm of the relative error for the integral of a possible stationary solution of \eqref{nl}, namely \(f^{\zeta}_t(x,v) = \frac{1}{(2\pi)^{3/2}\sigma}\exp\left(-\frac{(v-\mu)^2}{2\sigma^2}\right)\) for \(\mu = -1,\sigma^2 = 0.5\) when discretising \(\mathbb{T}\times\R\) on a bounded domain using the tangent map as a function of \(L_v\) and points in the velocity domain \(N_1\). The number of points in the periodic domain is fixed at \(N_2 = 20\). The error ranges from \(10^0\) in white, to \(10^{-15}\) in black. Beyond a certain value of \(L_v\), the error begins to increase. This is because there are not sufficient points around \(-1\) to interpolate the Gaussian accurately. There is also a cutoff for low \(L_v\) where the discretised domain contains very little of the mass of the density. The left-most plot (a) shows the error in the mass, \(\left|\iint f_t(x,v) \dif x\dif v - 1\right|\) while the centre plot (b) show the error in the first      moment,\(\left|\frac{\iint v f^{\zeta}_t(x,v) \dif x \dif v}{\iint f^{\zeta}_t(x,v) \dif x \dif v} - \mu\right|\). We normalise by the mass in the system so as not to include the error incurred in (a). From these plots, a value \(N_1=50\) was chosen for further simulations. The right-most plot (c) varies \(\alpha_1\) and \(L_v\), fixing \(N_1 = 50\). The concentration parameter \(\alpha_1\) is less sensitive than \(L\) and a range of values show very low errors. For ease, we fix \(\alpha_1 = L_v = 8\). }
    \label{fig:NvsLPlots}
\end{figure}

\begin{figure}
    \centering
    \begin{minipage}{0.47\linewidth}
    \centering
    \includegraphics[width=\linewidth]{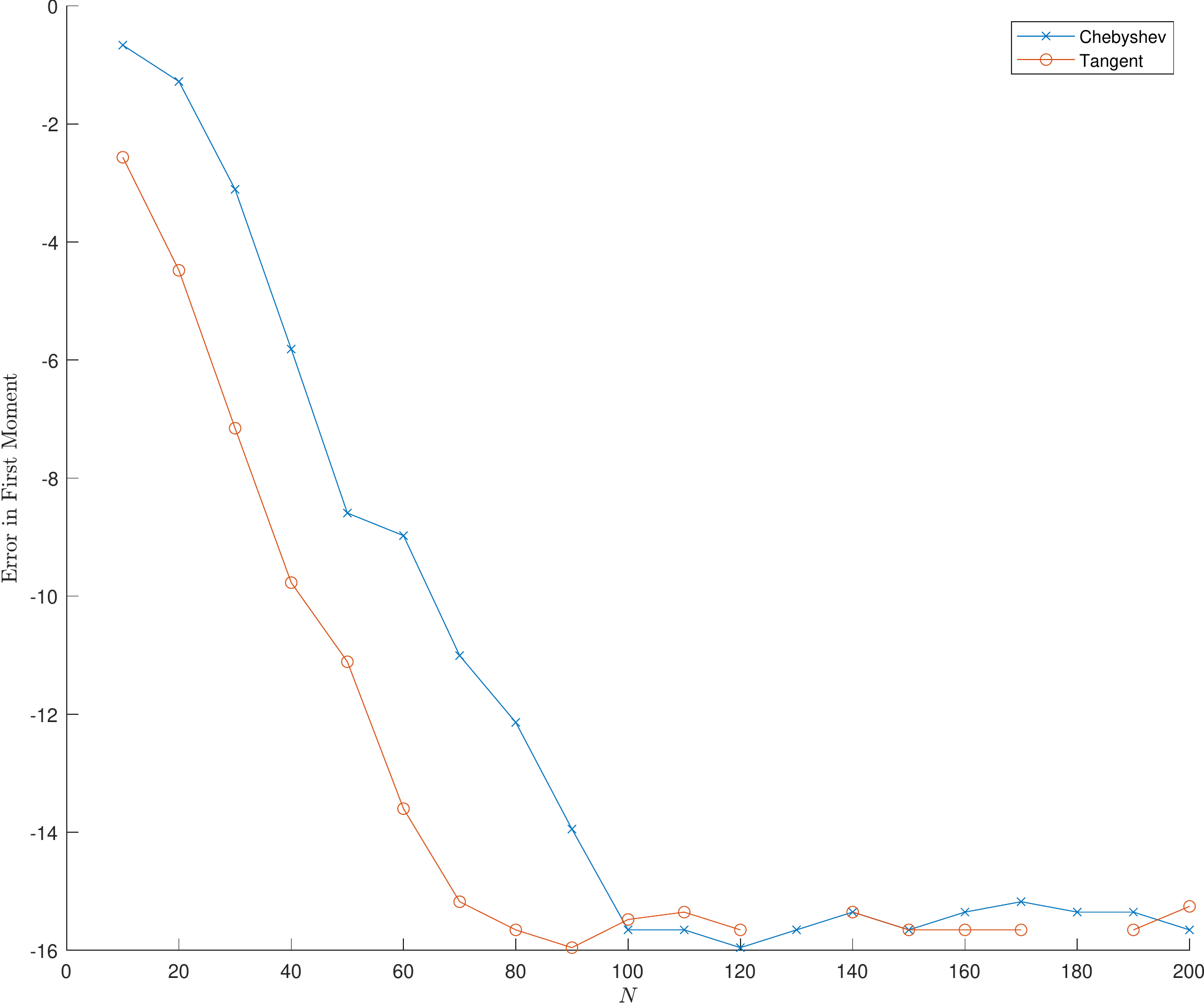}
    \subcaption{ }
    \end{minipage}%
    \begin{minipage}{0.47\linewidth}
    \centering
    \includegraphics[width=\linewidth]{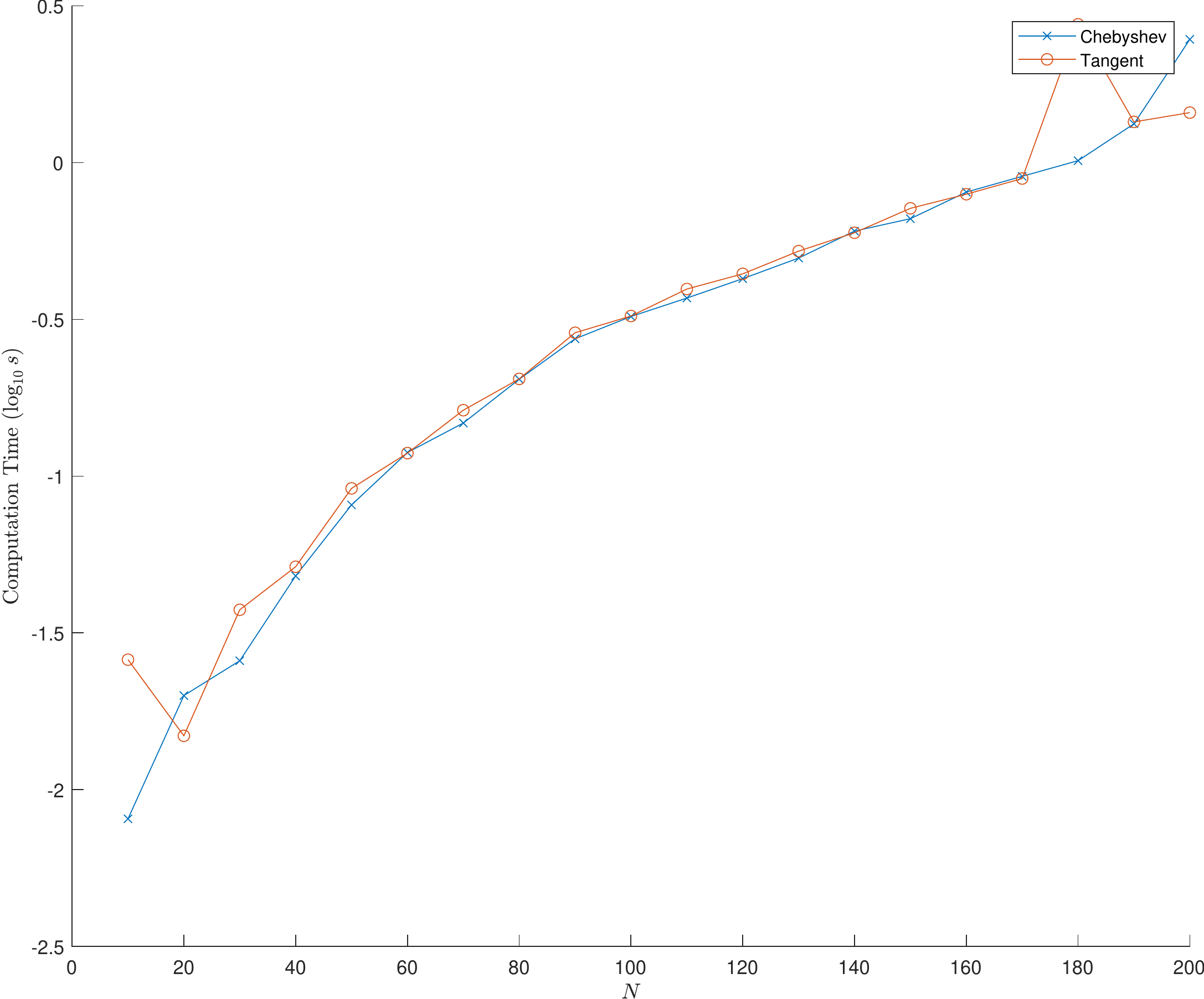}
    \subcaption{ }
    \end{minipage}
    \caption{A comparison of the standard linear mapping of the Gauss-Lobatto-Chebyshev points and the tangent mapping. Using the tangent mapping has little effect on computation time, scaling very similarly to the standard mapping (b). However, the left plot (a) shows that it is capable of decreasing the error by an order of magnitude for fixed \(N_1\) compared to the standard mapping. Here we plot the \(L_1\) norm of the error in the first moment when numerically integrating \(\mu_-\) with \(\sigma^2 = 0.5\), \(\left|\frac{\iint v f^{\zeta}_t(x,v) \dif x \dif v}{\iint f^{\zeta}_t(x,v) \dif x \dif v} - (-1)\right|\). Other parameters were \(\alpha_1 = 8, L_v = 8, N_2=40\). }
    \label{fig:ConvergenceNTan}
\end{figure}

\begin{figure}
    \centering
    \begin{minipage}{0.45\linewidth}
    \centering
    \includegraphics[width=\linewidth]{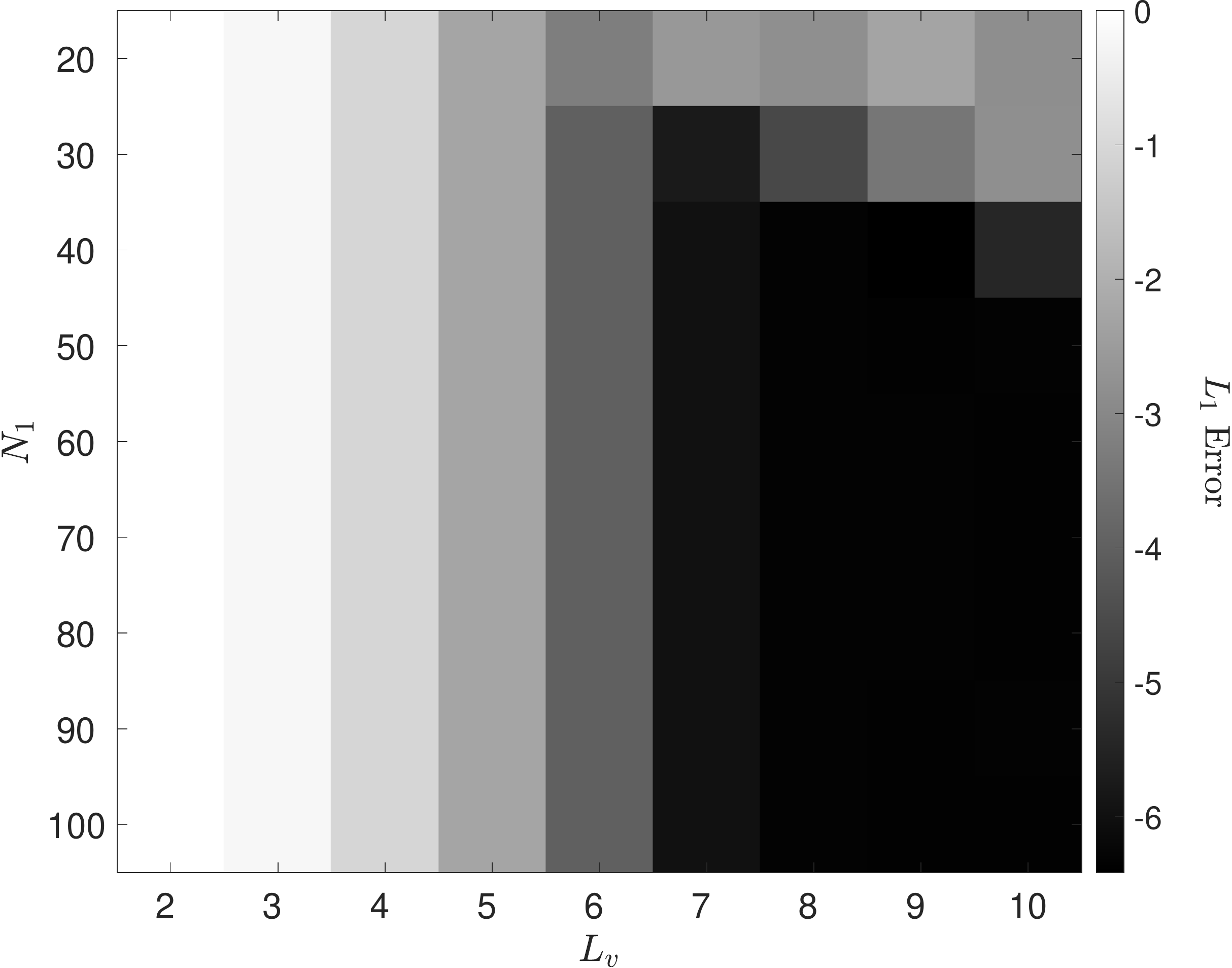}
    \subcaption{ }
    \end{minipage}%
    \begin{minipage}{0.45\linewidth}
    \centering
    \includegraphics[width=\linewidth]{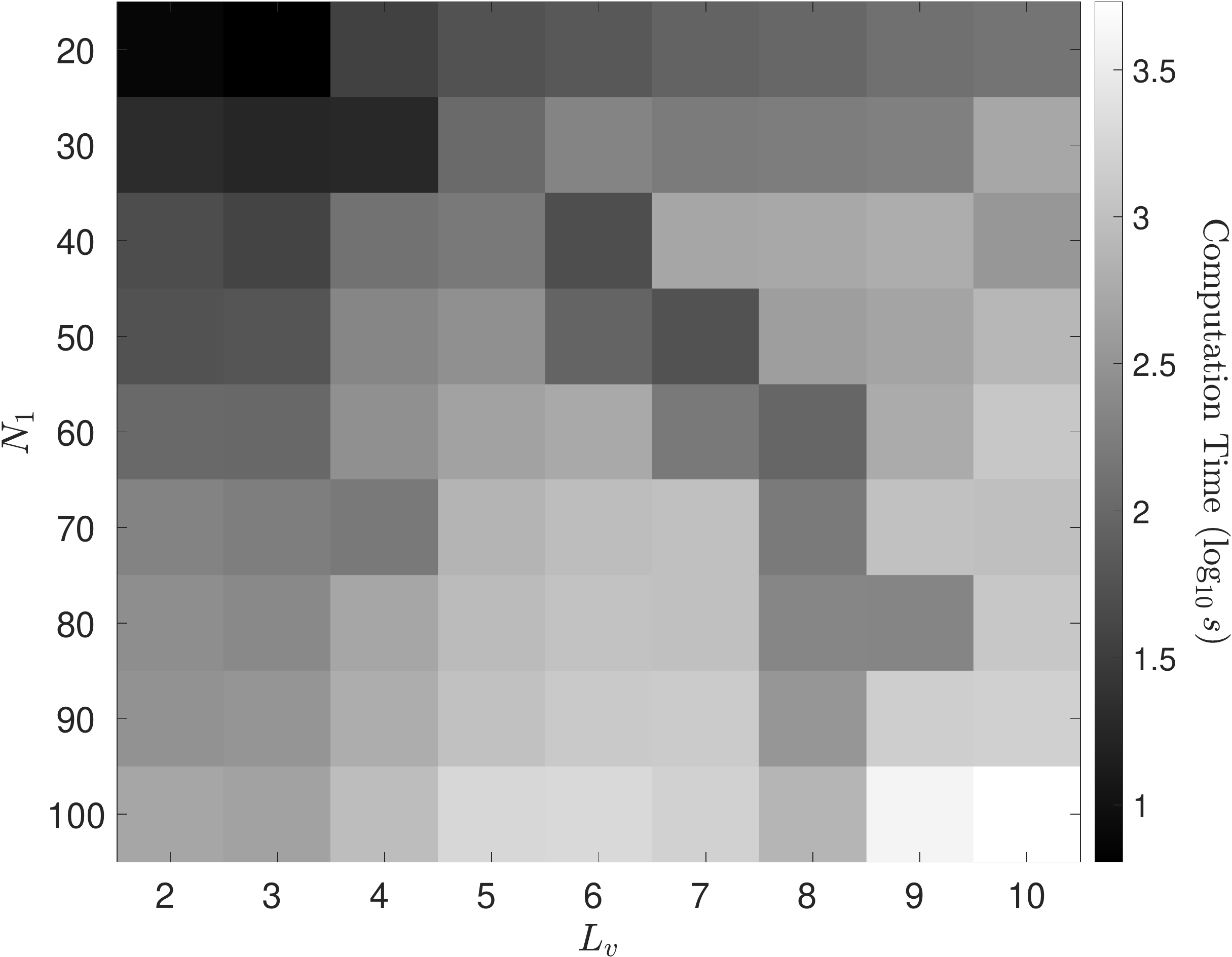}
    \subcaption{ }
    \end{minipage}
    \caption{The dynamics \eqref{nl} were simulated with IF \(\varphi\equiv1\), a smooth herding function \(G(u) = \frac{\atan(u)}{\atan(1)}\) and \(\sigma = 1\). The initial condition is uniform in position and Gaussian in velocity, with mean \(-0.5\) and variance \(0.5\). The number of position  points was fixed at \(N_2 = 40\) while the truncation point \(L_v\) was varied between 1 and 10. The number of velocity points \(N_1\) was varied from 20 to 100. The left plot (a) shows the difference between the simulated mean and the true mean of the stationary distribution. As the system is started out of stationarity, the maximum in the error was taken on the interval \(\lbrack 40,60 \rbrack\), after the moments have apparently converged. The right plot (b) shows the computation time for each \(N_1\) and \(L_v\) value. These plots confirm our earlier choice of \(L_v = 8, N=50\) -- increasing either parameter beyond this point will slow computation time without increasing accuracy a significant amount. The system was integrated using \emph{Abs. Tol.} \(=\) \emph{Rel. Tol.} \( = 10^{-9}\).  }
    \label{fig:DynamicNvsLPlots}
\end{figure}


\section*{Acknowledgements}
TMH was supported by The Maxwell Institute Graduate School in Analysis
and its Applications, a Centre for Doctoral Training funded by the EPSRC
(EP/L016508/01), the Scottish Funding Council, Heriot-Watt University and
the University of Edinburgh.
KJP acknowledges departmental funding through the `MIUR-Dipartimento di Eccellenza' programme.

\printbibliography
\end{document}